%% file: main.tex
\documentclass[11pt]{article}
\usepackage[margin=1in]{geometry}

\usepackage{amsmath, amssymb, amsthm}
\usepackage[mathic]{mathtools}
\usepackage{algorithmic}
\usepackage[T1]{fontenc}
\usepackage{bbm}
\usepackage{float}
\usepackage{subcaption}
\usepackage{booktabs}
\usepackage{libertine}
\usepackage{authblk}
\usepackage{comment}
\usepackage{multirow}
\usepackage{thm-restate}
\usepackage[libertine]{newtxmath} 
\usepackage[scaled=0.96]{zi4} 

\usepackage{a4wide, dsfont, microtype, xcolor, paralist, wrapfig}

\usepackage[ocgcolorlinks]{hyperref} 
\colorlet{DarkRed}{red!50!black}
\colorlet{DarkGreen}{green!50!black}
\colorlet{DarkBlue}{blue!50!black}
\hypersetup{
	linkcolor = DarkRed,
	citecolor = DarkGreen,
	urlcolor = DarkBlue,
	bookmarksnumbered = true,
	linktocpage = true
}

\input{macros.tex}

\title{Fairness in Influence Maximization through Randomization\footnote{An earlier version of this article has been published in the Thirty-Fifth AAAI Conference on Artificial Intelligence (AAAI 2021) with the same title~\cite{BeckerDGG21}.}}

\author[1]{Ruben Becker}
\author[1]{Gianlorenzo D'Angelo}
\author[1]{Sajjad Ghobadi}
\author[2]{Hugo Gilbert}
\affil[1]{\normalsize Gran Sasso Science Institute, 67100 L'Aquila, Italy}
\affil[2]{Université Paris-Dauphine, Université PSL, CNRS, LAMSADE, 75016 Paris, France}

\date{}
\begin{document}
\maketitle
\begin{abstract}
The influence maximization paradigm has been used by researchers in various fields in order to study how information spreads in social networks. While previously the attention was mostly on efficiency, more recently fairness issues have been taken into account in this scope. In the present paper, we propose to use randomization as a mean for achieving fairness. While this general idea is not new, it has not been applied in this area.

Similar to previous works like Fish et al.~(WWW '19) and Tsang et al.~(IJCAI '19), we study the maximin criterion for (group) fairness. In contrast to their work however, we model the problem in such a way that, when choosing the seed sets, probabilistic strategies are possible rather than only deterministic ones. We introduce two different variants of this probabilistic problem, one that entails probabilistic strategies over nodes (node-based problem) and a second one that entails probabilistic strategies over sets of nodes (set-based problem). After analyzing the relation between the two probabilistic problems, we show that, while the original deterministic maxmin problem was inapproximable, both probabilistic variants permit approximation algorithms that achieve a constant multiplicative factor of $1-1/e$ minus an additive arbitrarily small error that is due to the simulation of the information spread. For the node-based problem, the approximation is achieved by observing that a polynomial-sized linear program approximates the problem well. For the set-based problem, we show that a multiplicative-weight routine can yield the approximation result.

For an experimental study, we provide implementations of multiplicative-weight routines for both the set-based and the node-based problem and compare the achieved fairness values to existing methods. Maybe non-surprisingly, we show that the ex-ante values, i.e., minimum expected value of an individual (or group) to obtain the information, of the computed probabilistic strategies are significantly larger than the (ex-post) fairness values of previous methods. This indicates that studying fairness via randomization is a worthwhile path to follow. Interestingly and maybe more surprisingly, we observe that even the ex-post fairness values, i.e., fairness values of sets sampled according to the probabilistic strategies computed by our routines, dominate over the fairness achieved by previous methods on many of the instances tested.
\end{abstract}

 \input{introduction}
 \input{related_work}
 \input{preliminaries}
 \input{hardness}\input{algorithms}
 \input{experiments}
\input{conclusion}

\bibliography{references}
\end{document}

%% file: macros.tex
\DeclareMathOperator{\gr}{gr}

\DeclareMathOperator{\opt}{opt}

\DeclareMathOperator{\argmax}{argmax}

\DeclareMathOperator{\E}{\mathbb{E}}

\let\epsilon\varepsilon
\let\eps\varepsilon

\newcommand{\ones}{\mathds{1}}

\newcommand{\ZZ}{\ensuremath{\mathbb{Z}}}

\newcommand{\RR}{\ensuremath{\mathbb{R}}}

\newcommand{\CCC}{\mathcal{C}}

\newcommand{\SSS}{\mathcal{S}}
\newcommand{\XXX}{\mathcal{X}}
\newcommand{\LLL}{\mathcal{L}}
\newcommand{\PPP}{\mathcal{P}}
\newcommand{\OOO}{\mathcal{O}}
\newcommand{\QQQ}{\mathcal{Q}}

\usepackage{xspace}

\usepackage{tikz}
\usetikzlibrary{decorations.pathreplacing}
\usetikzlibrary{plotmarks}
\usetikzlibrary{positioning,automata,arrows}
\usetikzlibrary{shapes.geometric}
\usetikzlibrary{decorations.markings}
\usetikzlibrary{positioning, shapes, arrows}

\PassOptionsToPackage{usenames,dvipsnames,svgnames}{xcolor}
\definecolor{orange}{RGB}{235,90,0}
\definecolor{darkorange}{RGB}{175,30,0}
\definecolor{turkis}{RGB}{131,182,182}
\definecolor{darkturkis}{RGB}{31,82,82}
\definecolor{green}{RGB}{102,180,0}
\definecolor{darkgreen}{RGB}{51,90,0}
\definecolor{myblue}{RGB}{0,0,213}
\definecolor{mydarkblue}{RGB}{0,0,100}
\definecolor{mybrightblue}{HTML}{74B0E4}
\definecolor{mybrighterblue}{HTML}{B3EAFA}
\definecolor{lila}{RGB}{102,0,102}
\definecolor{darkred}{RGB}{139,0,0}
\definecolor{darkyellow}{RGB}{188,135,2}
\definecolor{brightgray}{RGB}{200,200,200}
\definecolor{darkgray}{RGB}{50,50,50}
\definecolor{amaranth}{rgb}{0.9, 0.17, 0.31}
\definecolor{alizarin}{rgb}{0.82, 0.1, 0.26}
\definecolor{amber}{rgb}{1.0, 0.75, 0.0}
\definecolor{green(ryb)}{rgb}{0.4, 0.69, 0.2}
\definecolor{hanblue}{rgb}{0.27, 0.42, 0.81}
\definecolor{grannysmithapple}{rgb}{0.66, 0.89, 0.63}

\newtheorem{theorem}{Theorem}[section]
\newtheorem{lemma}[theorem]{Lemma}

\newtheorem{observation}[theorem]{Observation}

\newtheorem{lemma-rstbl}[theorem]{Lemma}
\newtheorem{obs-rstbl}[theorem]{Observation}
\newtheorem{theorem-rstbl}[theorem]{Theorem}

\usepackage{environ}
\NewEnviron{cproblem}[1]{%
\begin{center}\fbox{\parbox{5.5in}{%
    {\centering\scshape #1\par}%
    \parskip=1ex
    \everypar{\hangindent=1em}%
    \BODY
}}\end{center}}

\newcommand{\nw}{\omega}

\DeclareMathOperator{\pof}{PoF}

%% file: introduction.tex
\section{Introduction}
The internet has revolutionized the way information spreads through the population.
One positive consequence is that important and valuable campaigns can be spread at little cost quite efficiently thanks to news platforms and social media.
Examples of such valuable campaigns are HIV prevention~\cite{wilder2018end,yadav2018bridging}, public health awareness~\cite{valente2007identifying} or financial inclusion~\cite{banerjee2013diffusion}. The information spreading process can be notably optimized by algorithms that identify key people in the network to act as seed users to initiate an efficient spread of the campaign.
The well known influence maximization problem formalizes this objective~\cite{kempe}: given a network and a probabilistic diffusion model, the task is to find a set of $k$ seed nodes from which the campaign will start to spread, in order to maximize the expected number of reached nodes. The problem has received a tremendous amount of attention~\cite{balancing,BorgsBCL14,borodin2017strategyproof,budak2011limiting,chen2017interplay,CohenDPW14,TangSX15,TangXS14}.

In the influence maximization problem, the objective is only concerned with the efficiency of the diffusion process, it does not take into account any fairness criteria. In order to underline the need of studying such fairness criteria in this scope, we start with the following motivating example:
Consider a simple random graph modeling a network similar to a core-periphery structure~\cite{BorgattiE00}. The network consists of two communities, the core $C$ and the periphery $D$. The probability of intra-community edges are $p_{C}$ and $p_{D}$ respectively, while the probability of inter-community edges is $q$. For concreteness, assume that $|C|=50$, $|D|=150$ and $p_{C}=0.5$, $p_{D}=0.1$ and $q=0.1$. I.e., we obtain a random network consisting of a well-connected rather small core and worse connected larger part of the graph that we refer to as the periphery of the network. Assume now that we use a state-of-the-art algorithm for influence maximization, e.g., the TIM implementation of the greedy algorithm due to Tang et al.~\cite{TangXS14}, in order to compute a seed set of size $k$ with large expected coverage in the graph. As we can see in Figure~\ref{fig:motivation-block-stochastic} on the left, this can lead to a significant discrepancy in the probability of nodes being reached in the two communities. For concreteness, if $k=5$, an average node in the core is reached with probability larger than $0.25$, while the average node in the periphery is reached only with a probability of around $0.05$.\footnote{In this example, we use the Independent Cascade model in order to model the information spread, edge weights are chosen to be $0.05$. We refer the reader to Section~\ref{sec : preliminaries} for the precise definition of the model used.} In the right plot in Figure~\ref{fig:motivation-block-stochastic}, we observe that the algorithm selects most of the seed nodes in the core. This is clearly because nodes in the core are better connected and thus choosing such nodes as seeds results in better coverage.
In summary, we observe that maximizing expected spread without considering any fairness criteria can lead to unfair coverage with respect to communities or groups in the network. Such observations have motivated researchers more recently, to take fairness issues in influence maximization into account. We continue by giving a brief review of the related literature.

\begin{figure}[ht]
    \centering
    \includegraphics[width=.49\linewidth]{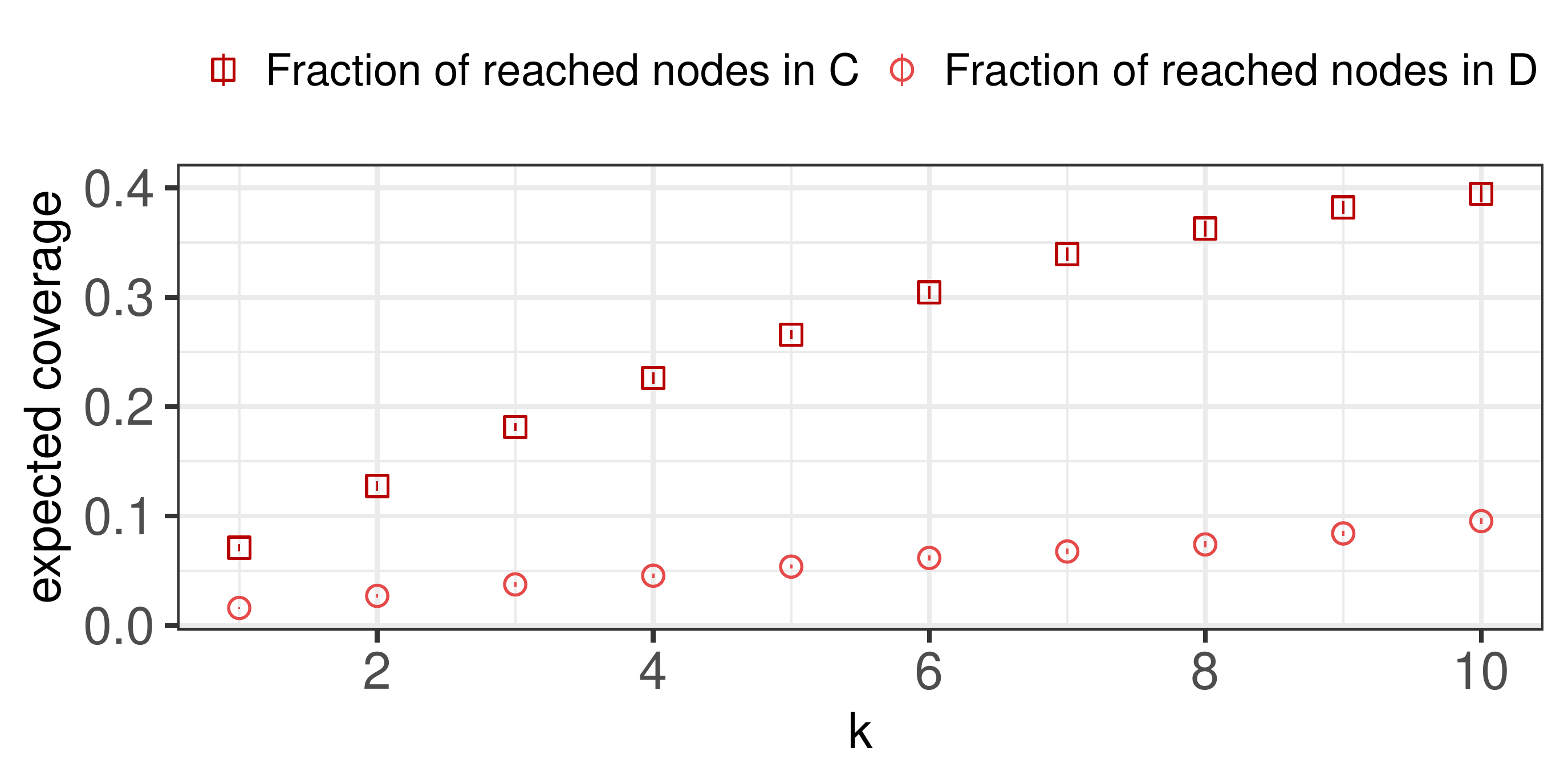}
    \includegraphics[width=.49\linewidth]{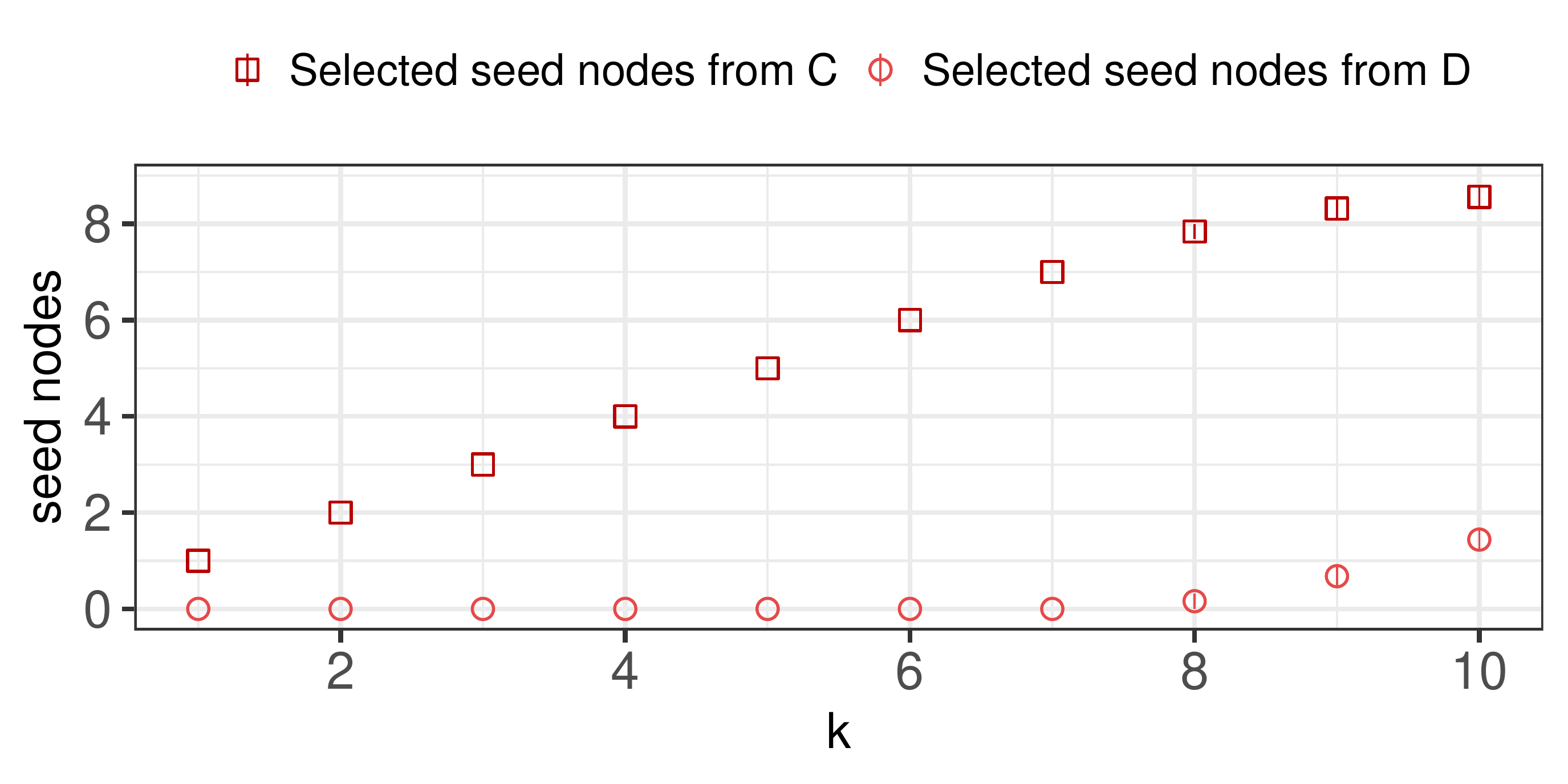}

    \caption{Results for the core-periphery model with a core of $50$ nodes and a periphery of $150$ nodes. The budget $k$ is increasing from $1$ to $10$.}%
    \label{fig:motivation-block-stochastic}
\end{figure}

A first sequence of papers has investigated a setting in which several competing players are paying the network's host to influence users in their favor.
The goal in these works is to ensure that the host picks seed nodes in a fair way with respect to the different players~\cite{chen2020maximizing,lu2013bang,yu2017fair}.
Another line of research has investigated the fairness of the diffusion process with respect to the vertices, i.e., the users in the network.
Indeed when only efficiency is being optimized, some users, or communities, i.e., groups of users, might get an unfairly low coverage~\cite{ali2019fairness,fish2019gaps,farnad2020unifying,khajehnejad2020adversarial,rahmattalabi2020fair,stoicaseeding,tsang2019group}.
An intuitive criterion to consider here is the maximin criterion. Here, the goal is to choose at most $k$ seed nodes to maximize the minimum probability of a user being reached. When generalized to groups of users or communities, the goal becomes to maximize the minimum expected fraction of users reached per community. The first problem has been considered by Fish et al.~\cite{fish2019gaps}, who showed that the problem is hard to approximate to any constant approximation factor, unless $P=NP$. The second problem has been considered by Tsang et al.~\cite{tsang2019group}. Building on previous work by \cite{ChekuriVZ10} and \cite{Udwani18}, the authors designed an algorithm with an asymptotic approximation ratio of $1-1/e$ provided that the number of communities is not much larger than $k$.

In the present paper, we extend these works by studying the impact of randomization on fairness.
Our approach is to allow randomized strategies for choosing the seed nodes rather than to restrict to deterministic strategies. 
Indeed, we introduce two randomized versions of the maximin problem. In our first problem, we consider randomized strategies that pick nodes as seeds independently. In contrast, in the second problem, we allow any probabilistic strategies that choose seed sets of expected size $k$, i.e., not restricting to independent distributions.
It is easy to envision that such randomized strategies provide certain advantages over deterministic ones. In fact, the use of randomization is a longstanding idea in computational social choice, where it often leads to more tractable results and more expressive solutions via for instance time-sharing mechanisms~\cite{david2013algorithmics}. It can also be used to incentivize participation~\cite{aziz2018rank} or to workaround impossibility results~\cite{brandl2016consistent}.
Lastly and closer to our work, using randomization is frequently used to obtain fairer solutions~\cite{aziz2013popular,bogomolnaia2001new,katta2006solution}.
Indeed, there may be optimization problems for which any deterministic solution is unfair. This was famously illustrated by Machina's mom example in which a mother should decide which of her two children will receive an indivisible treat~\cite{machina1989dynamic}.
In such cases, randomization may help evening things out by considering fairness in expectation, i.e., \textit{ex-ante fairness} rather than \textit{ex-post fairness}.
Randomization is both useful for one-shot and for repeated problems. In the former, it provides fairness over opportunities and in the latter it achieves fairness in the long run in a natural way. Lastly, randomization can be used to satisfy the fairness principle of \textit{equal treatment of equals}~\cite{moulin1991axioms}.
Despite being an old research topic, the study of randomized solutions is still a hot topic where many open problems remain to be solved~\cite{aziz2019probabilistic,brandt2019collective}.

\paragraph{Our Contribution.}
After recalling the necessary technical background related to influence maximization, in Section~\ref{sec : preliminaries}, we introduce the two randomized versions of the maximin problem. In the first one, we consider randomized strategies that pick nodes as seeds independently with some probability such that the expected size of the resulting seed set is bounded by $k$, we call this the \emph{node-based problem}. In the second problem, we study a more general feasible set. That is, we consider strategies that consist of probability distributions over seed sets of expected size $k$, i.e., not restricting to the special case of distributions that pick nodes independently but allowing for correlation.
We then analyze the relation between the two probabilistic problems.
In Section~\ref{sec: pof and hardness}, we quantify the loss in efficiency that can be incurred by following our fairness criteria, i.e., we show bounds on the price of fairness.
We continue by proving that both randomized variants of the maximin influence problems are NP-hard. For the node-based problem, we in addition show that, unless $P=NP$, there is no algorithm with approximation ratio better than $1-1/e$.  
Thereafter we show that still, in this setting of fairness in influence maximization, randomization leads to a number of advantages. In fact, in Section~\ref{sec: algorithms}, we prove that the resulting problems can be approximated to within a factor of $1-1/e$ (plus an additive $-\eps$ term that is also inherent in the work of Tsang et al.~\cite{tsang2019group}) even in the case when the number of communities exceeds the number of seed nodes $k$. For the node-based problem (up to the additive error term) we thus give a tight approximation result.
Furthermore, our work shows that the inapproximability result of Fish et al.~\cite{fish2019gaps} can be circumvent by introducing randomization to the problem. Our algorithms are comparatively simple.
For the node-based problem, the feasible set is of dimension $n$. After approximating (to within an additive $\pm\eps$ term) all functions $\sigma_v$ using concentration bounds, we still face the problem that the resulting optimization problem is not linear. We show however that the non-linear optimization problem is approximated to within a factor of $1-1/e$ by a linear program of the same size. Thus we obtain a polynomial time algorithm with multiplicative approximation ratio $1-1/e$ (plus the additive $-\eps$ term).
For the set-based problem, the situation is different. Here, by introducing a variable for every possible seed set, the problem can be approximated (to within an additive $\eps$ term) by a linear program. The downside of course is that this program is of dimension $\Theta(2^n)$. As the linear program is a covering linear program however, we are able to show that a multiplicative weights routine that is essentially a black-box application of a method by Young~\cite{Young95} can be used to obtain the described approximation. This method, as a subroutine, requires an algorithm for an oracle problem. We observe that the oracle problem in our case can be solved using standard (weighted) influence maximization and thus can be approximated to within a factor of $1-1/e$ efficiently both in theory and practice. Although the feasible set to the set-based problem is of exponential dimension, the computed solution that is guaranteed to be a multiplicative $1-1/e$ approximation (plus the additive $-\eps$ term) has only a linear support in $n$.
In Section~\ref{sec : numerical tests}, we evaluate implementations of multiplicative weight routines for both node and set-based problems on random instances, synthetic instances from the work of Tsang et al.~\cite{tsang2019group}, and a wide range of real world networks. We compare both the ex-ante and ex-post performance of our techniques with standard greedy techniques, as well as with the routines proposed by Tsang et al.~\cite{tsang2019group} and Fish et al.~\cite{fish2019gaps}. We observe that our ex-ante values are superior to the ex-post values of all other algorithms and, maybe surprisingly, our experiments indicate that even the ex-post values of our algorithms are competitive or even improve over the ex-post values achieved by the other techniques. We also experimentally evaluate the loss in efficiency, i.e., in total information spread resulting from using our algorithm over a standard IM algorithm that does not consider any fairness criteria. We conclude that our algorithms lead to much fairer solutions while incurring at most a small loss in total spread on all instances tested.


%% file: related_work.tex
\paragraph{Further Related Work.}
We review the works that have considered fairness issues in the context of influence maximization in more detail.

The line of research that investigates the fairness of the diffusion process with respect to the vertices (i.e., users) in the network is closest to our setting. Fish et al.~\cite{fish2019gaps}, to the best of our knowledge, are the first to study the maximin objective in order to maximize the minimum probability of nodes to be reached by the information spread. They show that this objective leads to an NP-hard optimization problem, and even more, is hard to approximate to within any constant factor unless $P=NP$. Even worse, the authors show that various greedy strategies have asymptotically worst-possible approximation ratios. 

In the work of Tsang et al.~\cite{tsang2019group}, the authors introduced the problem of maximizing the spread of a campaign while respecting a \emph{group}-fairness constraint. 
In their setting, each user of the network belongs to one or several communities and several criteria to guarantee that each community gets its fair share of information are considered.  
For each of these criteria, maximizing influence while respecting the related fairness constraint can be solved via a multi-objective submodular optimization problem. 
The authors design an algorithm to tackle such multi-objective submodular optimization problems that provides an asymptotic approximation guarantee of $1-1/e$. Their work cannot be directly extended to the case where fairness is considered with respect to individuals instead of communities. Indeed, their result requires that $m = o(k \log^3(k))$ where $m$ is the number of communities and $k$ is the seed set cardinality constraint. 

The above two works are the most closely related to the present paper. We proceed by reviewing more distant works that still treat fairness issues in influence maximization. 
Rahmattalabi et al.~\cite{rahmattalabi2020fair} further extend the group-fairness approach of Tsang et al.~\cite{tsang2019group} by following a different path. From the expected fraction of vertices reached in each community, the authors define a utility vector over the entire population of vertices, and then
take a welfare optimization approach by optimizing a decision criterion which is a function of this utility vector.  
Stoica, Han and Chaintreau~\cite{stoicaseeding} study how improving the diversity of nodes in the seed set can influence efficiency and fairness of the information diffusion process. 
In a rather specific setting, where the network is generated using a biased preferential attachment model yielding two unequal communities, the authors show that, under certain conditions, seeding strategies that take into account the diversity of nodes in the seed set are more efficient and equitable. 
Ali et al.~\cite{ali2019fairness} 
address fairness of the diffusion process with respect to different communities considering both the number of people influenced and the time step at which they are influenced. 
After illustrating that both, maximizing the expected number of nodes reached by choosing a seed set of fixed cardinality, and minimizing the number of seeds required to influence a given portion of the network may lead to unfair solutions, 
the authors propose an objective function which balances two objectives: the expected number of nodes reached which should be maximized, and the maximum disparity in influence between any two communities which should be minimized. 
More recently, Farnad, Babaki and Gendreau~\cite{farnad2020unifying} review the different notions of group-fairness criteria used in the influence maximization literature and show how influence maximization problems under these fairness criteria can be expressed as mixed integer linear programs. The authors then provide numerical tests to measure the price of fairness of different fairness criteria as well as the increase in fairness with respect to vanilla influence maximization.
Lastly, fair influence maximization was approached by Khajehnejad et al.~\cite{khajehnejad2020adversarial} based on machine learning techniques. The authors use an adversarial graph embedding approach to choose a seed set which both makes it possible to achieve high influence propagation and fairness between different communities.  

An even more distant line of work~\cite{chen2020maximizing,lu2013bang,yu2017fair} considers a setting in which several players compete with similar products on the same networks. 
These players pay a given budget to the network's host in order to influence as many users as possible. 
The question investigated is how to ensure that the host picks solutions in a fair way with respect to the different players. 
Stated otherwise, the average influence obtained per seed, also called amplification factor, should be similar for all players. We omit further details on this line of research as it is not too closely related to our work.

%% file: preliminaries.tex
\section{Preliminaries} \label{sec : preliminaries}

\subsection{Influence Maximization}\label{subsec:inflmax}
We consider the classical influence maximization setting where we are given a directed arc-weighted graph \(G=(V, A, w)\) with $V$ being the set of $n$ nodes, $A$ the set of arcs, and $w:A\rightarrow [0,1]$ an arc-weight function. In addition we are given an information diffusion model.
A broad variety of models can be used as information diffusion model. Two of the most popular models are the \emph{Independent Cascade} (IC) and \emph{Linear Threshold} (LT) models~\cite{kempe}. In both these models, given an initial node set \(S\subseteq V\) called \emph{seed nodes}, a spread of influence from the set \(S\) is defined as a randomly generated sequence of node sets $(S_t)_{t\in \mathbb{N}}$, where \(S_0=S\) and \(S_{t-1}\subseteq S_{t}\).\ These sets represent active users, i.e., we say that a node \(v\) is \emph{active} at time step \(t\) if \(v\in S_t\).\ The sequence converges as soon as \(S_{t^*}=S_{t^*+1}\), for some time step \(t^*\ge 0\) called the time of quiescence.
For a set \(S\), we use the standard notation \(\sigma(S) = \E[|S_{t^*}|]\) to denote the expected number of nodes activated at the time of quiescence when running the process with seed nodes \(S\), here the expectation is over the random process of information diffusion that depends on the weights $w$ and moreover on the information diffusion model at hand.
\paragraph{Information Diffusion Models.}
In the \emph{Independent Cascade} (IC) model, the values \(w_{a} \in [0,1]\) for \(a \in A\) are probabilities. The sequence of node sets $(S_t)_{t\in \mathbb{N}}$, is randomly generated as follows.  If \(u\) is active at time step \(t\ge 0\) but was not active at time step \(t-1\), i.e., \(u\in S_t\setminus S_{t-1}\) (with \(S_{-1}=\emptyset\)), node $u$ tries to activate each of its neighbors $v$, independently, and succeeds with probability \(w_{uv}\). In case of success, \(v\) becomes active at time step \(t+1\), i.e., \(v\in S_{t+1}\).

In the \emph{Linear Threshold} (LT) model, the values \(w_a \in [0,1]\) for \(a \in A\) are such that, for each node $v$, it holds that $\sum_{(u,v)\in A} w_{uv} \le 1$. The sequence of node sets $(S_t)_{t\in \mathbb{N}}$, is randomly generated as follows. At time step \(t+1\), every inactive node $v$ such that $\sum_{(u,v)\in A, u\in S_t} w_{uv} \ge \theta_v$ becomes active, i.e., \(v\in S_{t+1}\), where the thresholds $\theta_v$ are chosen independently and uniformly at random from the interval $[0, 1]$ for all nodes $v\in V$.

Both the IC and LT models can be generalized to what is known as the \emph{Triggering Model}, see~\cite[Proofs of Theorem 4.5 and 4.6]{kempe}.
For a node \(v\in V\), let \(N_v\) denote all in-neighbors of \(v\). In the Triggering model, every node \(v\in V\) independently picks a \emph{triggering set} \(T_v \subseteq N_v\) according to some distribution over subsets of its in-neighbors.
For a possible outcome \(L = (T_v)_{v\in V}\) of triggering sets for the nodes in \(V\), let \(G_L = (V,A_L)\) denote the sub-graph of \(G\) where \(A_L = \{(u,v)|v \in V, u\in T_v \}\). The graph $G_L$ is frequently referred to as live-edge graph and the edges $A_L$ are referred to as live edges. We denote with $\LLL$ the random variable that describes this process of generating outcomes or live-edge graphs, and with $L$ we mean a possible outcome, i.e., value taken by $\LLL$. We let \(\rho_\LLL(S)\) be the set of nodes reachable from \(S\) in \(G_\LLL\), then \(\sigma(S)=\E_\LLL[|\rho_\LLL(S)|]\).
The IC model is obtained from the Triggering model, if for each arc $(u,v)$, node $u$ is added to $T_v$ with probability $w_{uv}$. Differently, the LT model is obtained if each node $v$ picks at most one of its in-neighbors to be in her triggering set, selecting a node $u$ with probability $w_{uv}$ and selecting no one with probability $1 - \sum_{u \in N_v} w_{uv}$.

\paragraph{Further Notation.}
In what follows, we assume the triggering model to be the underlying model describing the information spread.\ We define \(\sigma_v(S):=\Pr_\LLL[v \in \rho_\LLL(S)]\) to be the probability that node $v$ is reached from seed nodes \(S\).\ Clearly, \(\sigma(S)=\E_\LLL[|\rho_\LLL(S)|] = \sum_{v\in V}\Pr_\LLL[v \in \rho_\LLL(S)] = \sum_{v\in V} \sigma_v(S)\).
We extend this notation in a natural way, that is, for $C\subseteq V$, we denote by $\sigma_C(S) = \frac{1}{|C|}\cdot\sum_{v\in C} \sigma_v(S)$ the average probability of being reached of nodes in $C$. Note that $\sigma_v(S)=\sigma_{\{v\}}(S)$ for nodes $v\in V$ and $\sigma(S) = |V|\cdot \sigma_V(S)$.

We use $\ones$ for the all-ones vector (of suitable dimension) and $\ones_i$ for the $i$'th unit vector. Furthermore, with a slight abuse of notation, we use $\ones_P$ to be the indicator function that equals 1 if $P$ holds and 0 otherwise.

For a maximization problem $\max\{F(x):x\in R\}$ with feasibility region $R$, objective function $F:R\rightarrow \RR_{\ge 0}$ and optimum value $\opt$. We say that $x$ is an $(\alpha, \beta)$-approximation for real values $\alpha\in (0,1]$ and $\beta\in[0, \infty)$, if $F(x)\ge \alpha \cdot \opt - \beta$.

\paragraph{Maximin optimization.}
The standard objective studied in influence maximization is finding a set $S$ maximizing \(\sigma(S)\) under a cardinality constraint $|S|\le k$ for some integer $k$. As this objective function does not take into account the fairness of the diffusion process with respect to nodes or communities, Fish et al.~\cite{fish2019gaps} and Tsang et al.~\cite{tsang2019group}, have investigated maximin variants of this objective that can be written as
\begin{align*}
\label{eq:maximin definition}
    \max_{S \in \binom{V}{k}} \min_{C\in \CCC} \sigma_{C}(S),
\end{align*}
where $\CCC$ is a set of $m\ge 1$ different communities $\emptyset \neq C\subseteq V$ that may not be disjoint and $\binom{V}{k}$ denotes the set of subsets of $V$ of size $k$.
If each node is its own community, this amounts to finding a set of $k$ seed nodes maximizing the minimum probability that a node is reached, which is the problem considered by Fish et al.~\cite{fish2019gaps}. We note that this is actually one instance of a broader class of optimization problems that ask to maximize a social welfare function, being the $-\infty$-mean here.
Fish et al.~\cite{fish2019gaps} considered the special case where the diffusion model is the Independent Cascade model and in which all arcs have the same probability of diffusion $\alpha$. They proved that the problem of choosing $k$ seeds $S$ such as to maximize $\min_{v \in V}\sigma_v(S)$ is NP-hard to be approximated within a factor better than $O(\alpha)$ and that minimizing the number of seeds to obtain the optimal solution cannot be approximated within a factor $O(\ln n)$. Furthermore, they analysed several natural heuristics which unfortunately exhibit worst-case approximation ratio exponentially small in~$n$.

\subsection{Fairness via Randomization} \label{subsec:probmaximin}
We initiate studying the impact of randomization to increase fairness for influence maximization.
We start with a simple example of an influence maximization problem to illustrate the impact of randomization.
Let us assume that we are using the IC model. Consider the graph in Figure~\ref{fig:pro-strategy} consisting of two nodes $u,v$, each forming their own community, connected in both directions by edges $(u,v), (v,u)$ with probabilities $1/2$. Assume that $k=1$. 
Then (due to symmetry) the
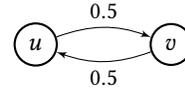
\begin{wrapfigure}{r}{7.5cm}
	\centering
	 \begin{tikzpicture}[scale=0.9, transform shape]
		\tikzset{vertex/.style = {shape=circle,draw = black, thick, fill = white, minimum size = 4mm}
			}
		\tikzset{edge/.style = {->,> = latex'}}

		\node[vertex] (1) at  (-1,0) {$u$};
		\node[vertex] (2) at  (1,0) {$v$};

        \draw[edge] (1) to[bend left=20, above] node {\small $0.5$} (2);

		\draw[edge] (2) to[bend left=20, below] node  {\small $0.5$} (1);
		\end{tikzpicture}%
		\caption{Simple instance showing that randomization allows to increase fairness in influence maximization.}
		\label{fig:pro-strategy}
\end{wrapfigure}
optimal deterministic strategy is to choose any of the two nodes achieving a minimum probability of being reached of $1/2$ for the non-chosen node. A probabilistic strategy however would be allowed to assign probabilities $1/2$ to both the sets $\{u\}$ and $\{v\}$. For each of the two nodes, this strategy achieves an expected probability of being reached of $1/2 + 1/4 = 3/4$, the $1/2$ being due to the fact that the node is a seed himself with probability $1/2$ and the $1/4$ being due to the probability of being reached (with probability $1/2$) from the other node if she is a seed (happens with probability $1/2$). While the example seems simplistic and artificial, it shows that the probabilistic strategy may in fact achieve a higher degree of fairness.
We consider two different ways of introducing randomness, either via distributions over sets or via distributions over nodes.

\paragraph{Probabilistically Choosing Sets.}
We relax the maximin problem by allowing for randomized strategies, i.e., feasible solutions in our \emph{set-based probabilistic maximin problem} are not simply sets of size at most $k$, but rather distributions over sets. Let $\PPP$ be the set of distributions over sets of expected size at most $k$, i.e., $\PPP:=\{p\in [0,1]^{2^{V}} :\ones^Tp =1, \sum_{S\subset V} p_S |S| \le k\}$ and let $S\sim p$ denote the random process of sampling $S$ according to the distribution $p$.
We now consider the optimization problem
\begin{align*}
    \opt_\PPP(G, \CCC, k) = \max_{p \in \PPP} \min_{C\in\CCC}\E_{S\sim p}[\sigma_{C}(S)],
\end{align*}
i.e., we are searching for the probability distribution that maximizes the minimum expected probability of the $m$ communities to be reached.
This notion 
is frequently referred to as ex-ante fairness in the literature~\cite{machina1989dynamic}.

We note that in the conference version~\cite{BeckerDGG21} of this article, we studied the problem where the probability distributions are restricted to be over sets of size \emph{exactly} $k$. Here, we explicitly allow sets of size different from $k$, the only restriction on the size is in expectation. This new problem constitutes a relaxation of the set-based problem studied in~\cite{BeckerDGG21}. We emphasize that all of our results hold for both versions of the problem. The main reason why we further relaxed the studied set-based problem is that this allows us to obtain a clean relationship (see the last paragraph of this section) between the set-based and the node-based problem that we introduce next.

\paragraph{Probabilistically Choosing Nodes.}
An alternative intuitive way of introducing randomness is obtained by considering a maximin problem where feasible solutions are not distributions over sets, but are characterized by probability values for nodes. In this setting, which we call the \emph{node-based probabilistic maximin problem}, we let $\XXX:=\{x\in [0,1]^n : \ones^Tx\le k\}$ be the feasible set and consider the process of randomly generating a set $S$ from $x$, denoted by $S\sim x$, by letting $i$ be in $S$ independently with probability $x_i$. In this setting we are thus interested in finding $x\in \XXX$ that maximizes the minimum expected coverage from $S$ of any community, when $S$ is generated from $x$ as described and the expectation is over this generation. We write this problem as
\[
    \opt_\XXX(G, \CCC, k) = \max_{x \in \XXX} \min_{C\in \CCC}\E_{S\sim x}[\sigma_{C}(S)].
\]


\paragraph{Extending Set Functions to Vectors.}
In what follows, we extend set functions to vectors in $\PPP$ and $\XXX$ in a straightforward way, i.e., for a set function $f$, for $p\in \PPP$, we let $f(p):=\E_{S\sim p} [f(S)]$ and, for $x\in \XXX$, we let $f(x):=\E_{S\sim x} [f(S)]$.

\paragraph{Relationship between Problems.}
We first observe that, for $x\in\XXX$, the vector $p^x$ defined as $p^x_S:=\prod_{i\in S} x_i \prod_{j\in V\setminus S} (1 - x_j)$, for $S\subseteq V$, is in $\PPP$ and furthermore $\sigma_C(x)=\sigma_C(p^x)$ for any \(C\in \CCC\). Hence, we obtain the following lemma.
\begin{lemma} \label{lemma:rel_set_node}
    For any $G$, $\CCC$, and $k$, it holds that
    $\opt_\XXX(G, \CCC, k) \le \opt_\PPP(G, \CCC, k)$.
\end{lemma}
We proceed by measuring the reverse relation. In fact, the concept of correlation gap can be used in order to upper bound $\opt_\PPP(G, \CCC, k)$ in terms of $\opt_\XXX(G, \CCC, k)$ incurring only a constant loss.
\begin{lemma}
    For any $G$, $\CCC$, and $k$, it holds that  $\opt_\PPP(G, \CCC, k) \le \frac{e}{e-1} \cdot\opt_\XXX(G, \CCC, k)$.
\end{lemma}
\begin{proof}
    Let $G$, $\CCC$, and $k$ be arbitrary.
    For a distribution $p\in\PPP$ over $2^V$, define the \emph{marginal probabilities $y^p$ w.r.t.\ $p$} by $y^p_i:= \Pr_{S\sim p}[i\in S] = \sum_{S\subseteq V: i\in S} p_S$. The correlation gap~\cite{AgrawalDSY10,Yan11} of $f:2^V\to \RR_{\ge 0}$ is defined as
    \[
        \gamma_f:=\sup_{p\in [0,1]^{2^V}} \frac{\E_{S\sim p} [f(S)]}{\E_{S\sim y^p}[f(S)]}
    \]
    and it is well-known that the correlation gap of a monotone submodular function is bounded from above by $\frac{e}{e-1}$, see~\cite[Corollary 1.2]{AgrawalDSY10} or \cite[Theorem 2.1]{Yan11}. We may thus conclude that, for all $C\in\CCC$, $\gamma_{\sigma_C}\le \frac{e}{e-1}$. Now, let $p\in \PPP$ be an optimal solution, i.e., $\opt_\PPP(G,\CCC, k) = \min_{C\in\CCC} \E_{S\sim p}[\sigma_C(S)]$. We obtain
    \begin{align*}
       \opt_\PPP(G,\CCC, k)
       &= \min_{C\in\CCC} \E_{S\sim p}[\sigma_C(S)]
       \le \min_{C\in\CCC} \Big \{\frac{e}{e-1} \cdot \E_{S\sim y^p}[\sigma_C(S)] \Big \}\\
       &= \frac{e}{e-1} \cdot \min_{C\in\CCC} \E_{S\sim y^p}[\sigma_C(S)]
       \le \frac{e}{e-1} \cdot \opt_\XXX(G,\CCC, k),
    \end{align*}
    where the last step uses that $\sum_{i\in V} y^p_i = \sum_{i\in V}\sum_{S\subseteq V: i\in S} p_S = \sum_{S\subseteq V} p_S \cdot |S| \le k$ and thus $y^p\in \XXX$.
\end{proof}
It remains to ask if the bound predicted by the above lemma is tight. We give the following simple example.
\begin{lemma}
    There exists a graph $G$, community structure $\CCC$, and integer $k$, such that $\opt_\PPP(G, \CCC, k) \ge \frac{5}{4} \cdot\opt_\XXX(G, \CCC, k)$ when using the IC model.
\end{lemma}
\begin{proof}
    Consider the graph $G$ consisting of two nodes $u$ and $v$ connected back and forth by two edges of weight $2/3$. Let $C$ be the singleton community structure, and $k=1$, i.e., the same instance as in Figure~\ref{fig:pro-strategy} with the difference that the edge weights are $2/3$. Then the best node-based solution achieves a value of $2/3$ (either by choosing one of the two nodes with probability $1$ or by choosing both with equal probability $1/2$). The optimal set-based solution that chooses the sets $\{u\}$ and $\{v\}$ both with probability $1/2$ however achieves a value of $1/2 + 1/2 \cdot 2/3 = 5/6$.
\end{proof}
We note that $e/(e-1)\approx 1.58$, while $5/4=1.25$. We consider tightening this gap to be an interesting open problem. 

\section{Price of Fairness and Hardness}
\label{sec: pof and hardness}

\subsection{Price of Group Fairness}
The price of group fairness is a quantitative loss measuring the decrease in efficiency that is incurred when we restrict ourselves to solutions respecting a group fairness requirement.
In the following, we denote the maximizing solutions to the node and general set-based problems by $F_{\XXX}(G, \CCC, k)= \argmax_{x \in \XXX} \min_{C\in \CCC}\E_{S\sim x}[\sigma_{C}(S)]$ and  $F_{\PPP}(G, \CCC, k)= \argmax_{p \in \PPP} \min_{C\in \CCC}\E_{S\sim p}[\sigma_{C}(S)]$, respectively. Then, the respective prices of fairness $\pof_{\XXX}(G, \CCC, k)$ and $\pof_{\PPP}(G, \CCC, k)$ incurred by restricting to strategies in $F_{\XXX}(G, \CCC, k)$ and $F_{\PPP}(G, \CCC, k)$ are given by
\[
    \pof_{\XXX}(G, \CCC, k)
    = \frac{\max_{S \in {\binom{V}{k}}} \sigma(S)}{\max_{x\in F_{\XXX}(G, \CCC, k)}\sigma(x)}
    \quad\text{ and }\quad
    \pof_{\PPP}(G, \CCC, k)
    = \frac{\max_{S \in {\binom{V}{k}}} \sigma(S)}{\max_{p\in F_{\PPP}(G, \CCC, k)}\sigma(p)}.
\]
We obtain that for both problems, the price of group fairness can be linear in the graph size.
\begin{restatable}{lemma-rstbl}{poflem}\label{lem:pof lower}
    For any even $n>0$, there is a graph $G$ with $n$ nodes and a community structure $\CCC$ such that $\pof_\XXX(G, \CCC, 1)=\pof_\PPP(G, \CCC, 1)=(n+2)/4$, when using the IC model.
\end{restatable}
\begin{proof}
	Let $G$ be composed of two disjoint sets $J$ and $I$ of $n/2$ vertices each. The only edges present in $G$ are the edges from one specific vertex $w\in J$ to all other vertices in $J$. Let the weight of these edges be $1$ and let $\CCC$ be the community structure consisting of singletons.
	Note that for all nodes $v\in I\cup \{w\}$, the probability of being reached is equal to the probability of being a seed as these nodes have no incoming edges. In other words, for any strategy $x\in \XXX$, it holds that $\sigma_v(x)=x_v$. Similarly, for any $p\in \PPP$, it holds that $\sigma_v(p) = y^p_v$, where $y^p_v:=\sum_{S:v\in S} p_S$ are the marginal probabilities with respect to $p$.
	Hence, it follows that the probabilistic solutions that maximize fairness split the budget $1$ equally among the nodes in  $I\cup \{w\}$. More precisely, when defining $\rho:=1/(\frac{n}{2} + 1)$, we get $F_\XXX(G,\CCC, 1) = \{\rho\cdot \ones_{I\cup \{w\}}\}$ and $F_\PPP(G,\CCC, 1) = \{p\in\PPP : y^p_v =\rho \text{ for all } v\in I\cup\{w\} \}$ and in both cases the achieved objective value is $\opt_\XXX(G,\CCC, 1) = \opt_\PPP(G, \CCC, 1) = \rho$. Furthermore
	\[
	    \max_{x\in F_{\XXX}(G,\CCC, 1)}\sigma(x) = \max_{p\in F_{\PPP}(G,\CCC, 1)}\sigma(p) = n \cdot \rho.
	\]
	The set $S$ of size 1 that maximizes the expected number of reached nodes however, selects the node $w$ yielding $\max_{S \in {\binom{V}{1}}}\sigma (S) = n/2$. Hence, we get a price of fairness that is of linear order. More precisely, $\pof_{\XXX}(G,\CCC, 1) = \pof_{\PPP}(G,\CCC, 1) = (\frac{n}{2} + 1)/2 = (n+2)/4$.
\end{proof}

On the positive side we obtain that the price of group fairness is never larger than $n/k$.
\begin{restatable}{lemma-rstbl}{pofuplem}\label{lem:pof up}
    For any graph $G$, community structure $\CCC$ and number $k$, it holds that $\pof_\XXX(G, \CCC, k)\le n/k$ and $\pof_\PPP(G, \CCC, k)\le n/k$.
\end{restatable}
\begin{proof}
    Note that, for both problems, there exist some optimal solution $x$ and $p$, such that the expected size of the seed set is exactly $k$, i.e., $\ones^Tx=k$ and $\sum_{S\subseteq V} p_S |S| = k$. Furthermore, the expected size of the spread ($\sigma(x)$ and $\sigma(p)$) is at least as large as the expected size of the seed set, i.e., for any $x\in \XXX$ and $p\in\PPP$, it holds that $\sigma(x)\ge k$ and $\sigma(p)\ge k$.
    Thus $\max_{p\in F_{\PPP}(G, \CCC, k)}\sigma(p)\ge k$ and $\max_{x\in F_{\XXX}(G, \CCC, k)}\sigma(x)\ge k$.
	Together with $\sigma(S)\le n$ for any set $S$, we obtain an upper bound of $n/k$ for the price of (group) fairness.
\end{proof}

\paragraph{Pessimistic Price of Fairness.}
A more pessimistic point of view leads to a different definition of the price of fairness. Indeed, the reader might have wondered if it is the correct choice to define $\pof_{\XXX}$ and $\pof_{\PPP}$ using the maximum spread over all fair solutions in the denominator. In fact, if we just compute any fair solution, the loss in terms of efficiency that we may incur could be as large as
\[
    \overline{\pof}_{\XXX}(G, \CCC, k):=
    \frac{\max_{S \in {\binom{V}{k}}} \sigma(S)}{\min_{x\in F_{\XXX}(G, \CCC, k)}\sigma(x)}
    \quad\text{and}\quad
    \overline{\pof}_{\PPP}(G, \CCC, k):=
    \frac{\max_{S \in {\binom{V}{k}}} \sigma(S)}{\min_{p\in F_{\SSS}(G, \CCC, k)}\sigma(p)},
\]
for the node-based problem and the set-based problem, respectively. We call this alternative definition the \emph{pessimistic price of fairness} for a graph $G$, community structure $\CCC$ and budget $k$.
We note that clearly $\overline{\pof}_{\XXX}(G, \CCC, k) \ge \pof_{\XXX}(G, \CCC, k)$ and $\overline{\pof}_{\PPP}(G, \CCC, k) \ge \pof_{\PPP}(G, \CCC, k)$. Moreover, we note that Lemma~\ref{lem:pof lower} holds still for this alternative definition as $\sigma(x)=\sigma(p)=n\cdot \rho$ for all $x\in F_\XXX(G,\CCC, 1)$ and $p\in F_\PPP(G,\CCC, 1)$.
In contrast, we observe that Lemma~\ref{lem:pof up} does not transfer to the pessimistic notion. In fact, we obtain only the following weaker lemma.
\begin{restatable}{lemma-rstbl}{barpoflempess}\label{lem:pof up pess}
    For any graph $G$, community structure $\CCC$ and number $k$, it holds that $\overline{\pof}_\XXX(G, \CCC, k)\le n$ and $\overline{\pof}_\PPP(G, \CCC, k)\le n$.
\end{restatable}
\begin{proof}
    Recall that all communities are non-empty. Fix an arbitrary community $C$. Now, consider $\sigma_C(x)$ as a function of $x_v$ for a node $v\in C$. This function is strictly monotonically increasing in $[0,1]$. And thus all fair solutions $x$ satisfy that the expected size of the seed set is at least $1$, i.e., $\ones^Tx\ge 1$. Similarly, it can be seen that all fair solutions $p$ to the set-based problem satisfy $\sum_{S\subseteq V} p_S |S| \ge 1$. Furthermore, the expected size of the spread $\sigma(x)$ and $\sigma(p)$ is larger than the expected size of the seed set, i.e., for any fair $x$ and $p$, it holds that $\sigma(x)\ge 1$ and $\sigma(p)\ge 1$.
    Thus $\min_{p\in F_{\PPP}(G, \CCC, k)}\sigma(p)\ge 1$ and $\min_{x\in F_{\XXX}(G, \CCC, k)}\sigma(x)\ge 1$.
	Together with $\sigma(S)\le n$ for any set $S$, we obtain an upper bound of $n$ for the price of (group) fairness.
\end{proof}
It turns out that the above bound is tight. Consider the following example. The graph $G$ consists of one isolated node $v$ and a (to $v$ unconnected) clique of $n-1$ nodes with edge probabilities being 1. Assume furthermore that the community structure $\CCC$ is such that the nodes in the clique do not participate in any community, while $v$ forms its own community. Furthermore, assume that the budget $k$ is $2$. The node-based solution $x$ that is zero everywhere but for $x_v=1$ and the set-based solution $p$ that is zero everywhere but for $p_{\{v\}}=1$ are optimal fair solutions as they achieve an objective value of 1. The deterministic solution $S=\{u,v\}$, where $u$ is an arbitrary node in the clique however satisfies $\sigma(S)=n$ and thus  $\overline{\pof}_\XXX(G, \CCC, 2)\ge n$ and  $\overline{\pof}_\PPP(G, \CCC, 2)\ge n$.
Finally, we remark that the above example heavily depends on the fact that it is not necessary to use the whole budget $k$ in order to obtain an optimal fair node-based or set-based solution.

%% file: hardness.tex
\subsection{Hardness}
Fish et al.~\cite{fish2019gaps} show that the standard maximin problem as introduced in Section~\ref{sec : preliminaries}
is NP-hard and even inapproximable.
In this section, we provide hardness results for our set-based and node-based probabilistic maximin problems that we introduced in Section~\ref{subsec:probmaximin}. In the first paragraph of this subsection, we prove that both problems are NP-hard. In the second paragraph, we even show that the node-based probabilistic maximin problem cannot be approximated to within $1-1/e + \eps$ for any $\eps>0$ unless $P=NP$. We note that, although it shows a stronger result, our reduction in the second paragraph is significantly less involved than the NP-hardness result in the first paragraph. The reader may thus wonder why we still present the more involved proof of the weaker NP-hardness of the node-based problem. The reason for this choice is that we think that the first reduction gives further insight into how the hardness is implied from the fairness criteria. Note that the reduction in the second paragraph uses a single community and thus, in a certain sense, the hardness of approximation is inherent to maximizing influence spread rather than due to any fairness issue. 

\paragraph{NP-Hardness.}
The main result of this paragraph is the following theorem.
\begin{restatable}{theorem-rstbl}{hardness}
\label{thrm:hardness}
    For a directed arc-weighted graph $G = (V,E,w)$ 
    it is NP-hard to decide if there is $p \in \PPP$ with $\min_{v\in V} \E_{S\sim p}[\sigma_v(S)] \ge \alpha$ (resp.\ $x\in \XXX$ with $\min_{v\in V} \E_{S \sim x}[\sigma_v(S)] \ge \alpha$) for any $\alpha \in (0,1)$ even when using the IC model.
\end{restatable}
The proof of the theorem is based on a reduction from the \emph{vertex cover problem}. In the vertex cover problem, we are given a graph $G = (V,E)$ where $V$ is a set of $n$ vertices and $E$ is a set of $m$ edges, and an integer $k$. The task is to determine if there exists a set $T = \{v_{i_1},\ldots, v_{i_k}\}$ of $k$ vertices such that $\forall e\in E, e\cap T\neq \emptyset$. We proceed by describing the reduction.

Let $\alpha\in (0,1)$ be arbitrary.
Given an instance of the vertex cover problem, we create instances of the set-based and node-based probabilistic maximin problem defined as follows, see Figure~\ref{fig:hardness}. Let us call the resulting directed graph $\overline{G} = (U,A)$ and let us assume that the IC model is the underlying diffusion model and that we are considering the singleton community structure. The node set $U$ contains (1)~one vertex $u_{v}$ for each vertex $v \in V$, (2)~one auxiliary vertex $u_a$, (3)~a vertex $u_{e}$ for each edge $e \in E$, (4)~a set $I_e = \{u_{e^1}, \ldots, u_{e^{\lambda}}\}$ of $\lambda:=\lceil\frac{mk(k+1)}{\alpha(1-\alpha)^2}\rceil + 1$ vertices for every edge $e\in E$. The edge set $A$ is defined as follows: (1)~Each vertex $u_v$ has an outgoing edge towards $u_a$ labelled with probability 1, while the vertex $u_a$ has an outgoing edge with probability $\alpha$ towards each vertex $u_v$. (2)~There is an edge labelled with probability $1$ from $u_{v}$ to $u_{e}$ if $v \in e$. (3)~There are edges from $u_{e}$ to all vertices $u_{e^{i}}$ for each edge $e\in E$ labelled with probability $\alpha$. We set the budget for both set-based and node-based problems equal to $k$.
\begin{figure}[ht]
	\centering
	\begin{tikzpicture} [scale=0.7, transform shape]
		\tikzset{vertex/.style = {shape=circle,draw = black,
				thick,
				fill = white,
				minimum size = 9mm}
		}
		\tikzset{edge/.style = {-}}

		\node[vertex] (v0) at  (-2,1.5)   {$v$};
		\node[vertex] (w0) at  (-2,-1.5)  {$w$};
		\draw[edge] (v0) to[left] node[midway] {$e$} (w0);

		\tikzset{edge/.style = {->,> = latex'}}

		\node[vertex] (ua) at (1, 0) {$u_a$};

		\node[vertex] (v) at  (4,2)  {$u_v$};
		\node[vertex] (w) at  (4,-2) {$u_w$};
		\node[vertex] (e) at  (4, 0) {$u_e$};
		\draw[edge] (v) to[right] node[midway] {$1$} (e);
		\draw[edge] (w) to[right] node[midway] {$1$} (e);

		\draw[edge] (ua) to[bend right=20, above left] node[midway] {$\alpha$} (v);
		\draw[edge] (ua) to[bend left=20, below left] node[midway] {$\alpha$} (w);
		\draw[edge] (v) to[bend right=20, above left] node[midway] {$1$} (ua);
		\draw[edge] (w) to[bend left=20, below left] node[midway] {$1$} (ua);

		\node[vertex] (u1) at  (7,1.5)  {$u_{e^1}$};
		\node[vertex] (u2) at  (7,.5) {$u_{e^2}$};
		\node (d) at (7, -.3) {$\vdots$};
		\node[vertex] (ul) at  (7,-1.5)  {$u_{e^\lambda}$};

		\draw[edge] (e) to[above] node[midway] {$\alpha$} (u1);
		\draw[edge] (e) to[below] node[midway] {$\alpha$} (u2);
		\draw[edge] (e) to[below] node[midway] {$\alpha$} (ul);

		\path [draw = black, rounded corners, inner sep=100pt, dotted]
               (6.25, 2.2)
            -- (8.25, 2.2)
	        -- (8.25, -2.2)
            -- (6.25, -2.2)
            -- cycle;
        \node  (empty)    at (8, 1.8)  {$I_e$};
	\end{tikzpicture}%

	\caption{Illustration of the reduction: Scheme in $\overline G$ that is obtained for one edge $e=\{v,w\}$ in $G$. Note that there is just one node $u_a$ in $\overline G$, while there is a node $u_e$ and a set of nodes $I_e$ for each edge $e\in E$.
	}
	\label{fig:hardness}
\end{figure}
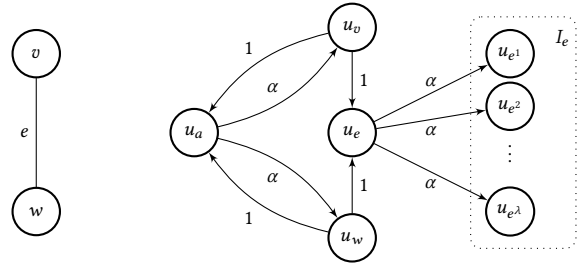
Our aim is now to show that there exists a vertex cover of size $k$ in $G$ if and only if there exists $p \in \PPP$ (resp.\ $x\in \XXX$) such that $\min_{u\in U} \E_{S\sim p}[\sigma_u(S)] \ge \alpha$ (resp.\ $\min_{u\in U} \E_{S \sim x}[\sigma_u(S)] \ge \alpha$) in $\overline G$. The following direction is immediate:
\begin{lemma}
    If there exists a vertex cover $\{v_{i_1},\ldots, v_{i_k}\}$ of size $k$ in $G$, then there exists $p \in \PPP$ (resp.\ $x\in \XXX$) such that $\min_{u\in U} \sigma_u(p) \ge \alpha$ (resp.\ $\min_{u\in U} \sigma_u(x) \ge \alpha$) in $\overline G$.
\end{lemma}
\begin{proof}
    Consider the set-based solution $p\in\PPP$ such that $p_S=1$ for $S=\{u_{v_{i_1}},\ldots, u_{v_{i_k}}\}$ and the node-based solution $x\in\XXX$ such that $x_u = 1$ if $u\in\{u_{v_{i_1}}, \ldots, u_{v_{i_k}}\}$. Then, clearly $\min_{u\in U} \E_{S\sim p}[\sigma_u(S)] \ge \alpha$ and $\min_{u\in U} \E_{S \sim x}[\sigma_u(S)] \ge \alpha$.
\end{proof}
It remains to argue the reverse direction. We first fix the following two observations.
\begin{observation}\label{obs: hard}
    \begin{enumerate}
    \item There exists an optimal solution $p^*$ (resp. $x^*$) such that, for any set $S\subseteq V$ with $(\{u_a\} \cup \{u_{e} : e\in E\}) \cap S \neq \emptyset$, it holds that $p_S=0$  (resp. $p^{x^*}_S= 0$, where $p^{x^*}_S:=\prod_{i\in S} x^*_i \prod_{j\in V\setminus S} (1 - x^*_j)$).
    \item For an edge $e=\{v,w\}\in E$, let $u_{e^{i}}\in I_e$ be a corresponding vertex and let $S\subseteq U$. Then
    (i) $\sigma_{u_{e^{i}}}(S)=1$ if $u_{e^{i}}\in S$, (ii) $\sigma_{u_{e^{i}}}(S)=\alpha$ if $u_{e^{i}}\notin S$ and $\{u_v, u_w\}\cap S \neq \emptyset$, and (iii) $\sigma_{u_{e^{i}}}(S) \le \alpha (1-(1-\alpha)^2)$ otherwise.
    \end{enumerate}
\end{observation}
\begin{proof}
    \begin{enumerate}
        \item Assume that $p_S>0$ (resp. $p^x_S>0$) for some set $S\subseteq V$ with $(\{u_a\} \cup \{u_{e} : e\in E\}) \cap S \neq \emptyset$.
        Let $u$ be a node in the intersection and let $S'$ be the set $S$ with $u$ replaced by one of its in-neighbors, call it $u'$. Then, clearly $p':= p - p_S \ones_S + p_S \ones_{S'}$ (resp.\ $x':= x- x_u \ones_u+ x_u \ones_{u'}$) satisfies $\sigma_u(p')\ge\sigma_u(p)$ (resp. $\sigma_u(x')\ge\sigma_u(x)$) for all vertices $u \in U$. It follows that $p$ (resp. $x$) can be transformed into a distribution (resp. vector) satisfying $p_S=0$ (resp. $p^{x^*}_S=0$) for all $S\subseteq V$ with $(\{u_a\} \cup \{u_{e} : e\in E\}) \cap S \neq \emptyset$.
        \item The first two cases are trivially true. For the third case, note that in that case $u_{e^{i}}$ can be reached at most via $u_a$. The probability that at least one of $u_v, u_w$ gets reached from $u_a$ is at most $1 - (1-\alpha)^2$ and thus $\sigma_{u_{e^{i}}}(S)\le \alpha (1-(1-\alpha)^2)$. \qedhere
    \end{enumerate}
\end{proof}
We are now ready to prove the lemma showing the reverse direction.
\begin{lemma}
    If there exists no vertex cover $\{v_{i_1},\ldots, v_{i_k}\}$ of size $k$ in $G$, then for all optimal solutions $p^* \in \PPP$ and $x^*\in \XXX$ it holds that $\min_{u\in U} \sigma_u(x^*) \le \min_{u\in U} \sigma_u(p^*) < \alpha$ in $\overline G$.
\end{lemma}
\begin{proof}
    Following Observation~\ref{obs: hard}, let $p^*$ be such that, for any set $S\subseteq V$ with $(\{u_a\} \cup \{u_{e} : e\in E\}) \cap S \neq \emptyset$, it holds that $p^*_S=0$.
	Our goal is to show that there exists one vertex $u\in U$ such that $\sigma_u(p^*) < \alpha$.
	Let $e=\{v, w\}\in E$ be arbitrary and recall that $I_e$ is of size $\lambda$. Thus, there exists $i\in [\lambda]$ such that $\Pr_{S\sim p^*}[u_{e^i}\in S] \le k/\lambda$ as otherwise $\E_{S\sim p^*}[|S|]> k$. W.l.o.g.\ let us assume that $i=1$ for all edges $e\in E$.
	Observe that then, for all $e\in E$, it holds that
	\begin{align}\label{formula: sigma small}
	\begin{split}
	    \sigma_{u_{e^1}}(p^*)
	    &= \Pr_{S\sim p^*}[u_{e^1} \in S] + \E_{S\sim p^*} [\sigma_{u_{e^1}}(S) \; | \; u_{e^1}\notin S] \cdot \Pr_{S\sim p^*} [u_{e^1}\notin S]\\
	    &\le \frac{k}{\lambda} + \E_{S\sim p^*} [\sigma_{u_{e^1}}(S) \; | \; u_{e^1}\notin S].
	\end{split}
	\end{align}
    In order to upper bound the above expectation, we define $\rho_e:=\Pr_{S\sim p^*}[S\cap \{u_v, u_w\}=\emptyset]$. Using Observation~\ref{obs: hard}, we get
    \begin{align}\label{formula: cond exp}
	    \E_{S\sim p^*} [\sigma_{u_{e^1}}(S) \; | \; u_{e^i}\notin S]
	    \le \alpha (1-(1-\alpha)^2) \cdot \rho_e + \alpha \cdot (1 - \rho_e)
	    = \alpha - \rho_e\cdot \alpha (1-\alpha)^2.
	\end{align}
	In order to complete the proof it thus remains to find $e\in E$ for which $\rho_e$ can be bounded from below. We deduce this bound from a lower bound on the sum of all $\rho_e$'s. We have
	\begin{align*}
	   \sum_{e\in E} \rho_e
	   &\ge \sum_{e=(v,w)\in E} \E_{S\sim p^*}[\ones_{S\cap \{u_v, u_w\}=\emptyset} \; | \; |S|\le k] \cdot \Pr_{S\sim p^*}[|S|\le k]\\
	   &\ge \frac{\E_{S\sim p^*}[\sum_{e\in E}\ones_{S\cap \{u_v, u_w\}=\emptyset} \; | \; |S|\le k]}{k + 1}
	\end{align*}
	using that $\E_{S\sim p^*}[|S|]\le k$ and hence using Markov's inequality
	\[
	    \Pr_{S\sim p^*}[|S| \le k]
	    = 1 - \Pr_{S\sim p^*}[|S| \ge k+1]
	    \ge 1 - \frac{k}{k+1}
	    = \frac{1}{k+1}.
	\]
	We now use that there exists no vertex cover of size $k$ in $G$ and thus for each $S$ with $|S|\le k$, there exists $e\in E$ such that $S\cap \{u_v, u_w\}=\emptyset$. Hence, we get $\sum_{e\in E} \rho_e \ge 1/(k+1)$. This also implies that there exists one edge $\bar e\in E$ for which $\rho_{\bar e}\ge 1/(m(k+1))$. Plugging this into~\eqref{formula: cond exp} and this again into~\eqref{formula: sigma small} gives
	\[
	    \sigma_{u_{\bar e^1}}(p^*)
	    \le \frac{k}{\lambda} + \alpha - \frac{\alpha(1-\alpha)^2}{m(k+1)}
	    < \alpha,
	\]
	using that $\lambda > \frac{mk(k+1)}{\alpha(1-\alpha)^2}$ by definition.
\end{proof}
The above two lemmata directly prove Theorem~\ref{thrm:hardness}.

\paragraph{Hardness of Approximation for Node-Based Problem.}
We continue by proving an even stronger result for  the node-based probabilistic maximin problem.
\begin{theorem}
    For any $\eps>0$, the node-based probabilistic maximin problem, cannot be approximated to within a factor $1-1/e+\epsilon$ unless $P=NP$.
\end{theorem}
\begin{proof}
    The proof is by reduction from the \textsc{Max-Coverage} problem. An instance of \textsc{Max-Coverage} is given by an integer $\kappa$ and a collection of subsets $D=\{S_1,\ldots, S_\mu\}$ of a universe of elements $U=\{e_1,\ldots, e_\nu\}$. The task is to find a subset $T$ of at most $\kappa$ subsets from $D$ such that the number of elements in their union $|\bigcup_{S\in T}S|$ is maximized. Recall that, for any $\eps>0$, \textsc{Max-Coverage} cannot be approximated to within $1-1/e+\eps$, unless $P=NP$~\cite[Theorem 5.3]{Feige98}.
    
    We construct the following instance of the node-based probabilistic maximin problem. The graph $G=(V, E)$ has an nodeset $V$ that consists of two sets of nodes, the first being $L:=\{u^{1}, \ldots, u^{\mu}\}$ and the second being $R:=\{v^{1}, \ldots, v^{\nu}\}$. There is an edge from node $u^i$ to node $v^j$, if and only if $e_j\in S_i$. The weight of all  these edges is 1. The community structure $\CCC$ consists of a single community $C$ that is equal to $R$. Now, w.l.o.g., we can assume that any node-based solution $x$ satisfies $x_v=0$ for all nodes $v\in R$. Suppose otherwise, then transferring $x_v$ to any of $v$'s ingoing neighbors can only increase the value of $\sigma_C(x)$.
    There is a one-to-one correspondence between the feasible sets to \textsc{Max-Coverage} and the integral solutions to the node-based probabilistic maximin problem as follows. A feasible set $T$ in \textsc{Max-Coverage} that covers $\ell$ elements implies an integral solution to the node-based probabilistic maximin problem with $x_{u^i}=1$ for all $S_i$ in $T$ and 0 otherwise that achieves a value of $\sigma_C(x) = \ell/\nu$. The other direction holds as well. 
    
    Now, let $\eps>0$ and assume that there exists an algorithm with approximation ratio $1-1/e+\eps$ for the node-based probabilistic maximin problem. Let $x\in \XXX$ be the output of that algorithm for the instance constructed above. Notice that the objective function in the node-based probabilistic maximin problem is equal to 
    \[
        f(x) 
        = \frac{1}{\nu} \sum_{v\in R} \sigma_v(x)
        = \frac{1}{\nu} \sum_{v\in R} 
            \Big(
                1 - \prod_{u\in\delta^{in}(v)}(1-x_u)
            \Big),
    \]
    where $\delta^{in}(v)$ denote all the ingoing neighbors of node $v$. It is clear that $f$ is $\epsilon$-convex in the sense of Ageev and Sviridenko~\cite{AgeevS04} and thus Pipeage rounding~\cite{AgeevS04} can be used in order to compute a vector $\ones_T$ from $x$ such that $f(\ones_T)\ge f(x)$. Denoting with $x^*$ and $T^*$ an optimal solution to the node-based probabilistic maximin problem and the \textsc{Max-Coverage} problem, respectively, this implies 
    \[
        \frac{\big|\bigcup_{S\in T} S\big|}{\nu}
        = f(\ones_T)
        \ge f(x)
        \ge \big(1 - \frac{1}{e} + \eps \big) \cdot f(x^*)
        \ge \big(1 - \frac{1}{e} + \eps \big) \cdot f(\ones_{T^*})
        = \big(1 - \frac{1}{e} + \eps \big) \cdot \frac{\big|\bigcup_{S\in T^*} S\big|}{\nu}.
    \]
    Hence, we obtain an algorithm for \textsc{Max-Coverage} with approximation ratio $1-1/e+\eps$ which is impossible unless $P=NP$. This completes the proof.
\end{proof}

%% file: algorithms.tex
\section{Approximation Algorithms} \label{sec: algorithms}
In this section, we show that there are algorithms that compute $(1-1/e,\eps)$-approximations to both the node-based and the set-based probabilistic maximin problems. We start with a standard step that allows us to approximate the functions $\sigma_C(p)$ and $\sigma_C(x)$ to within an additive error of $\eps$ for any $\eps>0$.

\paragraph{Approximation via Hoeffding's bound.}
The functions $\sigma_C(p)$ and $\sigma_C(x)$ involved in the optimization problems at hand are not computable exactly in polynomial time (even for a vector $p$ of polynomial support). Even worse, they are not relatively approximable using Chernoff bounds as there is no straightforward absolute lower bound on $\sigma_C(S)$ for sets $S$ of size $k$ and communities $C\in\CCC$.
Here, we will show that the functions can be absolutely approximated by functions $\tilde \sigma_C(p)$ and $\tilde \sigma_C(x)$ that are obtained by sampling a sufficiently large number of live-edge graphs.
Optimal solutions to the resulting maximin problems involving the approximate functions can thus be shown to be $(1,\eps)$-approximations to $\opt_\SSS(G,\CCC, k)$ and $\opt_\XXX(G,\CCC, k)$, respectively.

Formally, for $T\in\ZZ_{\ge 0}$, we let $L_1,\ldots, L_T$ denote a set of $T$ live-edge graphs sampled according to the Triggering model (that entails both the IC and LT model). Then, for $v \in V$ and $S\in 2^V$, we define
\[
    \tilde{\sigma}_v(S) := \frac{1}{T} \sum_{t=1}^{T} \ones_{v\in \rho _{L_{t}}(S)}.
\]

\begin{lemma}\label{lemma k-maximin sampling}
	Let $\delta\in (0,1/2)$ and $\eps\in(0,1)$.
	If
	$T \ge \eps^{-2}\cdot [n + \log n + \log\delta^{-1}]$,
	then, with probability at least $1-\delta$, we have that
	$
	|\tilde{\sigma}_v(S) - \sigma_{v}(S) |\leq \eps
	$
	holds for all $ v \in V $ and $ S\in 2^V$.
\end{lemma}
\begin{proof}
	We use Hoeffding's Bound, see for example Theorem~4.12 in the book by Mitzenmacher and Upfal~\cite{MitzenmacherU17}.
	Fix a node $v\in V$ and a set $S\in 2^V$.
	Note that the graphs $L_1, \ldots, L_T$ are sampled independently
	and that $\ones_{v\in \rho_{L_t}(S)}\in [0,1]$. Hence,
	$
	\Pr[|\tilde{\sigma}_{v}(S) - \sigma_v(S)| \geq \eps]
	\le 2e^{-2T\eps^2}
	\le \delta\cdot 2^{-n}/n
	$
	by the choice of $T$ and assuming that $n\ge 2$.
	Using a union bound over all $2^n$ sets $S\in 2^V$ and all $n$ nodes $v\in V$, we obtain that with probability at least $1-\delta$, we have
	$
	|\tilde{\sigma}_v(S) - \sigma_{v}(S) |\leq \eps
	$
	for all $v \in V$ and for all $S\in 2^V$.
\end{proof}
We now observe that the absolute $\eps$-approximations $\tilde \sigma_v(S)$ for all nodes $v\in V$ and sets $S\in 2^V$ imply also that $\tilde \sigma_v(p)=\E_{S\sim p}[\tilde\sigma_v(S)]$ is an absolute  $\eps$-approximation of $\sigma_v(p):=\E_{S\sim p}[\sigma_v(S)]$ for any $p\in\PPP$. The same holds true for $\tilde \sigma_v(x)=\E_{S\sim x}[\tilde\sigma_v(S)]$ for any $x\in\XXX$.

Furthermore, we get the same result for $\tilde \sigma_C(p):=\frac{1}{|C|}\sum_{v\in C} \tilde \sigma_v(p)$ for any $p\in\PPP$ and $C\in\CCC$ and for $\tilde \sigma_C(x):=\frac{1}{|C|}\sum_{v\in C} \tilde \sigma_v(x)$ for any $x\in\XXX$ and $C\in\CCC$ as these functions are again just averages over other absolute $\eps$-approximations. Hence we get the following lemma.

\begin{restatable}{lemma-rstbl}{hoeffdingapprox}
\label{lemma: opt approx}
    Let $\delta\in (0,1/2)$ and $\eps\in(0,1)$. Assume that $T \ge 4 \eps^{-2} \cdot [n + \log n + \log\delta^{-1}]$ and that $\tilde \sigma_C(\cdot)$ is as above. Let $p\in\PPP$ and $x\in\XXX$ be $(\alpha, \beta)$-approximations for
    $
        \max_{p\in\PPP}\min_{C\in\CCC} \tilde\sigma_C(p)
    $ and
    $
        \max_{x\in\XXX}\min_{C\in\CCC} \tilde\sigma_C(p)
    $, respectively. Then $p$ and $x$ are $(\alpha, \beta +\eps)$-approximations of $\opt_\PPP(G, \CCC, k)$ and $\opt_\XXX(G, \CCC, k)$ with probability at least $1-\delta$, respectively.
\end{restatable}
\begin{proof}
   	For any $q\in\PPP$, define $m(q):=\min_{C\in\CCC} \sigma_C(q)$ and $\tilde m(q):=\min_{C\in\CCC} \tilde \sigma_C(q)$. Then, according to Lemma~\ref{lemma k-maximin sampling} and the comments preceding this lemma, with probability at least $1-\delta$, it holds that $\tilde m(q)\in[m(q)-\eps/2, m(q)+\eps/2]$. Let $p^*$ and $\tilde p^*$ be maximizing solutions for $m$ and $\tilde m$, respectively. Then
   	\begin{align*}
       	m(p)
       	\ge \tilde m(p) - \frac{\eps}{2}
       	\ge \alpha \cdot\tilde m(\tilde p^*) -\beta - \frac{\eps}{2}
       	\ge \alpha \cdot\tilde m(p^*) - \big(\beta + \frac{\eps}{2}\big)
       	\ge \alpha \cdot m(p^*) - (\beta + \eps).
   	\end{align*}
   	The proof for the node-based problem is completely analogous.
\end{proof}

\subsection{Probabilistically Choosing Nodes}
\label{subsec: prob nodes}
We start with the node-based problem. It entails to solve the optimization problem
$
    \opt_\XXX(G,\CCC, k) := \max_{x \in \XXX} \min_{C \in \CCC}\sigma_C(x),
$
where $\sigma_C(x)=\E_{S\sim x}[\sigma_C(S)]$ for $C\in\CCC$ and $x\in \XXX:=\{x\in [0,1]^n : \ones^Tx\le k\}$. Recall that $S\sim x$ denotes the random process of independently letting $i\in V$ be in $S$ with probability $x_i$.
We use Lemma~\ref{lemma: opt approx} in order to approximate $\sigma_C(\cdot)$ and thus, in what follows, we focus on finding an approximation algorithm for the problem
$
    \max_{x\in\XXX}\min_{C\in\CCC} \tilde\sigma_C(x)
$.
We note that Theorem II.5 from~\cite{ChekuriVZ10} in combination with a binary search on a threshold can be used in order to get  a $(1-1/e, 0)$-approximation for this problem.
In what follows we give a more direct derivation of such an approximation.

\paragraph{Node-based Problem via LP.}
Note that the optimization problem at hand is not linear as, for given $x$, the probability to sample $S\in 2^V$ is equal to $\prod_{i\in S}x_i\prod_{i\notin S}(1-x_i)$. We will now argue however that the problem can be constantly approximated by an LP.

For a live-edge graph $L$ and a node $v\in V$, what is the probability of sampling a set $S$ that can reach $v$ in $L$, i.e., what is $q_v(L, x):=\Pr_{S\sim x}[v\in \rho_L(S)]$?
It is the opposite event of not sampling any node that can reach $v$ in $L$, hence $q_v(L, x) = 1 - \prod_{i\in V: v\in \rho_L(i)} (1 - x_i)$
and this is approximated by the function
$
  p_v(L, x) := \min\{ 1, \sum_{i\in V: v\in \rho_{L}(i)} x_i\}
$ as shown in the following observation.
\begin{restatable}{obs-rstbl}{constapprox}
\label{obs:const approx}
  For any live-edge graph $L$, node $v\in V$, and $x\in \XXX$, it holds that
  $
  q_v(L, x) \in[(1-\frac{1}{e}) \cdot p_v(L, x), p_v(L, x)].
  $
  \begin{proof}
  	We start with the lower bound. For simplicity, we let $\{i\in V: v\in \rho_L(i)\} =: \{1,\ldots, r\} = R$. Using the geometric-arithmetic mean inequality, we get
  	\begin{align*}
  	q_v(L, x)
  	&=   1 - \prod_{i\in R} (1 - x_i)
  	\ge  1 - \Big(\frac{1}{r}\sum_{i\in R}(1-x_i)\Big)^r
  	=    1 - \Big(1 - \frac{1}{r}\sum_{i\in R}x_i\Big)^r\\
  	&\ge \Big(1 - \Big(1 - \frac{1}{r}\Big)^{r}\Big) \cdot \min\Big\{1, \sum_{i\in R}x_i\Big\}
  	\ge  \Big(1 - \frac{1}{e}\Big) \cdot p_v(L, x),
  	\end{align*}
  	where the second to last inequality uses that $f(x)= 1 - (1 - x/r)^r$ is concave on the interval $[0,1]$.

  	We prove the upper bound by induction on $r$. Clearly if $r=1$, by the definition of $\XXX$, we have that $p_v(L, x) = \min\{1, x_1\} = x_1 = q_v(L, x)$. Let us show the statement for $r$, assuming that it holds for $r - 1$.
  	If $p_v(L, x) = 1$, the statement is obvious as $q_v(L, x)\le 1$. If $p_v(L, x) < 1$, we get
  	\begin{align*}
  	q_v(L, x)
  	=    1 - \prod_{i=1}^{r - 1} (1 - x_i) + x_{r} \cdot \prod_{i=1}^{r - 1} (1 - x_i)
  	\le  \sum_{i=1}^{r - 1}x_i + x_{r} \cdot \prod_{i=1}^{r - 1} (1 - x_i)
  	\le \sum_{i=1}^{r}x_i = p_v(L, x),
  	\end{align*}
  	where the first inequality uses the induction hypothesis and $\min\{1, \sum_{i=1}^{r - 1} x_i\} = \sum_{i=1}^{r - 1} x_i$ as $p_v(L, x) < 1$, while the second inequality uses that $\prod_{i=1}^{r - 1} (1 - x_i)\le 1$.
  \end{proof}
\end{restatable}
We now define $\lambda_v(x):=\frac{1}{T}\sum_{t=1}^T p_v(L_t, x)$ and analogously $\lambda_C(x):=\frac{1}{|C|}\sum_{v\in C} \lambda_v(x)$. As $p_v(L, x)$ provides an approximation for $q_v(L, x)$, we can show that the node-based maximin problem can be approximated by a solution to a maximin problem involving the functions $\lambda_C(x)$.

\begin{restatable}{lemma-rstbl}{lpapproxtau}
\label{lemma: LP approx tau}
  Let $x\in \XXX$ be an optimal solution to
  $
      \max_{x\in \XXX}\min_{C \in \CCC} \lambda_C(x)
  $,
  then $x$ is a $(1-1/e, 0)$-approximation to
  $
    \max_{x\in \XXX} \min_{C\in \CCC} \tilde \sigma_C(x)
  $.
\end{restatable}
\begin{proof}
  	Note that $\tilde \sigma_C(x)=\E_{S\sim x}[\tilde \sigma_C(S)]=\frac{1}{|C|}\sum_{v\in C} \tilde \sigma_v(x)$ for any $x\in\XXX$ and that furthermore $\tilde \sigma_v(x)=\frac{1}{T}\sum_{t=1}^T q_v(L_t, x)$. Hence the definition of $\lambda_C(x)$ and Observation~\ref{obs:const approx} yield the result.
\end{proof}
Together with Lemma~\ref{lemma: opt approx} and the fact that
$
    \max_{x\in \XXX}\min_{C \in \CCC} \lambda_C(x)
$
can be written as a linear program by introducing a threshold variable for the minimum, we get the following result.
\begin{restatable}{theorem-rstbl}{algorithmnodes}
    Let $\delta\in (0,1/2)$ and $\eps\in(0,1)$. There is a polynomial time algorithm that, with probability at least $1-\delta$, computes $x\in\XXX$ such that $\min_{C\in \CCC} \sigma_C(x) \ge (1-\frac{1}{e}) \opt_\XXX(G, \CCC, k) - \eps$.
\end{restatable}
\begin{proof}
	It remains to observe that an optimal solution to $\max_{x\in \XXX}\min_{C \in \CCC} \lambda_C(x)$ can be obtained by solving the following linear program of polynomial size:
	\begin{align*}
	\max \Big\{\tau:
	&\sum_{i=1}^{n} x_i \le k,\;
	y_{v t} \leq \sum_{i:v \in \rho_{L_t}(i)} x_i  \;\forall v \in V, t \in [T], \\
	&\sum_{t\in[T]} \sum_{v\in C}  y_{t v} \geq T|C|\cdot \tau
	\;\forall C \in\CCC, \\
	&x\in [0,1]^n, y_{v t} \in [0,1]\;\forall v \in V, t \in [T]
	\Big\}.
	\end{align*}
	By combining Lemma~\ref{lemma: LP approx tau} and Lemma~\ref{lemma: opt approx}, we get that, with probability at least $1-\delta$, the optimal solution to the above LP is a $(1-1/e,\eps)$-approximation to $\opt_\XXX(G, \CCC, k)$.
\end{proof}

\subsection{Probabilistically Choosing Sets}
\label{subsec: prob sets}

Recall the set-based probabilistic maximin problem
$
    \opt_\PPP(G,\CCC, k) := \max_{p \in \PPP} \min_{C \in \CCC}\sigma_C(p),
$
where $\sigma_C(p)=\E_{S\sim p}[\sigma_C(S)]$ for $C\in\CCC$ and $p\in \PPP:=\{p\in [0,1]^{2^V} :\ones^Tp =1, \sum_{S\subseteq V} p_S |S| \le k\}$.
In the light of Lemma~\ref{lemma: opt approx}, we focus on
finding approximate solutions to $\max_{p \in \PPP} \min_{C \in \CCC}\tilde \sigma_C(p)$.

Allowing for distributions over sets rather than sets turns the optimization problem at hand, $\max_{p \in \PPP} \min_{C \in \CCC}\tilde \sigma_C(p)$, into a problem that can be written as a linear program. While the original problem, i.e., the problem of choosing a set maximizing the approximate minimum probability, can be written as an integer linear program using a variable to model a threshold to be maximized. Hence, from an algorithmic point of view, one may think that this makes the problem polynomial time solvable. The caveat is of course that the dimension of $\PPP$ is large, namely $\Theta(2^n)$, which turns the dimension of the corresponding linear program super-polynomial.
In this section, we show that, nevertheless, the problem can be approximated to within a constant factor using a specific kind of linear programming algorithm. The essential observation is that the linear program at hand actually is a covering linear program. We will use a result due to Young~\cite{Young95} that shows that such linear programs can be solved efficiently independent of their dimension under the condition that a certain oracle problem can be solved efficiently. We proceed by introducing the result of Young.

\paragraph{Young's Algorithm.}
Young~\cite{Young95} gives algorithms for solving packing and covering linear programs. A covering problem in the sense of Young is of the following form:
Let $P\subseteq \RR^\nu$ be a convex set and let $f:P\rightarrow \RR^\mu$ be a $\mu$-dimensional linear function over $P$. Assume that $0\le f_j(x) \le \omega$ for all $j\in[\mu]$ and $x\in P$, where $\omega$ is the width of $P$ w.r.t.\ $f$. The \emph{covering problem} consists of computing
\(
  \lambda^* := \max_{x\in P}\min_{j\in [\mu]} f_j(x),
\)
when $f_j(x)\ge 0$ for all $x\in P$.
\begin{theorem}[\cite{Young95}]\label{thm:young}
  Let $\eta\in (0, 1)$ and assume that there is an oracle that, given a non-negative vector $z\in\RR^\mu$ returns $x\in P$ and $f(x)$ satisfying
  \(
    \sum_{j\in [m]} z_j f_j(x)
    \ge \alpha\cdot \max_{x\in P}\{ \sum_{j\in [m]} z_j f_j(x)\}
  \)
  for some constant $\alpha\le 1$,
  then there is an algorithm that computes $x\in P$ with $\min_{j\in [\mu]} f_j(x)\ge \alpha(1-\eta)\cdot \lambda^*$ in $O(\omega \eta^{-2} \log \mu / \lambda^*)$ iterations in each of which it does $O(\mu)$ work and calls the oracle once. The output $x$ is the arithmetic mean of the vectors returned by the oracle.
\end{theorem}

\paragraph{Set-Based Problem via Young's Algorithm.}
Clearly $\tilde \sigma_C$ is a linear function in $p$, namely $\tilde\sigma_C(p)=\sum_{S \subseteq V} p_S \tilde \sigma_C(S)$ and thus the problem $\max_{p \in \PPP} \min_{C \in \CCC}\tilde \sigma_C(p)$ takes exactly the form of a covering problem in the sense of Young with $\nu=2^n$, $\mu=m=|\CCC|$, $P=\PPP$, and $\omega=1$. Hence, we can compute a $(\alpha, 0)$-approximation for
$
    \max_{p\in\PPP}\min_{C\in\CCC} \tilde\sigma_C(p)
$,
if we provide an oracle with multiplicative approximation $\alpha$.
Hence, let us take a closer look at the requirements of Theorem~\ref{thm:young} in terms of the oracle problem. Given a non-negative vector $z\in\RR^m$, the oracle is required to return $p\in \PPP$ and $\tilde\sigma_C(p)$ for $C\in\CCC$ such that $\sum_{C\in\CCC} z_C \tilde \sigma_C(p) \ge \alpha \cdot \max_{p\in \PPP}\{\sum_{C\in\CCC} z_C \tilde \sigma_C(p)\}$ for some $\alpha\le 1$. Note that, by linearity of expectation
\[
    \sum_{C\in\CCC} z_C \tilde\sigma_C(p)
    = \E_{S\sim p}\Big[\sum_{C\in \CCC} z_C\cdot \frac{1}{|C|}\sum_{v\in C}\tilde\sigma_v(S)\Big]
    = \E_{S\sim p}\Big[\sum_{v\in V} \nw_v\cdot \tilde\sigma_v(S)\Big],
\]
where $\nw_v:= \sum_{C\in \CCC: v\in C}z_C/|C|$. Hence the oracle problem that the Young's algorithm has to solve takes the form
\begin{align}\label{formula: oracle}
    \opt_{\OOO}(G, \CCC, k, \nw) := \max_{p\in \PPP} \tilde\sigma^\nw(p).
\end{align}
We obtain the following lemma that shows that there is always an optimal solution of linear support.
\begin{lemma}
    It holds that $\opt_{\OOO}(G, \CCC, k, \nw) = \opt_{\QQQ}(G, \CCC, k, \nw)$ with
    \begin{align}\label{formula: linear}
        \opt_{\QQQ}(G, \CCC, k, \nw)
        := \max_{q\in \QQQ} \sum_{i=1}^n q_i \tilde\sigma^\nw(S^*_i),
    \end{align}
    where $\QQQ:=\{q\in [0,1]^n: \ones^Tq = 1, \sum_{i=1}^n i\cdot q_i \le k\}$
    and $S^*_i\in \argmax\{\sigma^\nw(S):S\in \binom{V}{i}\}$ for $i\in[n]$.
\end{lemma}
\begin{proof}
    Let $q^*$ be an optimal solution to~\eqref{formula: linear}. Define $p\in\PPP$ as $p_S = q^*_i$ if $S=S^*_i$ and 0 otherwise. Then, $p$ is a feasible solution to~\eqref{formula: oracle} and thus $\opt_{\OOO}(G, \CCC, k, \nw)\ge \opt_{\QQQ}(G, \CCC, k, \nw)$. It remains to prove that $\opt_{\OOO}(G, \CCC, k, \nw)\le \opt_{\QQQ}(G, \CCC, k, \nw)$. For this sake let $p^*\in\PPP$ be an optimal solution to~\eqref{formula: oracle}. Assume $p_S>0$ for some set $S$ of cardinality $i$, but $S\neq S^*_i$. Consider the solution $p':= p - p_S \ones_S + p_S \ones_{S^*_i}$. Clearly, $\tilde \sigma^\nw(p) \le \tilde \sigma^\nw(p')$ as $S^*_i$ by definition is a maximizing set of size $i$. This modification can be repeated until we obtain a solution $\bar p$ with $\bar p_S=0$ for all sets $S$ but the sets $S^*_1,\ldots S^*_n$. Clearly $\bar p$ is an optimal solution to~\eqref{formula: oracle} as each modification did not decrease the value of $\tilde \sigma^\nw(\cdot)$. Now consider the vector $q$ defined by $q_i=\bar p_{S^*_i}$ for $i\in[n]$. Then, $q$ is a feasible solution to~\eqref{formula: linear} and thus $\opt_{\OOO}(G, \CCC, k, \nw)\le \opt_{\QQQ}(G, \CCC, k, \nw)$.
\end{proof}
In other words, among the vectors that attain the optimum $\opt_{\OOO}(G, \CCC, k, \nw)$, there is also one that assigns a positive value to at most $n$ sets. Namely, to one set $S^*_i\in \binom{V}{i}$ for each $i\in [n]$. We now observe that $\tilde\sigma^\nw(\cdot)$ is a submodular and monotone set function. Hence, for each $i\in [n]$, the greedy hill climbing algorithm computes $(1-1/e, 0)$ approximations to $\quad\max\{\tilde\sigma^\nw(S):S\subseteq V, |S| \le i\}$.
Let $S_1,\ldots, S_n$ denote these approximate solutions. Note that we can assume $S_1\subseteq S_2 \subseteq \ldots \subseteq S_n$ as we can compute all the sets via one run of the greedy algorithm with budget $n$. Now consider the optimization problem
\begin{align}\label{formula: tilde linear}
    \opt_{\QQQ}^{\gr}(G, \CCC, k, \nw)
    := \max_{q\in \QQQ} \sum_{i=1}^n q_i \tilde\sigma^\nw(S_i)
\end{align}
that is identical to~\eqref{formula: linear} up to the replacement of $S^*_i$ by $S_i$.
We obtain the following lemma.
\begin{lemma}
    The vector $\ones_k$ is an optimal solution to the problem in~\eqref{formula: tilde linear}. Consequently, $\tilde\sigma^\nw(S_k)= \opt_{\QQQ}^{\gr}(G, \CCC, k, \nw)\ge (1 - 1 / e) \cdot \opt_{\OOO}(G, \CCC, k, \nw)$.
\end{lemma}
\begin{proof}
    Let $q\in\QQQ$ be arbitrary. For $i\in[n]$, define $\alpha_i:=\sum_{j=i}^n q_j$ and $\Delta_i:=\sigma^\nw(S_i) - \sigma^\nw(S_{i - 1})$ with $S_0=\emptyset$. Recall that $\QQQ:=\{q\in [0,1]^n: \ones^Tq = 1, \sum_{i=1}^n i\cdot q_i \le k\}$ and notice that the last constraint implies $\sum_{i=1}^n\alpha_i = \sum_{i=1}^n i \cdot q_i \le k$. Now, consider the optimization problem $\max_{\beta\in [0,1]^n}\{\sum_{i=1}^n \beta_i \Delta_i : \sum_{i=1}^n \beta_i \le k\}$. Note that, by submodularity, $\Delta_1\ge \Delta_2\ge \ldots \ge \Delta_n$ and thus the optimum is $\sum_{i=1}^k \Delta_i$. This implies that
    \[
        \sum_{i=1}^n q_i \tilde\sigma^\nw(S_i)
        = \sum_{i=1}^n \alpha_i \Delta_i
        \le \tilde\sigma^\nw(S_k).\qedhere
    \]
\end{proof}
We showed that the set $S_k$ obtained by greedily maximizing $\tilde \sigma^\nw(\cdot)$ subject to a budget of $k$ yields a $(1-1/e, 0)$-approximation to the oracle problem. Hence we get the following theorem.

\begin{restatable}{theorem-rstbl}{algorithmsets}
    Let $\delta\in (0,\frac{1}{2})$ and $\eps\in(0,1)$. There is a polynomial time algorithm that, with probability at least $1-\delta$, computes $p\in\PPP$ s.t.\ $\min_{C\in \CCC} \sigma_C(p) \ge (1-\frac{1}{e}) \opt_\PPP(G, \CCC, k) - \eps$. Moreover, the support of $p$ is $O(\eps^{-2}n\log m/k)$.
\end{restatable}
\begin{proof}
    We have argued that, for $z\in\RR^m$, the greedy hill climbing algorithm can be used on $\tilde \sigma^\nw(\cdot)$, where $\nw$ is such that $\nw_v:= \sum_{C\in \CCC: v\in C}\frac{z_C}{|C|}$ for $v\in V$, for obtaining a set $S$ of cardinality $k$ that is a $(1-1/e, 0)$-approximate solution to the problem of maximizing $\sum_{C\in\CCC} z_C \tilde\sigma_C(p)$ over $\PPP$.
    Thus, we described an oracle with multiplicative approximation $1-1/e$.
    Applying Theorem~\ref{thm:young} with $\eta=\eps/2$ thus implies that Young's algorithm returns a solution $p\in\PPP$ with
    $\min_{C\in\CCC} \tilde \sigma_C(p) \ge \alpha \cdot (1-\frac{\eps}{2}) \max_{p\in\PPP} \min_{C \in \CCC} \tilde \sigma_C(p)\ge \alpha \cdot \max_{p\in\PPP} \min_{C \in \CCC} \tilde \sigma_C(p) -\frac{\eps}{2}$
    after $O(\eps^{-2}\log m/\lambda^*)$ iterations. 
    Observe that for any $p\in\PPP$ and $v \in V$, it holds that $\sigma_v(p) \ge \Pr_{S\sim p}[v \in S] $. Hence, $\sigma_v(p) \ge k/n$ for all $v \in V$ and $\lambda^*\ge k/n$. 
    Thus the number of iterations is bounded by $O(\eps^{-2}n\log m/k)$. As the oracle returns a single set in every iteration it follows that the support of $p$ is upper bounded in this way as well.
    Applying Lemma~\ref{lemma: opt approx} with $\eps/2$ leads that we get a $(1-1/e,\eps)$-approximation to $\opt_\PPP(G, \CCC, k)$ in polynomial time with probability at least $1-\delta$.
\end{proof}

%% file: experiments.tex
\section{Experiments} \label{sec : numerical tests}
We report on an experimental study on the two probabilistic maximin problems. 
In fact, we provide implementations of multiplicative-weight routines for both the set-based and the node-based problems.
The routine for the set-based problem is the one described in Section~\ref{subsec: prob sets}. We refer to this algorithm as \texttt{set-based} in what follows. For the node-based problem, an implementation of the LP-based algorithm from Section~\ref{subsec: prob nodes} does not seem promising as it requires solving a large LP. Instead, we propose a heuristic approach that is again based on a multiplicative-weight routine. The essential observation is that the optimization problem $\max_{x\in \XXX}\min_{C \in \CCC} \lambda_C(x)$ from Lemma~\ref{lemma: LP approx tau} is again a covering LP and thus can be solved using a similar multiplicative-weight routine. In this case however, the oracle problem turns out to be the LP-relaxation of the standard influence maximization problem and thus we are again faced with a linear program of a similar form. This is where our approach becomes heuristic, we propose to again use a greedy algorithm for influence maximization in order to obtain feasible solutions for this LP. While this comes without any guarantee on approximation ratio, it turns out to be very efficient in practice. We refer to this algorithm as \texttt{node-based} in the remainder of this section. 
In our study we use random, artificial, as well as real-world instances. On these instances, we compare our methods, both in terms of fairness and efficiency (i.e., total spread) with a standard implementation of the greedy algorithm for influence maximization and the (very straightforward) methods proposed by Fish et al.~\cite{fish2019gaps} as well as the more involved method due to Tsang et al.~\cite{tsang2019group}. 
We continue by describing the experimental setting in detail.

\subsection{Experimental Setting}
\paragraph{Competitors.}
The methods of Fish et al.~\cite{fish2019gaps} are simple heuristics. First, they propose to use the greedy algorithm that iteratively picks $k$ seeds such as to maximize the minimum probability of any node to be reached (note that this is not the same as the greedy algorithm for influence maximization). In our implementation of this algorithm, in order to break ties, we use the node of maximum degree. We refer to this algorithm as \texttt{greedy\_maximin}. Second, Fish et al.\ propose a routine called myopic that after choosing the node of maximum degree, iteratively chooses the node that has the minimum probability of being reached as a seed node for $k-1$ times. We call this algorithm \texttt{myopic\_fish}. As a third heuristic, called naive-myopic, they propose to choose the $k-1$ nodes of smallest probability all at once instead of in $k-1$ iterations. We omit the results of naive-myopic as they are much worse than the ones of \texttt{myopic\_fish}.

The algorithm of Tsang et al.~\cite{tsang2019group} is much more involved. They phrase the problem as a multi-objective submodular optimization problem and design an algorithm to tackle such multi-objective submodular optimization problems that provides an asymptotic approximation guarantee of $1-1/e$. Their algorithm, that improves over previous work by Chekuri et al.~\cite{ChekuriVZ10} and Udwani~\cite{Udwani18}, is a Frank-Wolfe style algorithm that simultaneously optimizes the multilinear extensions of the submodular functions that describe the coverage of the respective communities.
We stress that their setting is less general than ours as the algorithm only satisfies an approximation guarantee in the case where the number of communities is $o(k \log^3(k))$. We use their python implementation as provided while choosing gurobi as solver since the other alternative md (their implementation of a mirror-descent) is much less efficient on the instances tested. We refer to the algorithm by Tsang et al as \texttt{tsang}.

We also compare our algorithm to the standard greedy algorithm for influence maximization. We use the slightly more involved and very efficient TIM implementation~\cite{TangXS14}. We refer to this method as \texttt{tim} in our evaluation.

We also compare our algorithms to the ultimate baseline for randomized strategies that is given by the uniform node-based solution, i.e., every node $v$ is chosen as a seed with probability $x_v=k/n$. We refer to this method as \texttt{uf\_node\_based}.

\paragraph{Implementation Details.}
We implement the multiplicative-weight routines for both the set-based and the node-based problems in C++ and the routines from Fish et al.\ in python using networkx~\cite{SciPyProceedings_11} for graph related computations. We used the TIM algorithm for influence maximization in order to solve the oracle problems for both multiplicative weight routines. We choose the $\eps$ parameter of the multiplicative weight routine (called $\eta$ in Theorem~\ref{thm:young}) to be 0.1. We use the IC model as diffusion model in all our experiments.
For the methods due to Fish et al.~\cite{fish2019gaps} and Tsang et al.~\cite{tsang2019group} (as also proposed by them), we use a constant number of 100 live-edge graphs for simulating the information spread instead of the number that guarantees $1-\eps$-approximations with probability $1-\delta$. As suggested by the confidence intervals in all our plots, this leads to sufficiently small variance on the instances tested.

All experiments were executed on a compute server running Ubuntu 16.04.5 LTS with 24 Intel(R) Xeon(R) CPU E5-2643 3.40GHz cores and a total of 128 GB RAM. The code was executed with python version 3.7.6 and C++ implementations were compiled with g++ $7.5.0$.
For the random generation of the graphs and the random choices of the live-edge graphs, we do not explicitly set the random seeds used by the random number generator. This does not prevent reproducability of our results as all the reported results are averages that are robust and independent of the random seeds chosen as indicated by the confidence intervals reported. 

\paragraph{Evaluation Details.}
For random and synthetic instances, each datapoint in our plots is the result of averaging over 25 experiments, 5 runs on each of 5 graphs generated according to the respective graph model. For real world instances each datapoint is the result of averaging over 5 runs on each graph. Error-bars in our plots indicate 95-\% confidence intervals. For the evaluation of $\sigma_v(x)$ we choose to approximate the value using a Chernoff bound in a way to obtain a $1\pm \eps$-approximation of the values with probability $1-\delta$ and in the reported experiments we choose both parameters as $0.1$.

We report both ex-ante and ex-post fairness values for our methods (for short, we use \texttt{ea} and \texttt{ep} as suffices). These have the following precise meaning. After computing probabilistic strategies $p$ or $x$ for the set-based and node-based problems, the ex-ante fairness values correspond to the objective values $\min_{C\in \CCC} \E_{S\sim p}[\sigma_C(S)]$ for the set-based and $\min_{C\in \CCC} \E_{S\sim x}[\sigma_C(S)]$ for the node-based problem. The ex-post values on the other hand are obtained by sampling a single set $S$ according to the probabilistic strategy $p$ or $x$ and then reporting the value $\min_{C\in \CCC} \sigma_C(S)$. 
We report also both ex-ante and ex-post values for the method of Tsang et al.~\cite{tsang2019group}, since, at the core, their algorithm works with the multilinear extension and thus also computes a continuous solution $x\in \RR^n$, i.e., a feasible solution to the node-based problem. Hence for their method we report both the value $\min_{C\in \CCC} \E_{S\sim x}[\sigma_C(S)]$ as ex-ante value and a value $\min_{C\in \CCC} \sigma_C(S)$ as ex-post value, where $S$ is computed by swap rounding from $x$ as described in their paper.

\begin{figure}[H]
    \centering
    \includegraphics[trim={0 1.5cm 0 0}, clip, width=.49\linewidth]{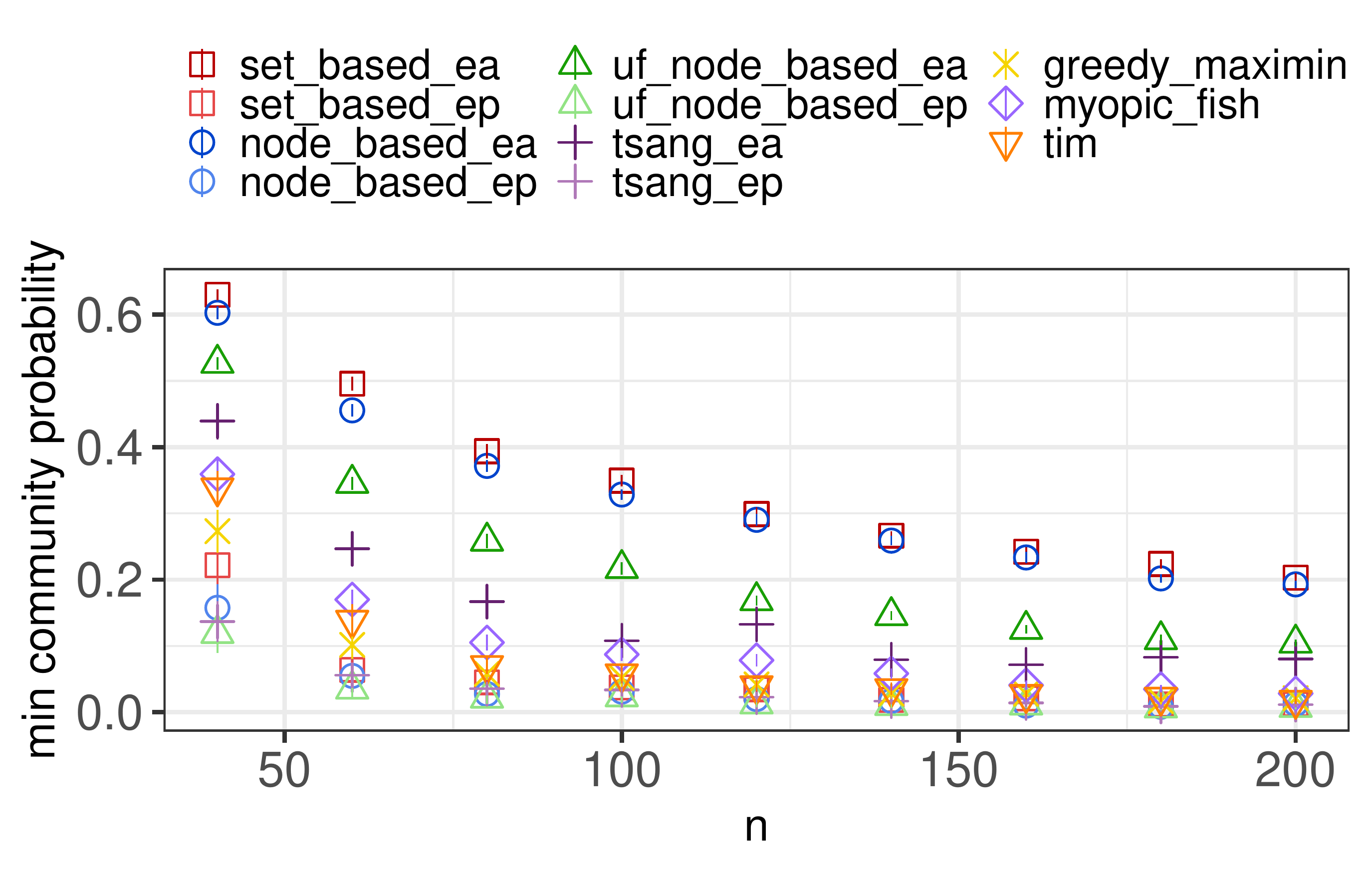}
    \includegraphics[trim={0 1.5cm 0 0}, clip, width=.49\linewidth]{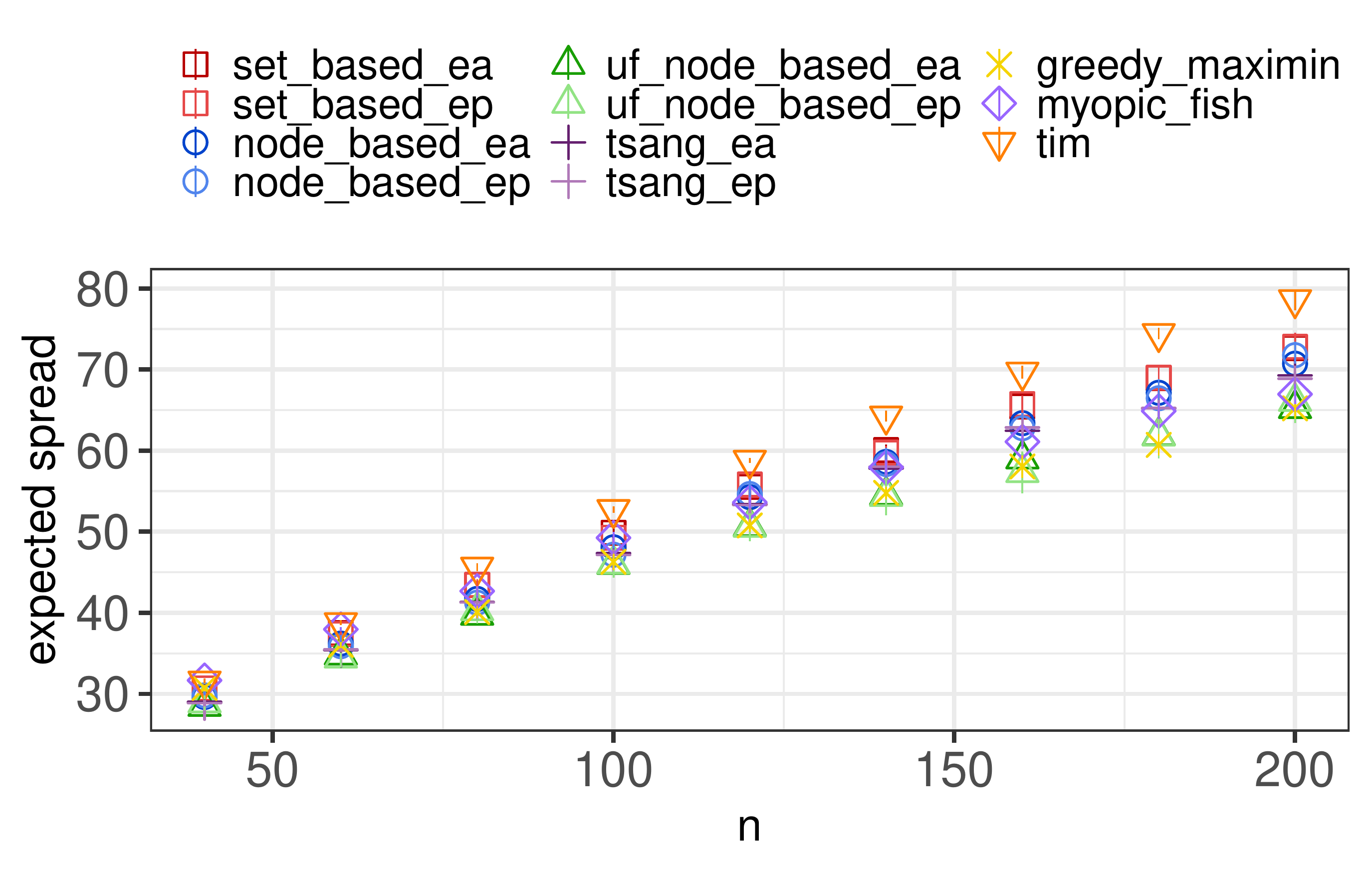}

    \includegraphics[trim={0 1.5cm 0 4.9cm}, clip, width=.49\linewidth]{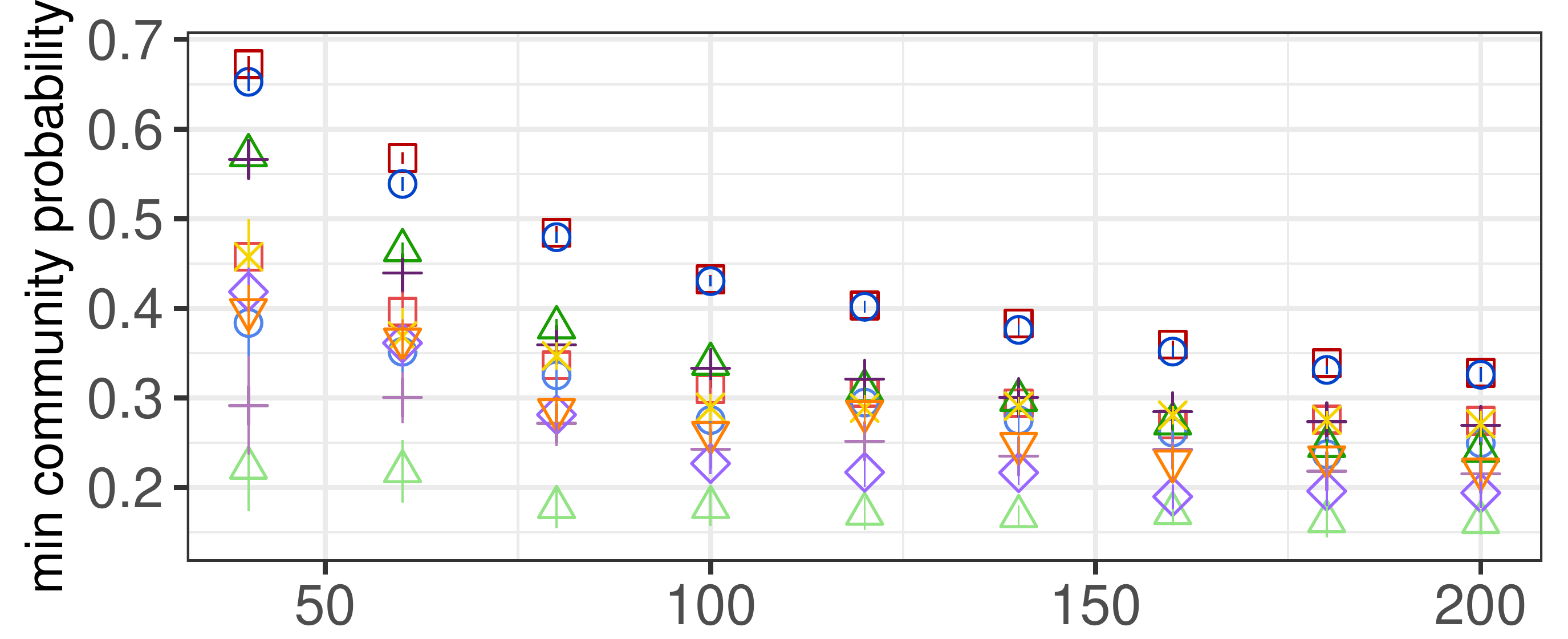}
    \includegraphics[trim={0 1.5cm 0 4.9cm}, clip, width=.49\linewidth]{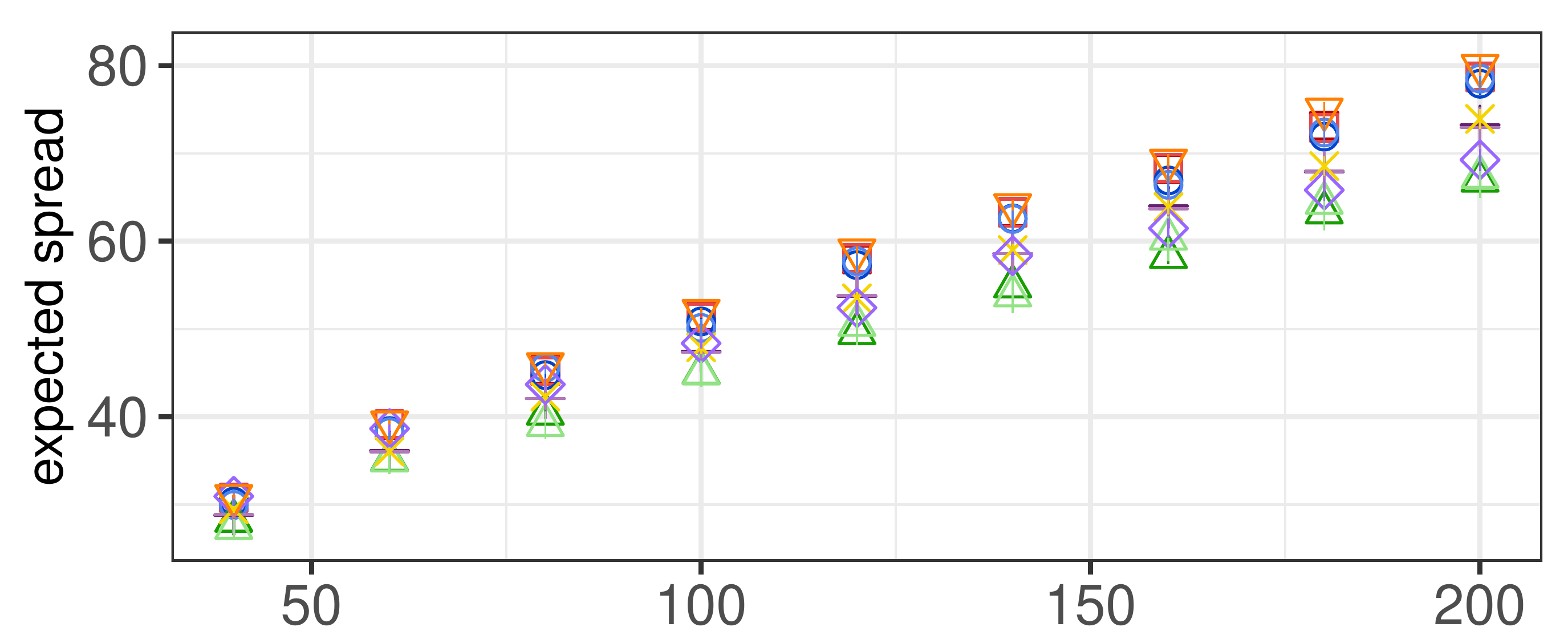}

    \includegraphics[trim={0 0 0 4.9cm}, clip, width=.49\linewidth]{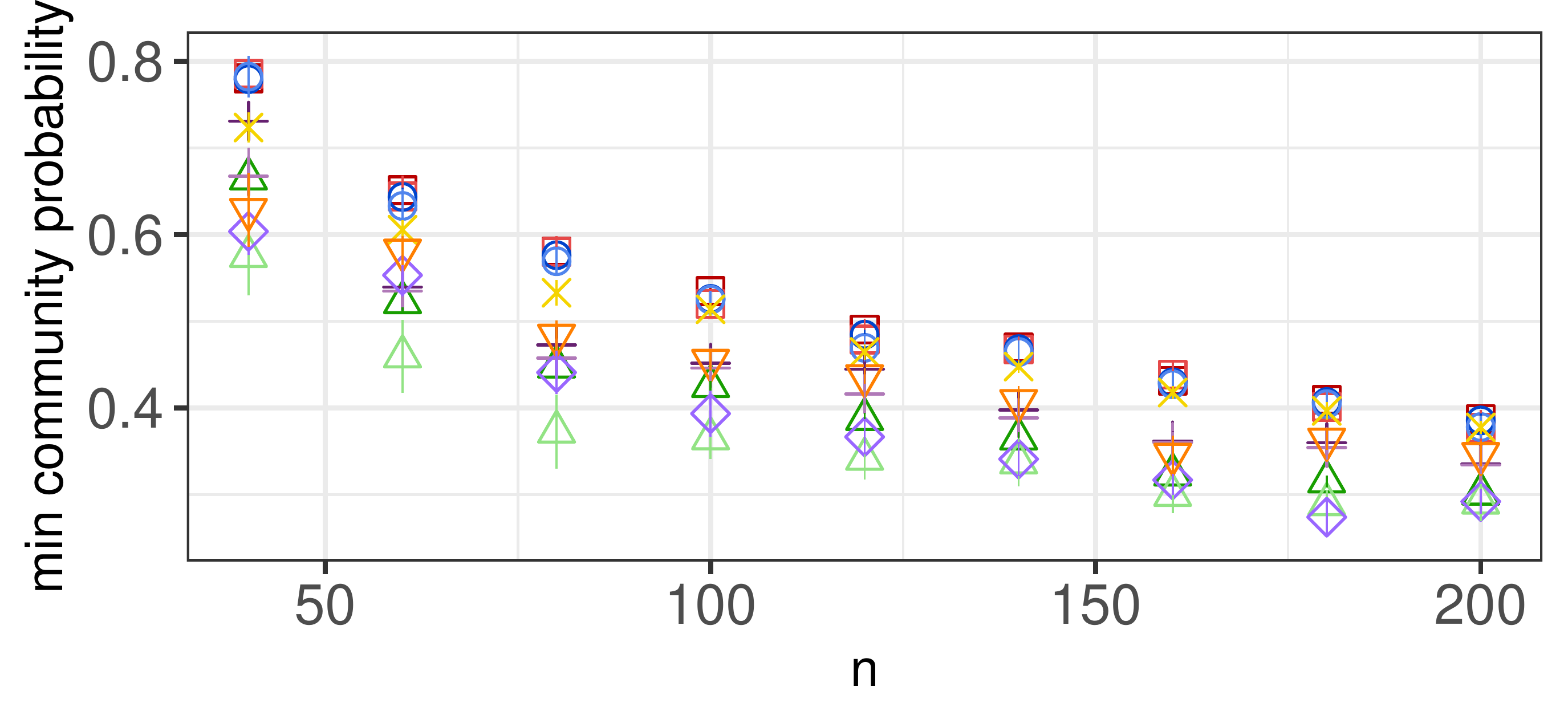}
    \includegraphics[trim={0 0 0 4.9cm}, clip, width=.49\linewidth]{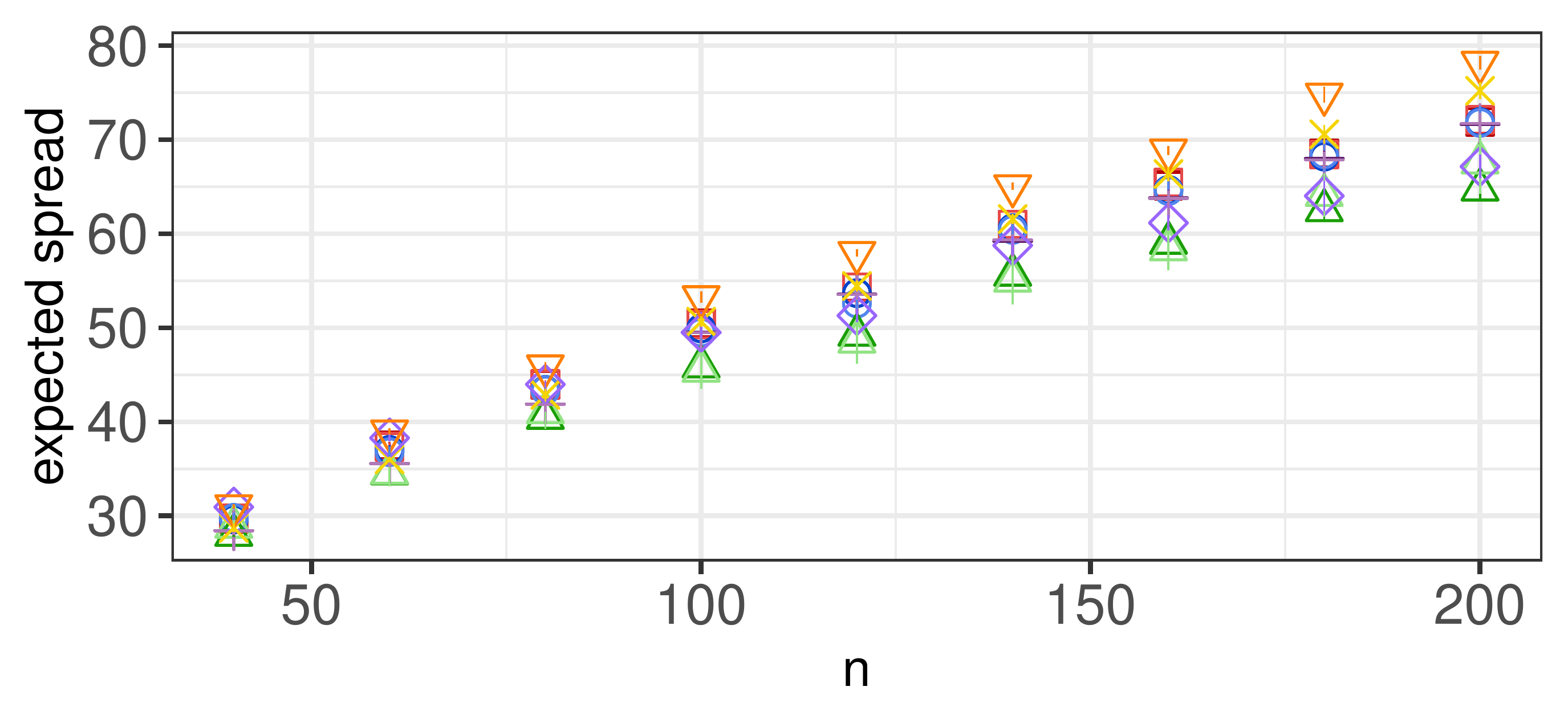}

    \caption{Results for Barab\'{a}si-Albert instances: $k=20$, $n$ increasing from $40$ to $200$ in steps of $20$. The minimum community probability is shown on the left, while expected spread is shown on the right. From top to bottom, we see (1) singleton community structure, (2) BFS community structure, and (3) random imbalanced community structure.}%
    \label{fig:barabasi albert}%
\end{figure}

\paragraph{Instances.} We evaluate the different algorithms on a vast set of instances. We proceed by describing the networks and the community structures that we use in our study.

\subparagraph{Networks.}
We use the following networks:
(1) Random instances generated according to the Barab\'{a}si-Albert model~\cite{Albert:2002:rmp} that yields scale-free networks -- a property frequently observed in social networks. We use the parameter modeling the preferential attachment to be $2$, i.e., every newly introduced node is connected to two previously existing nodes. 
(2) Random instances generated according to the block-stochastic model that is a natural choice in our setting as the instances come along with a community structure.
(3) The publicly available synthetic instances from the work of Tsang et al.~\cite{tsang2019group}. 
(4) Real world instances from the SNAP database~\cite{snapnets} and a paper by Guimer\`{a} et al~\cite{guimera2003self}, some of which were used also by Fish et al. We describe the real-world instances in detail below in the corresponding paragraph.
For most of our experiments and unless stated differently, we choose edge weights uniformly at random in the interval $[0,0.4]$ for random instances and the synthetic instances of Tsang et al., and in the interval $[0,0.2]$ for the real world instances.

\subparagraph{Community Structure.}
Regarding the community structure, in the case of some networks such as the block-stochastic networks, the synthetic networks due to Tsang et al., and some real world graphs the community structures are given. On all other networks we use some of the following different ways of constructing the community structure: 
(1) Singleton communities: every node is his own community. 
(2) BFS community structure: for a predefined number of communities $m$, we iteratively grow communities of size $n/m$ by breadth first search from a random source node (once there are no more reachable nodes but the community is still not of size $n/m$, we pick a new random source, until every node is in one of the $m$ communities).  
(3) Random imbalanced community structure: we randomly assign nodes to one of $m$ communities of fixed sizes. We use different values for the sizes and specify them for each of the experiments.

We note that the BFS community structure results in a rather connected community structure which is realistic for some applications. On the other hand, the random imbalanced community structure, is rather unconnected. Also this setting is realistic for some applications as for example if the groups indicate gender or ethnicity and we assume that people connect independently of these attributes.

\subsection{Results}
\paragraph{Barab\'{a}si-Albert Graphs.}
For the Barab\'{a}si-Albert graphs, we explore singleton communities, the BFS community structure with $m=k$,\footnote{We note that in this case the algorithm of Tsang et al.~satisfies its approximation guarantee.} and random imbalanced community structures of sizes $4n/10, 3n/10, 2n/10, n/10$. The results are reported in Figure~\ref{fig:barabasi albert}.
In the left plots in Figure~\ref{fig:barabasi albert}, we can see that the ex-ante values of our methods \texttt{set\_based} and \texttt{node\_based} dominate over all other ex-ante and ex-post values. Furthermore, we can see that particularly in the last plot, where the community structure is less simplistic, even the ex-post values of our methods dominate over the ones of all competitors. In the right plots, where we show the expected spread, we can see that \texttt{tim} outperforms other methods for all values of $n$. We note however that the advantage in efficiency of \texttt{tim} is not too pronounced, particularly in comparison to the disadvantage it yields in terms of fairness, see for example the plot on the top left.

\begin{figure}[ht]
    \centering
    \includegraphics[trim={0 1.5cm 0 0}, clip, width=.49\linewidth]{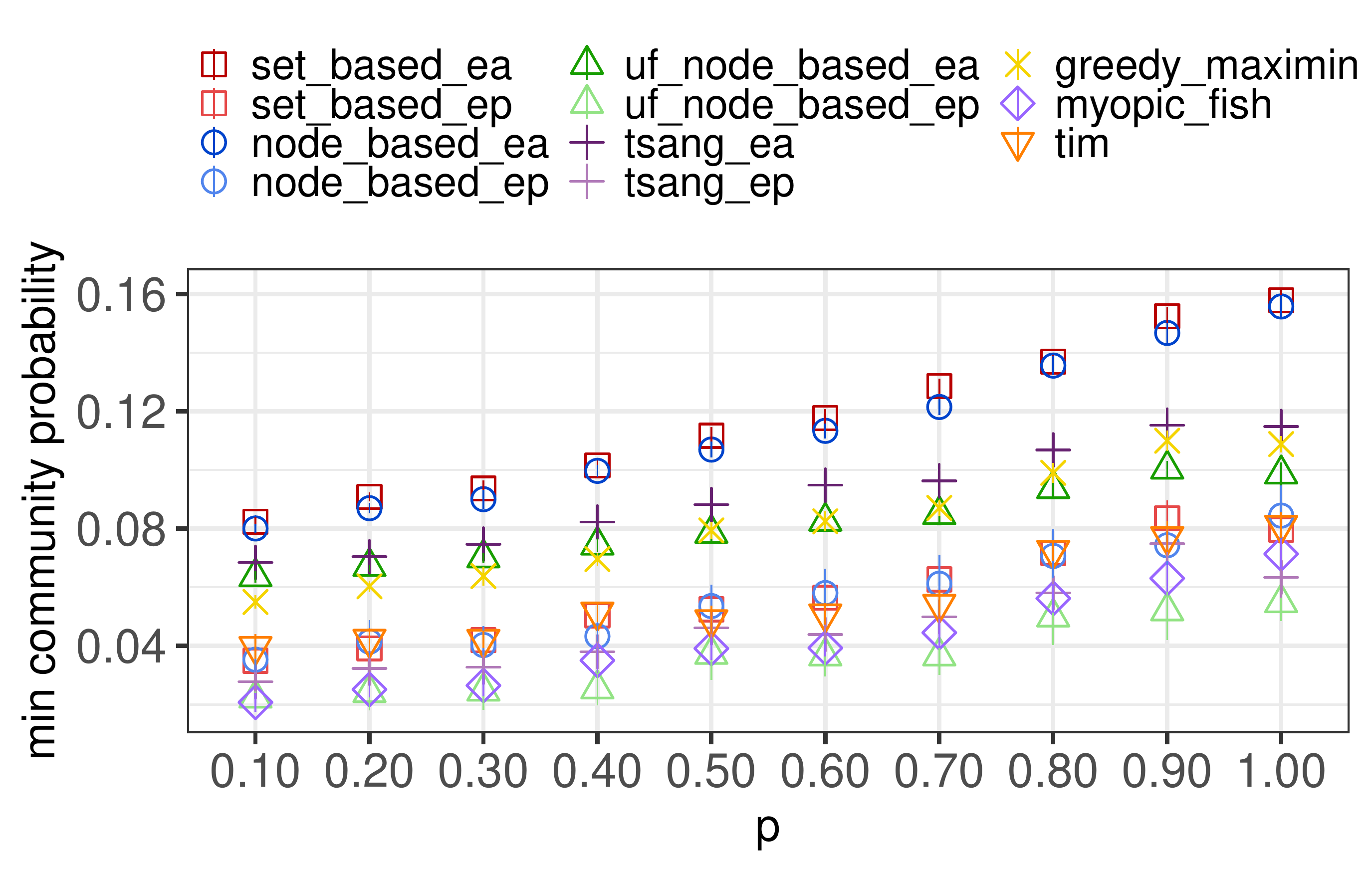}
    \includegraphics[trim={0 1.5cm 0 0}, clip, width=.49\linewidth]{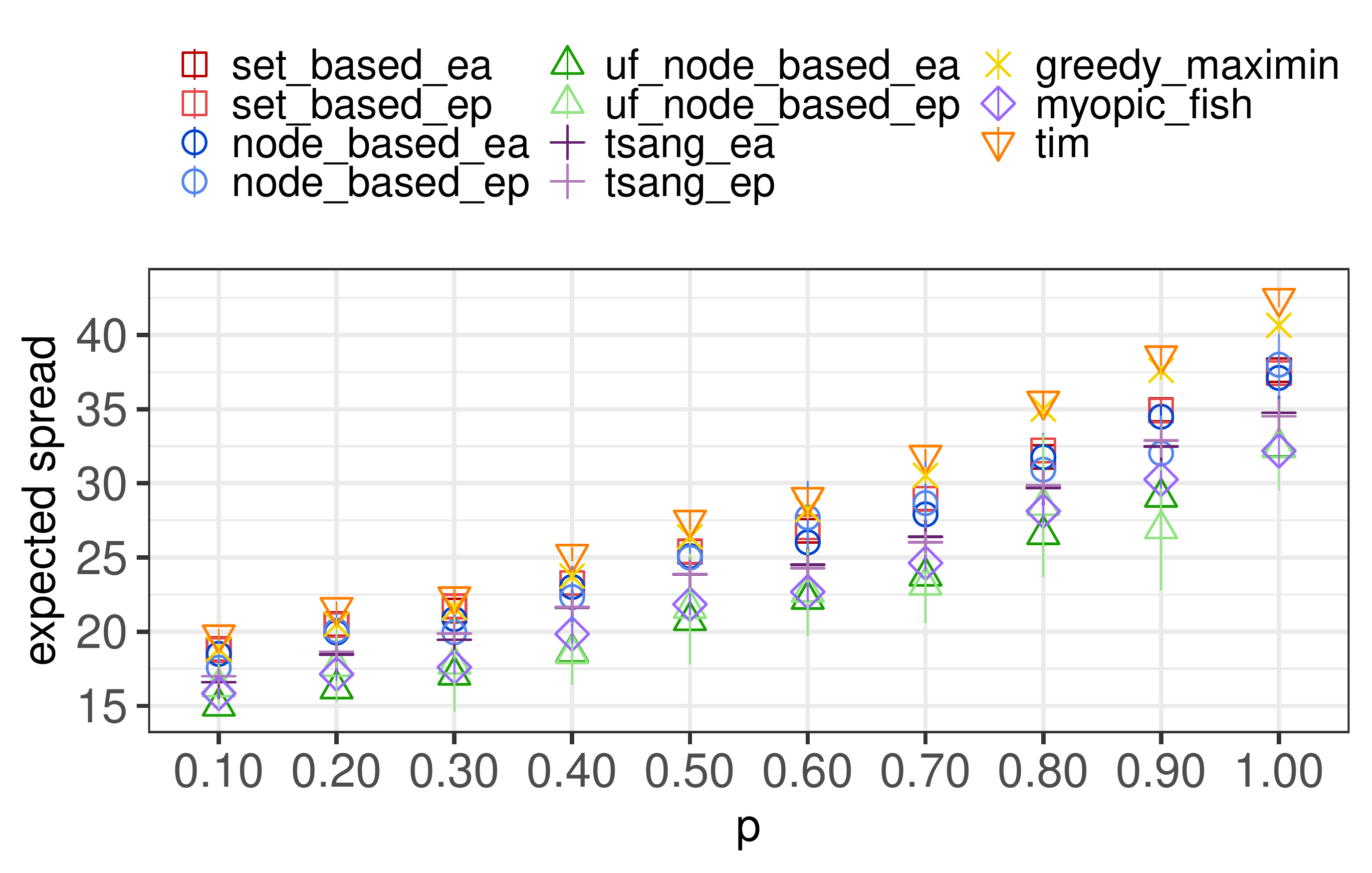}

    \includegraphics[trim={0 0 0 4.6cm}, clip, width=.49\linewidth]{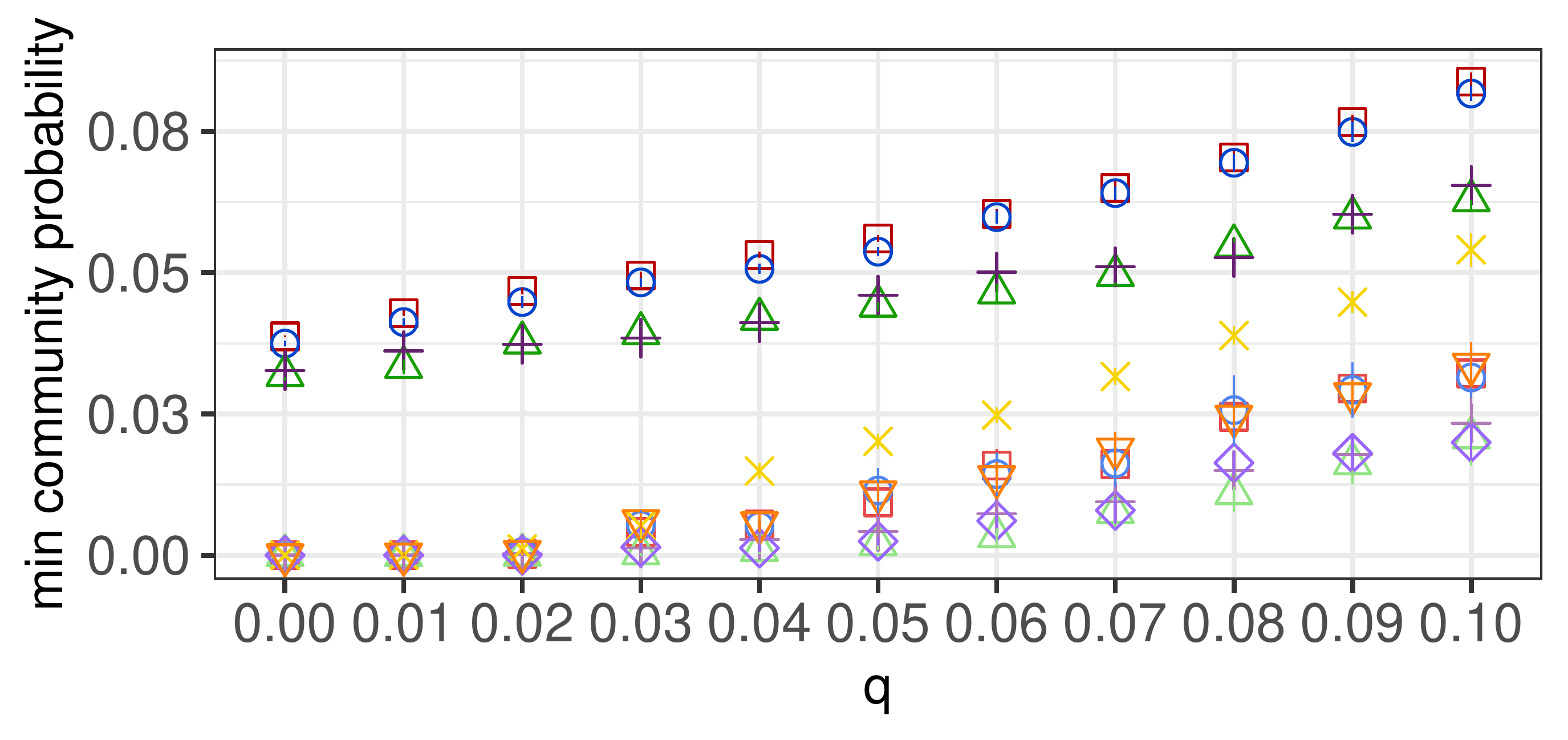}
    \includegraphics[trim={0 0 0 4.6cm}, clip, width=.49\linewidth]{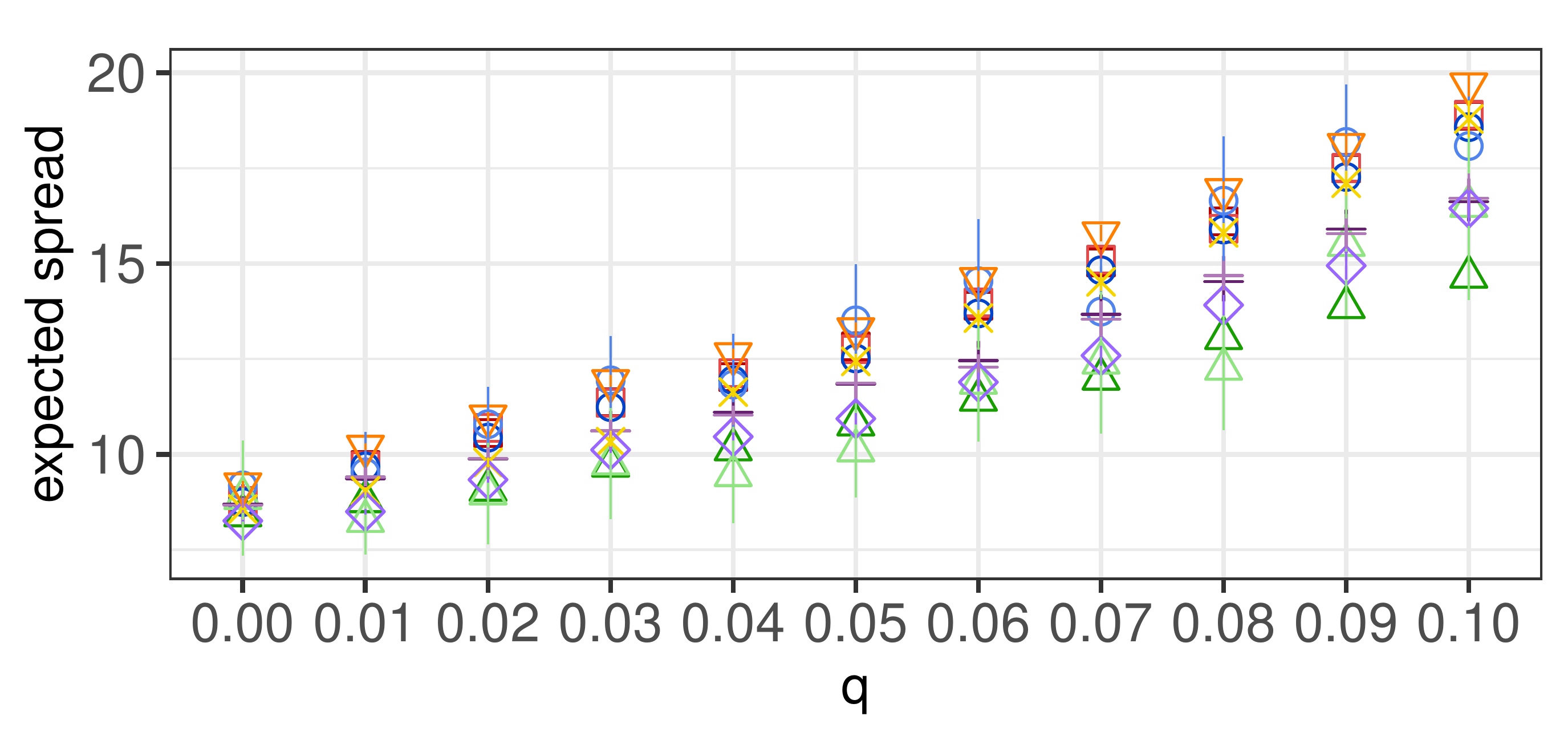}

    \caption{Results for block stochastic graphs: edge weights constant $0.05$, $k=8$, $n=200$. The minimum community probability is shown on the left, while expected spread is shown on the right. From top to bottom, we see (1) $q=0.1$ and $p$ increasing from $q$ to $1$ in steps of $0.1$, (2) $p=0.1$ and $q$ increasing from $0$ to $p$ in steps of $0.01$.}
    \label{fig: block stochastic}
\end{figure}

\paragraph{Block-Stochastic Graphs.} In order to further explore how the connectivity of the community structure influences the performance of the different approaches, we generate Block Stochastic graphs as follows. We fix the number of nodes to $200$, the number of communities to 16 with 6 communities of size $n/40$, 4 communities of size $2n/40$, 4 communities of size $4n/40$ and 2 communities of size $5n/40$.
We then choose two parameters $p$ and $q$, and create a sequence of instances where the probability of an edge within a community is $p$ and between communities $q$.
The larger choices of $p$ and $q$ yield very dense graphs and thus instances become trivial. We choose edge weights to be $0.05$ in this experiment as for the larger choice the instances become trivial as the minimum community probability becomes very large. 
The results are reported in Figure~\ref{fig: block stochastic}.
In the left plots we can see that again the ex-ante values of \texttt{set\_based} and \texttt{node\_baserd} dominate over all other values.
Clearly, by increasing $p$ and $q$ in both experiments, the values of all algorithms are increased. In the second experiment, for smaller $q$, when communities are better connected within each other than between each other, there is a bigger advantage for ex-ante values over ex-post values. In the right column, for smaller $p$ and $q$, all algorithms are close to each other. Again \texttt{tim} dominates over the other algorithms in terms of expected spread, while also the other algorithms -- and in particular \texttt{set\_based} and \texttt{node\_based} -- perform well.

\begin{figure}[htp]
    \centering
    \includegraphics[trim={0 1.5cm 0 0}, clip, width=.49\linewidth]{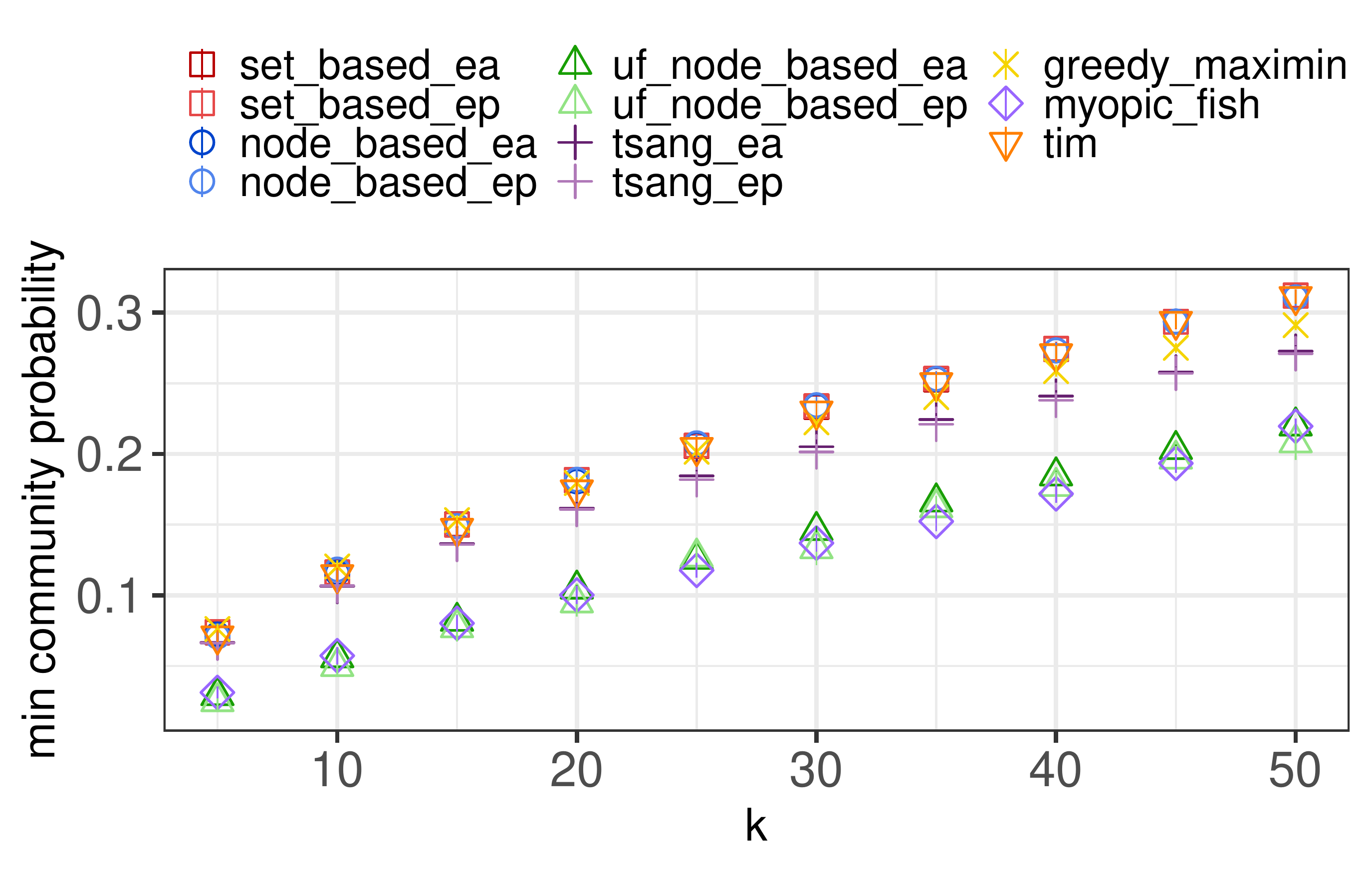}
    \includegraphics[trim={0 1.5cm 0 0}, clip, width=.49\linewidth]{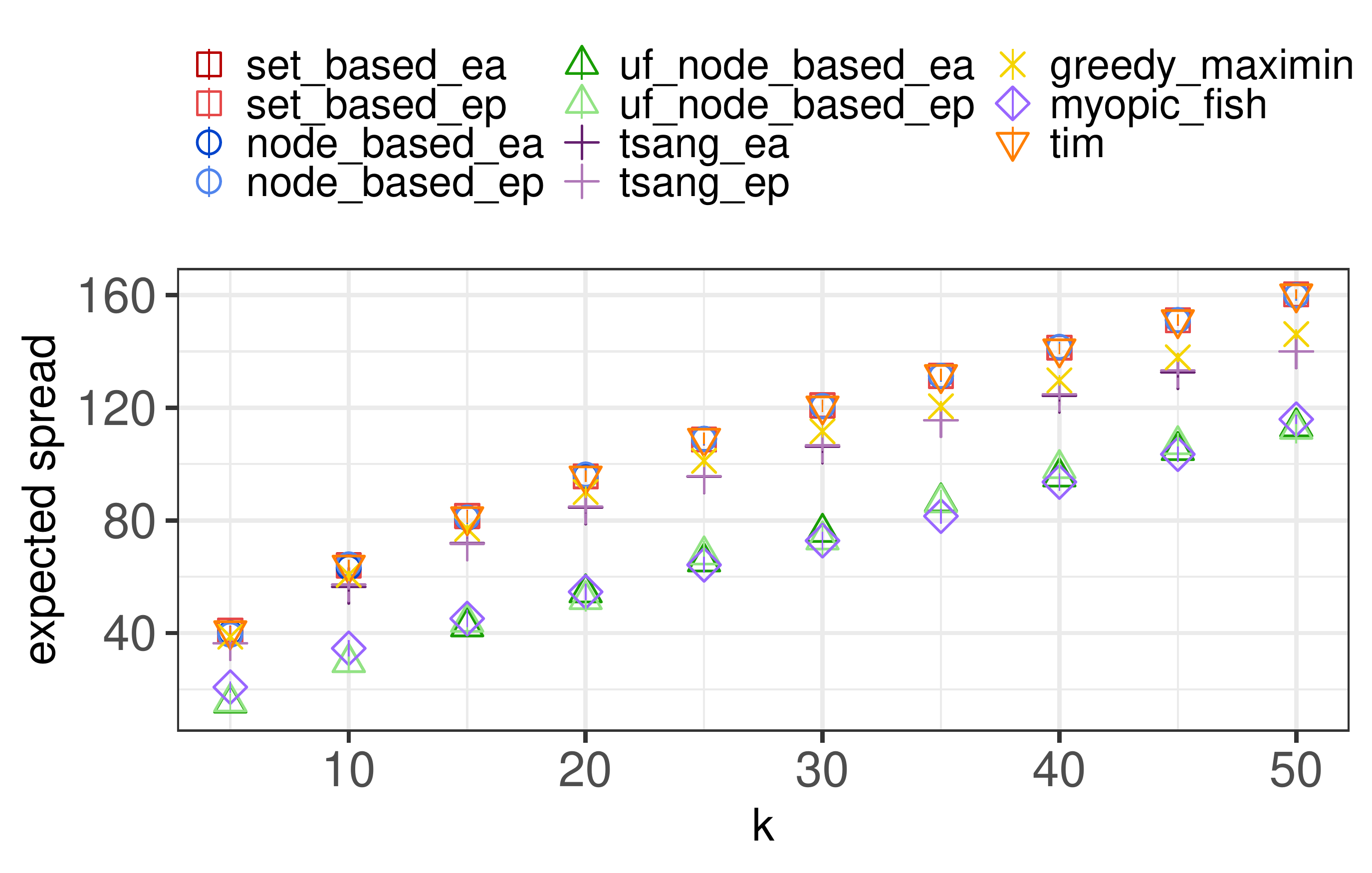}

    \includegraphics[trim={0 0 0 4.9cm}, clip, width=.49\linewidth]{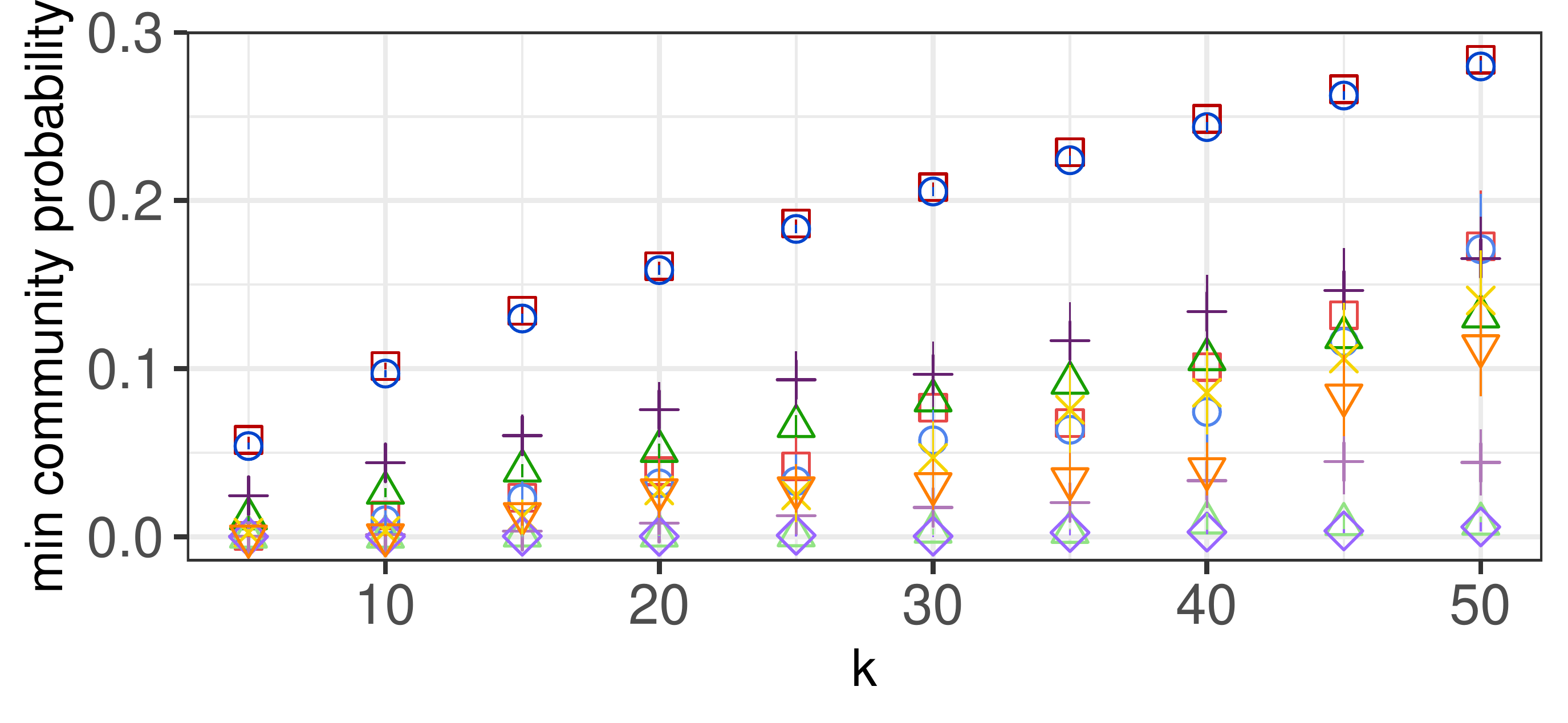}
    \includegraphics[trim={0 0 0 4.9cm}, clip, width=.49\linewidth]{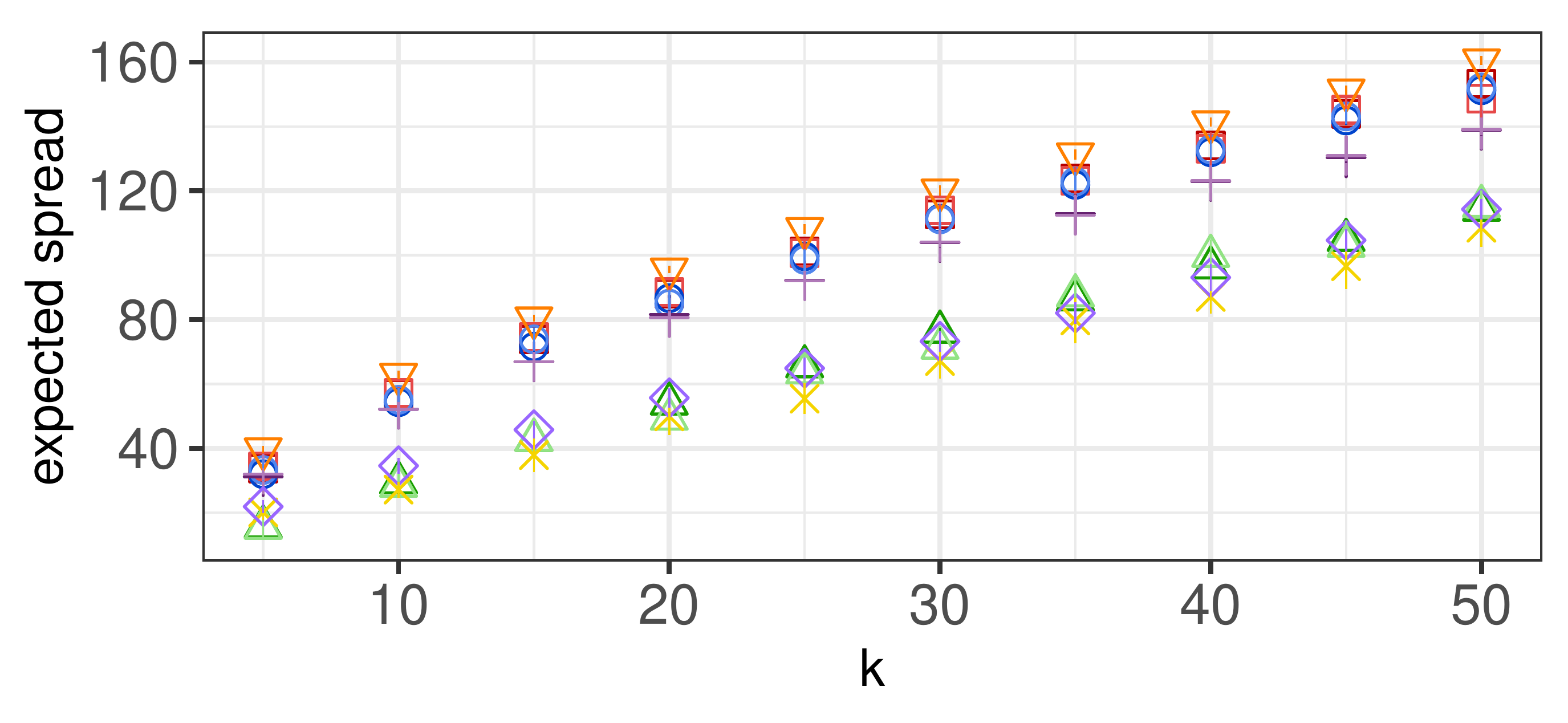}

\caption{Results for the instances of Tsang et al.~\cite{tsang2019group}: $k$ increasing from $5$ to $50$ in steps of $5$. The minimum community probability is shown on the left, while expected spread is shown on the right. From top to bottom, we see (1) community structure induced by attribute gender,
(2) community structure induced by attributes region, gender and ethnicity.}
\label{fig: tsang}
\end{figure}

\paragraph{Instances of Tsang et al.}
Next we evaluate the algorithms on the instances used by Tsang et al.~\cite{tsang2019group}. These are synthetic networks introduced by Wilder et al.~\cite{WilderOHT18} in order to analyze the effects of health interventions. Each of the 500 nodes in these networks has some attributes (region, ethnicity, age, gender, status) and more similar nodes are more likely to share an edge. The attributes induce communities and we test, as proposed by Tsang et al.~\cite{tsang2019group}, all algorithms w.r.t.\ group fairness of the communities induced by some of those attributes. The results are reported in Figure~\ref{fig: tsang}.
Again the ex-ante fairness values of our methods dominate over all other algorithms as can be seen in the left column. 
In the first experiment (communities induced by gender), the ex-post values of \texttt{set\_based}, \texttt{node\_based}, \texttt{tsang}, and \texttt{uf\_node\_based} are all almost identical to their respective ex-ante values.
In the second experiment (communities induced by three attributes, namely region, gender, and ethnicity) we obtain a much more complex community structure. Here, our algorithms \texttt{set\_based} and \texttt{node\_based} perform best not only in the ex-ante values, but also in terms of ex-post values for most values of $k$. Even more, in the right plots, we can see that the achieved values in expected spread by \texttt{tim} and our methods are very close to each other. 

\begin{figure}[ht]
  \centering
  \includegraphics[trim={0 1.5cm 0 0}, clip, width=.49\linewidth]{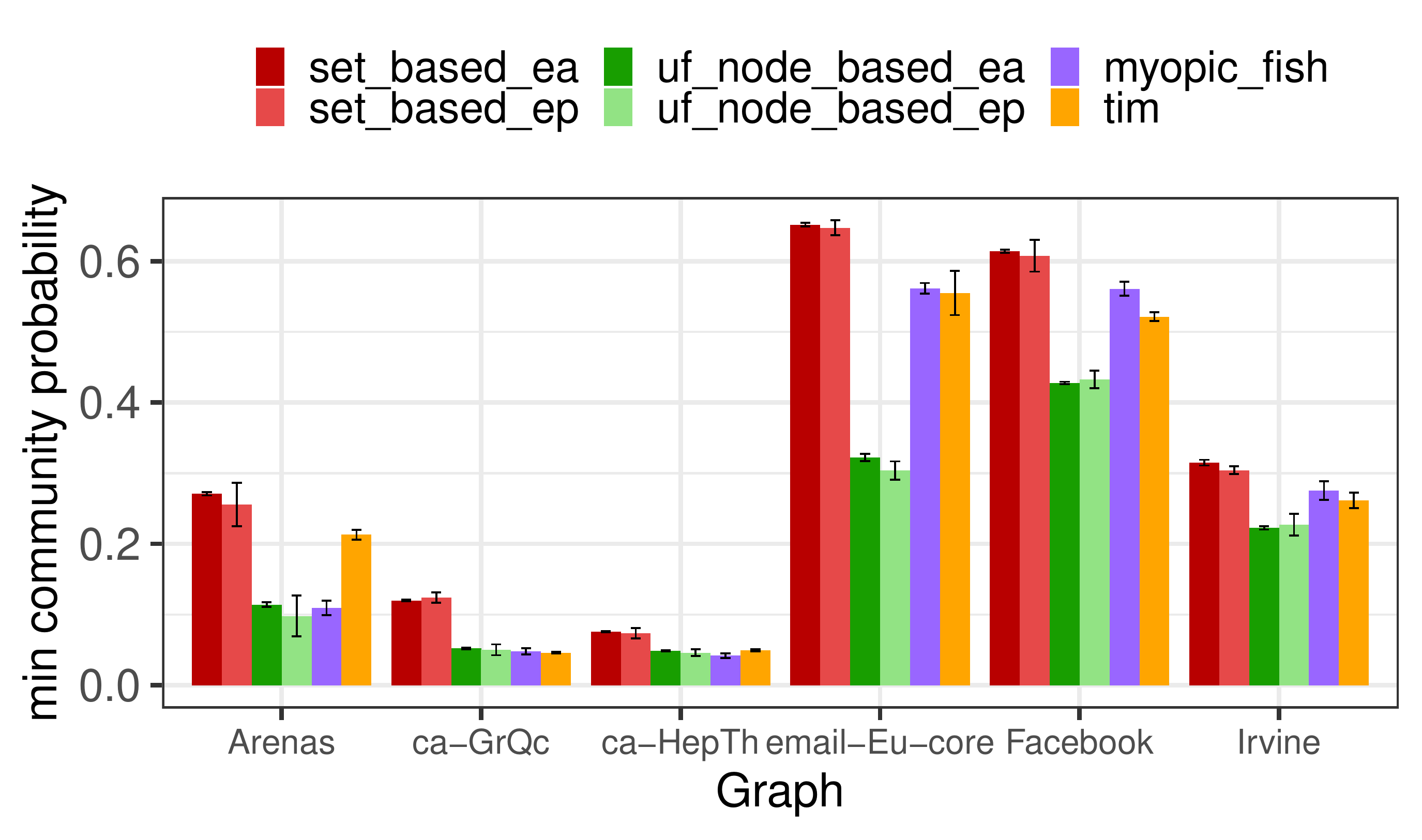}
  \includegraphics[trim={0 1.5cm 0 0}, clip, width=.49\linewidth]{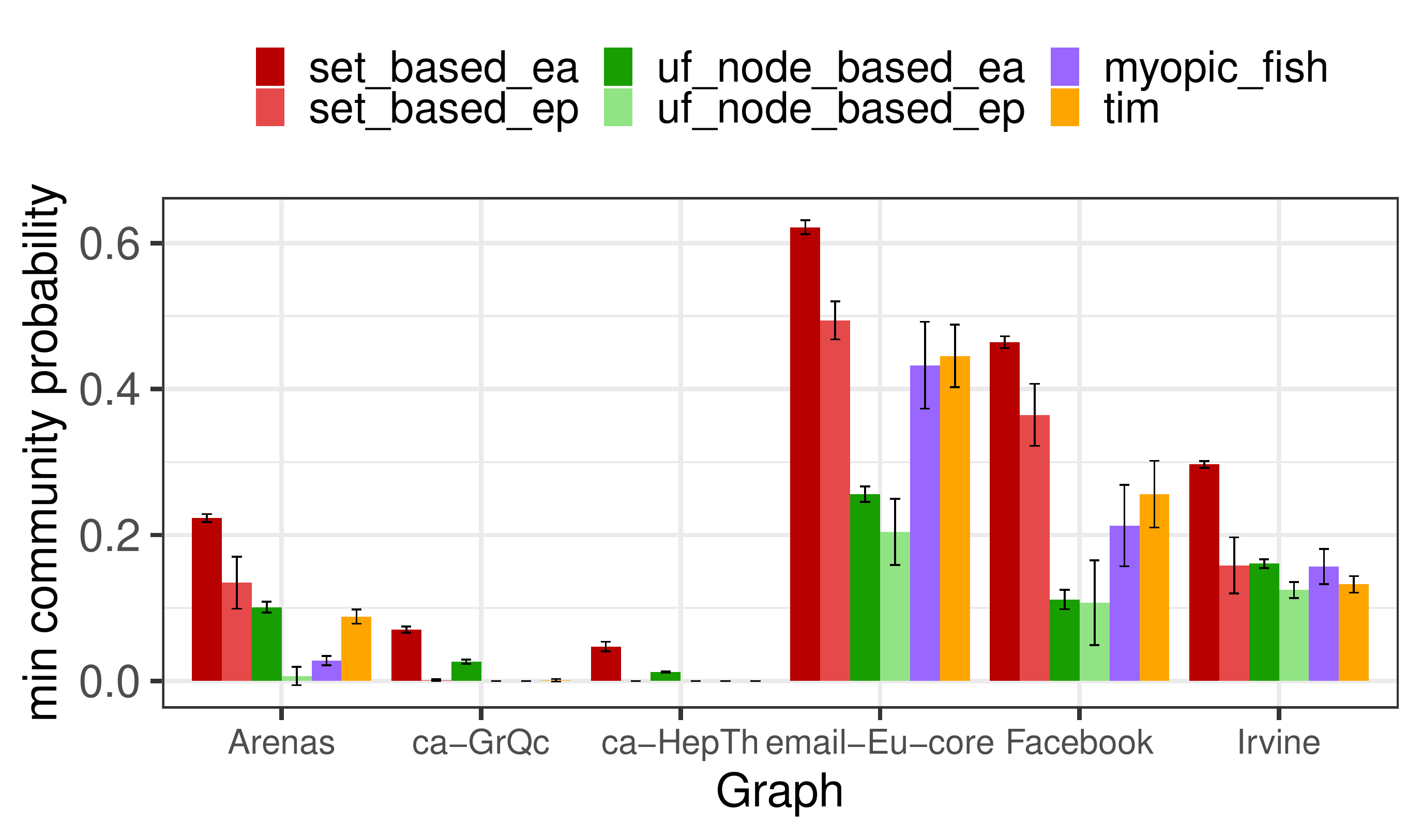}
  
  \includegraphics[trim={0 0 0 3.1cm}, clip, width=.49\linewidth]{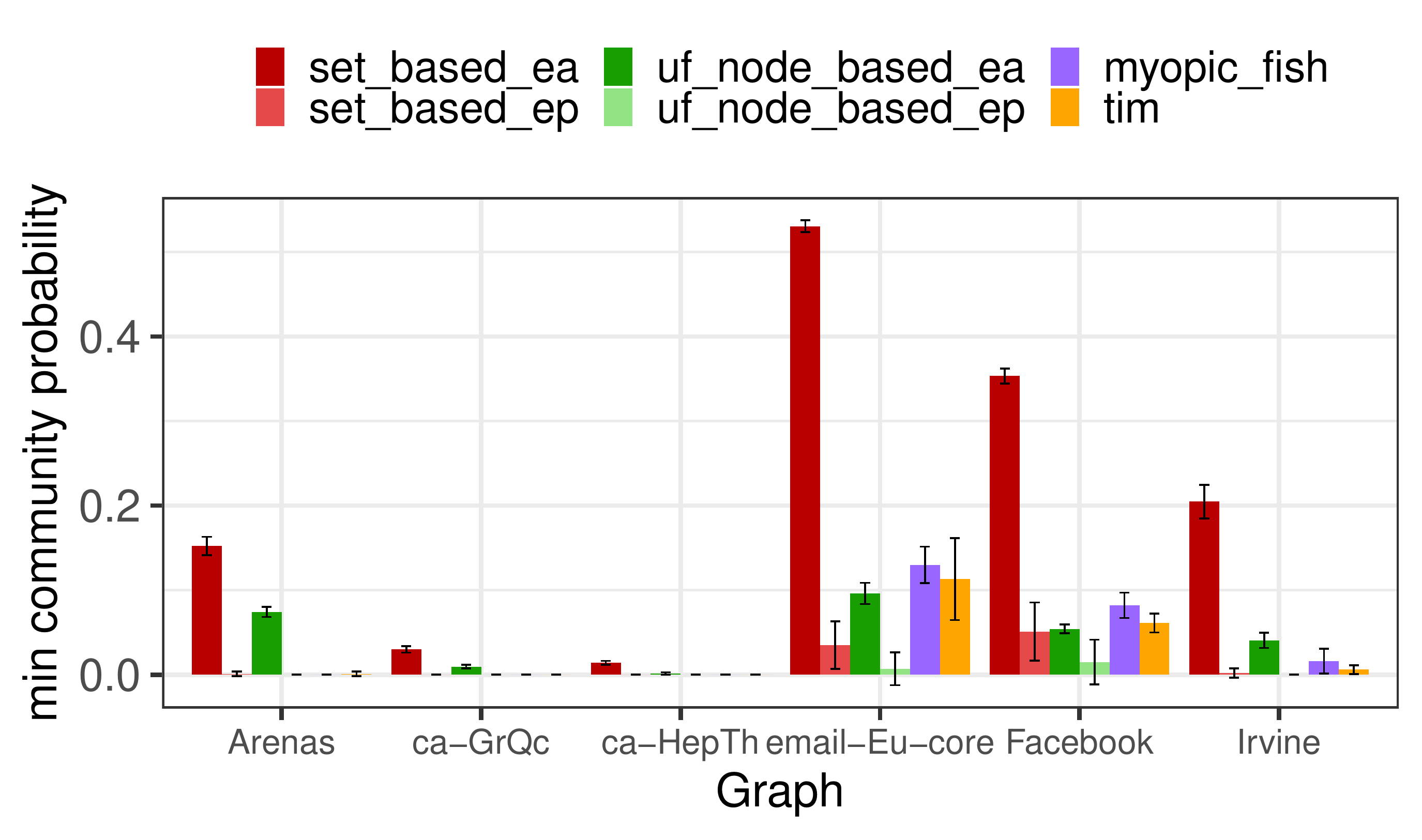}
  \includegraphics[trim={0 0 0 3.1cm}, clip, width=.49\linewidth]{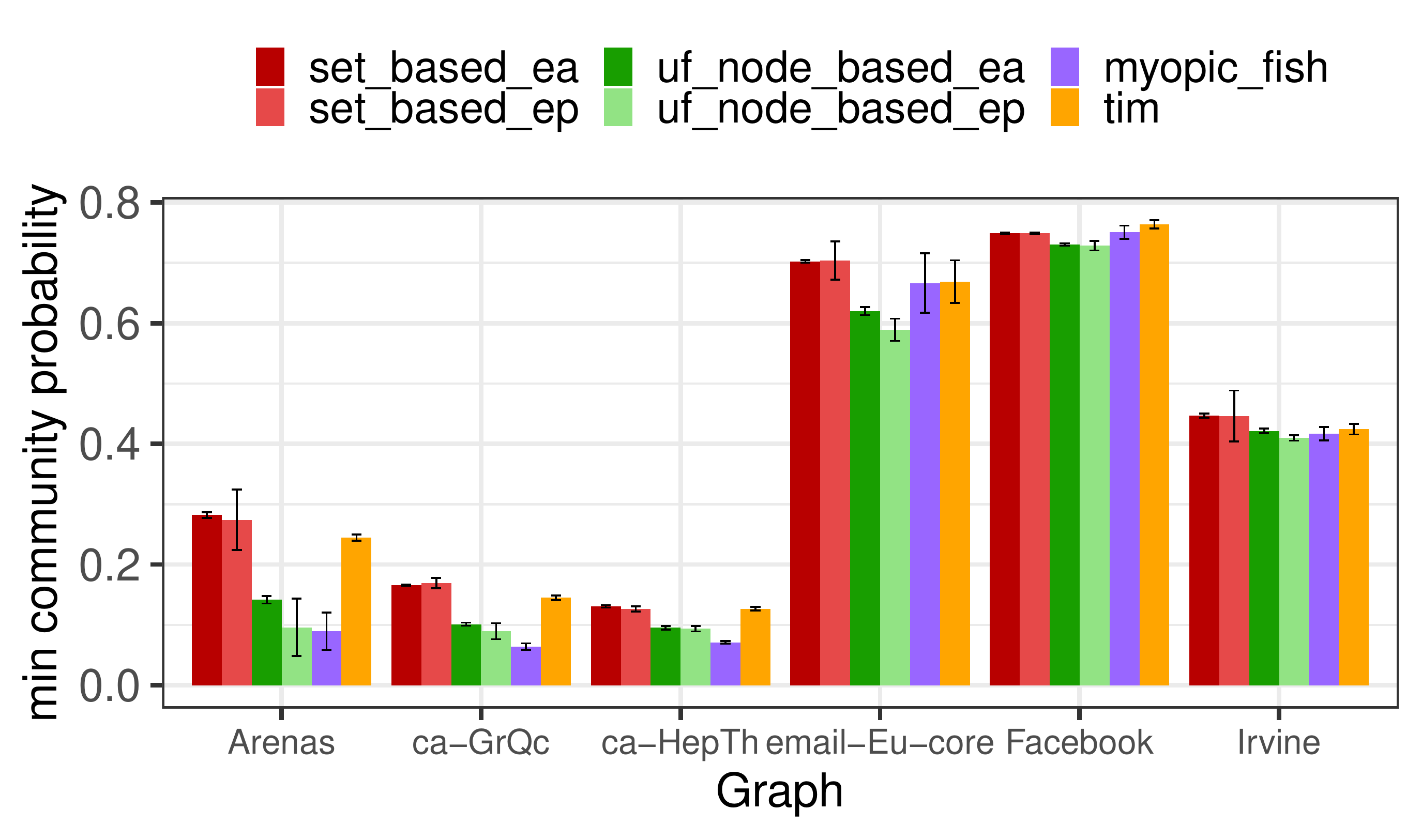}
\caption{Results for the instances used by Fish et al.\ for $k=100$. BFS community structure with (top left) $10$ communities, (top right) $n/10$ communities, (bottom left) $n/2$ communities. (bottom right) Random imbalanced community structure with $16$ communities.}
\label{fig: fish}
\end{figure}

\paragraph{Real World Instances.} We proceed by describing the used real world instances.
\begin{description}
    \item[\texttt{Arenas}~\cite{guimera2003self}] This dataset represents an email communication network at the University Rovira i Virgili (Spain). Each user is represented by a node and there is a directed edge from a node $u$ to a node $v$ if $u$ sent at least one email to $v$.
    \item[\texttt{email-Eu-core}~\cite{leskovec2007graph}] Also this dataset is an email network, this time from a large European research institution. Each member of the research institution belongs to one of 42 departments, which predefines a community structure. Nodes and edges have the same interpretation as in \texttt{Arenas}.
    \item[\texttt{ca-GrQc} and \texttt{ca-HepTh}~\cite{leskovec2007graph}] These datasets are co-authorship networks for two different categories of arXiv (General Relativity and Quantum Cosmology and High Energy Physics - Theory). The nodes in the networks correspond to authors and there is an undirected edge between two nodes if the authors co-authored at least one arXiv paper in this category. 
    \item[\texttt{Facebook}~\cite{leskovec2012learning}] This dataset represents a part of the Facebook network, where nodes are users and edges indicate friendships.
    \item[\texttt{Irvine}~\cite{opsahl2009clustering}] This dataset is a network created from an online community at the University of California, Irvine. Nodes here represent students and each directed edge represents that at least one online message was sent among the students.
    \item[\texttt{com-Youtube}~\cite{yang2015defining}] This dataset consists of a snapshot of the social network included in Youtube. The network has $1134890$ nodes and $2987624$ edges. Nodes correspond to users, edges to friendships between users. Also this network contains a predefined community structure that is given by the so-called Youtube groups. On Youtube, users can open groups that others can join. As the complete network is very large, we use a connected sub-network. We first remove all nodes that do not belong to any community. We then obtain the sub-network as the induced graph among the first 3000 nodes that are seen by a BFS from a random source node while removing singleton communities. The resulting network contains $3000$ nodes (some may not belong to any community), $29077$ edges. The number of communities is $1575$.
\end{description}
The number of nodes and edges as well as the information whether the networks are directed or undirected are summarized in Table~\ref{datasets}. If networks are undirected, we interpret the edges as existent in both directions.
\begin{wraptable}{H}{8cm}\small
	\begin{center}
        \begin{tabular}{c|c|c|c}
            \toprule
	        Dataset & Nodes & Edges & Direction  \\
    		\midrule 
    			\texttt{Arenas} & $1133$ & $5451$ & Directed\\
    			\texttt{email-Eu-core} & $1005$ & $25571$ & Directed\\
    			\texttt{ca-GrQc} & $5242$ & $14496$ & Undirected\\
    			\texttt{ca-HepTh} & $9877$ & $25998$ & Undirected\\
    			\texttt{Facebook} & $4039$ & $88234$ & Undirected\\
    			\texttt{Irvine} & $1899$ & $20296$ & Directed\\
    			\texttt{com-Youtube} & $3000$ & $29077$ & Undirected\\
    			\bottomrule 
        \end{tabular}
        \caption{Properties of real world networks.}
        \label{datasets}
   \end{center}
\end{wraptable}
In order to obtain non-trivial results, i.e., achieve non-zero minimum probabilities in the experiments (especially for the ex-post values), for each network (other than \texttt{com-Youtube}) we considered the largest weakly connected component.
We exclude \texttt{greedy} and the method of Tsang et al.\ from the further experiments as they are not efficient enough to deal with instances of this size. We also restrict to the set-based method from our two methods as the results of our two methods are very similar. The results are reported in the following figures.

\begin{figure}[ht]
    \centering
    \includegraphics[trim={0 1.5cm 0 0}, clip, width=.49\linewidth]{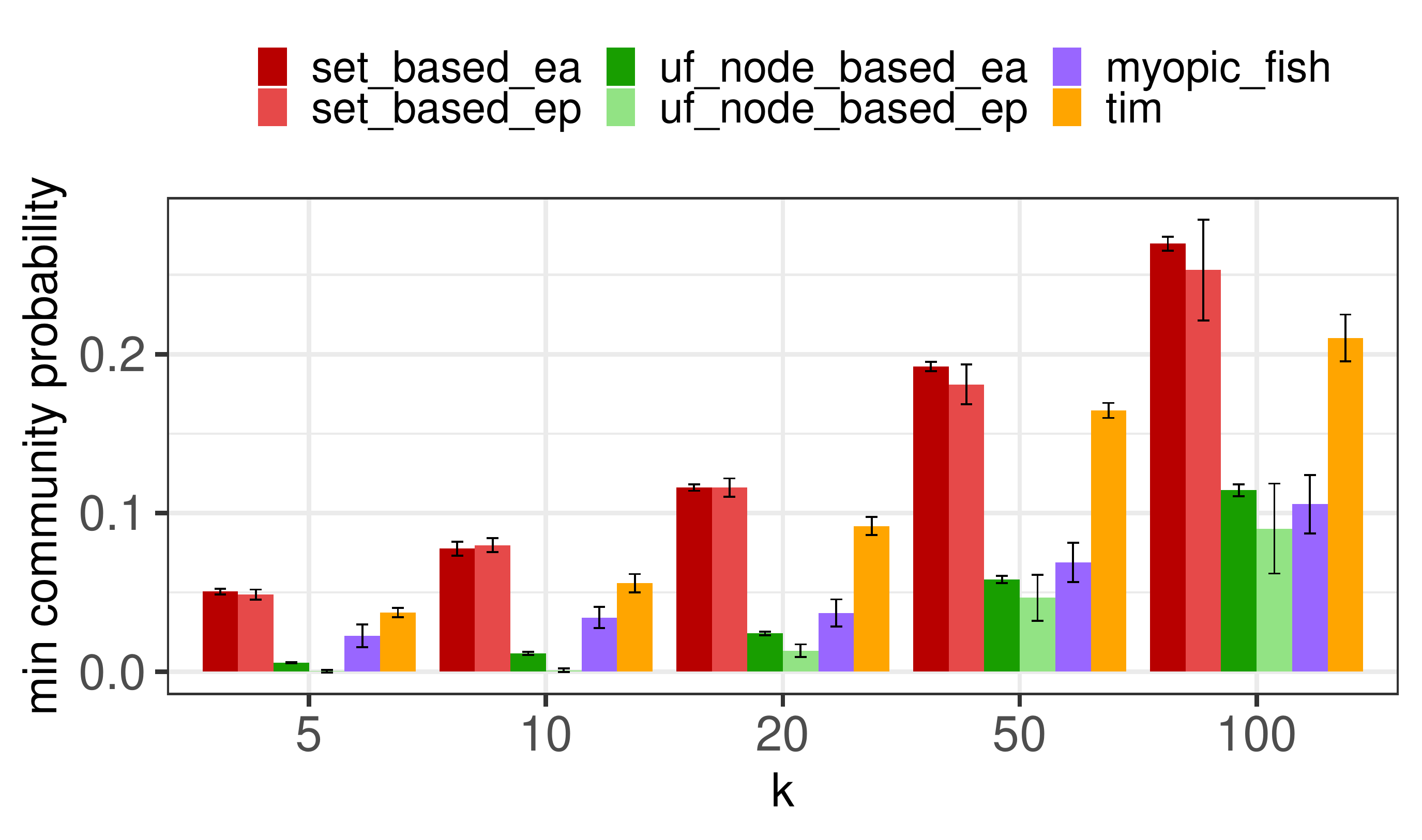}
    \includegraphics[trim={0 1.5cm 0 0}, clip, width=.49\linewidth]{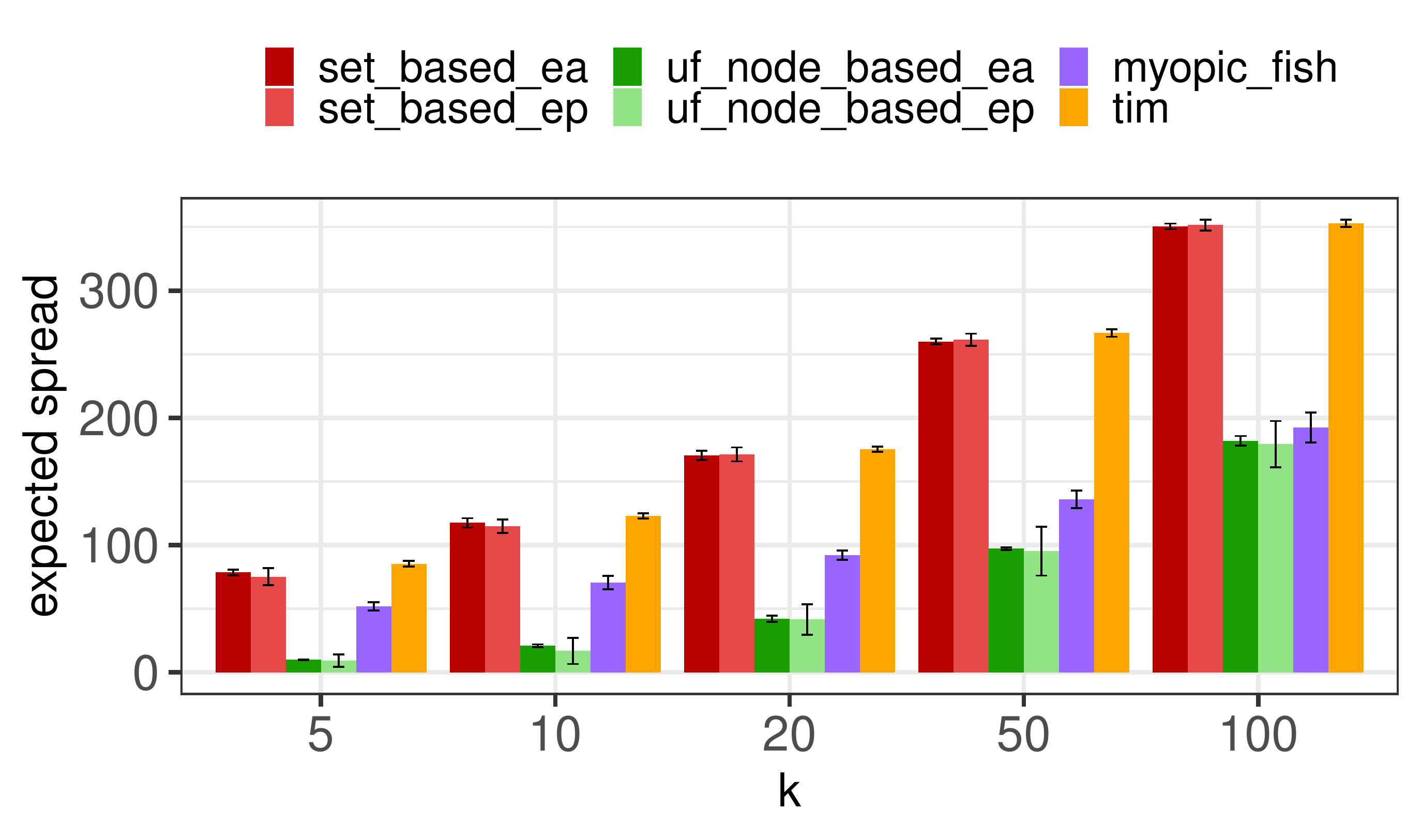}

    \includegraphics[trim={0 0 0 3.5cm}, clip, width=.49\linewidth]{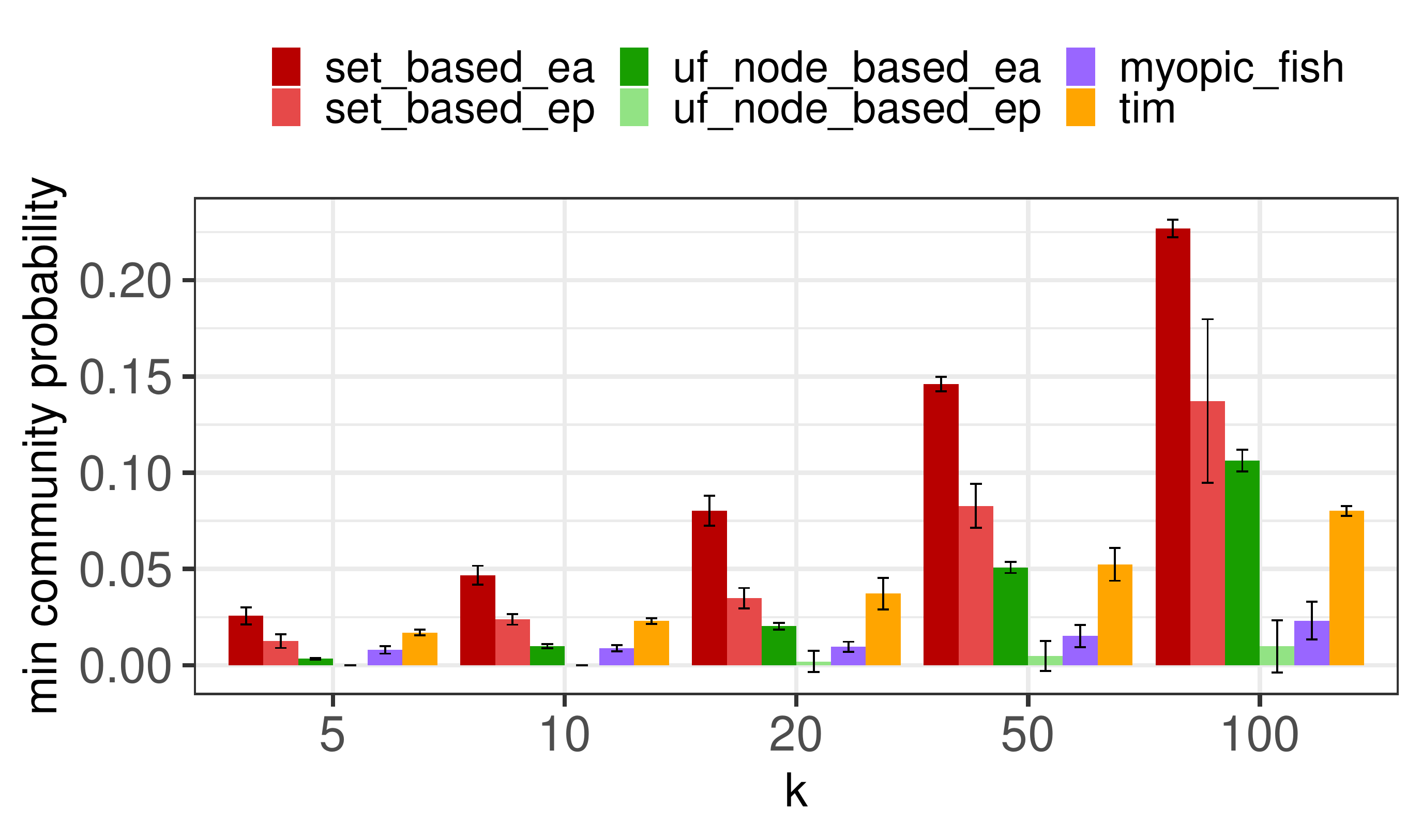}
    \includegraphics[trim={0 0 0 3.5cm}, clip, width=.49\linewidth]{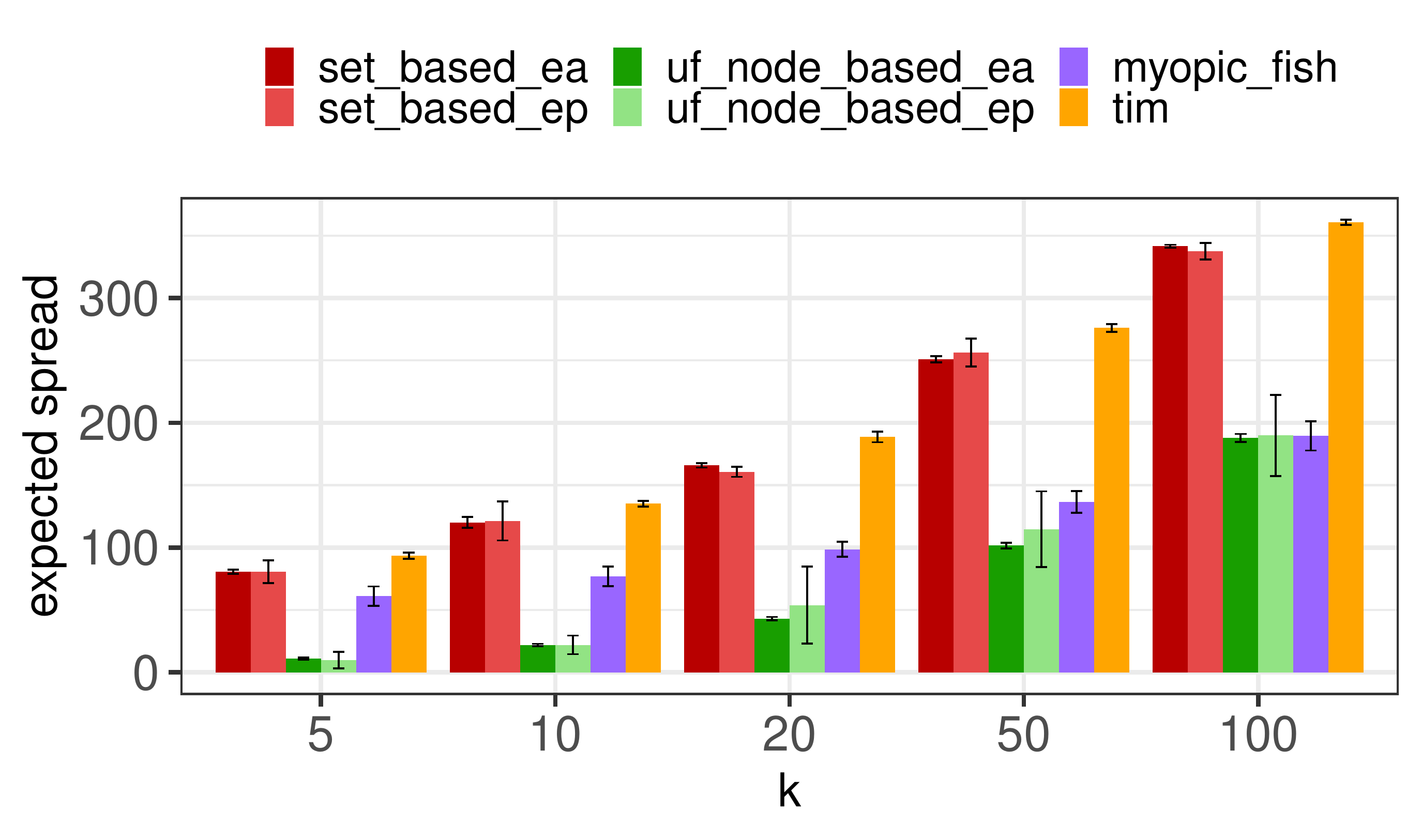}

    \caption{Results for the \texttt{Arenas} network for increasing $k = 5, 10 , 20, 50, 100$. The minimum community probability is shown on the left, while expected spread is shown on the right. BFS community structure with 
    (1) $10$ communities,
    (2) $n/10$ communities.}%
    \label{fig: arenas}%
\end{figure}

\begin{figure}[ht]
    \centering
    \includegraphics[trim={0 1.5cm 0 0}, clip, width=.49\linewidth]{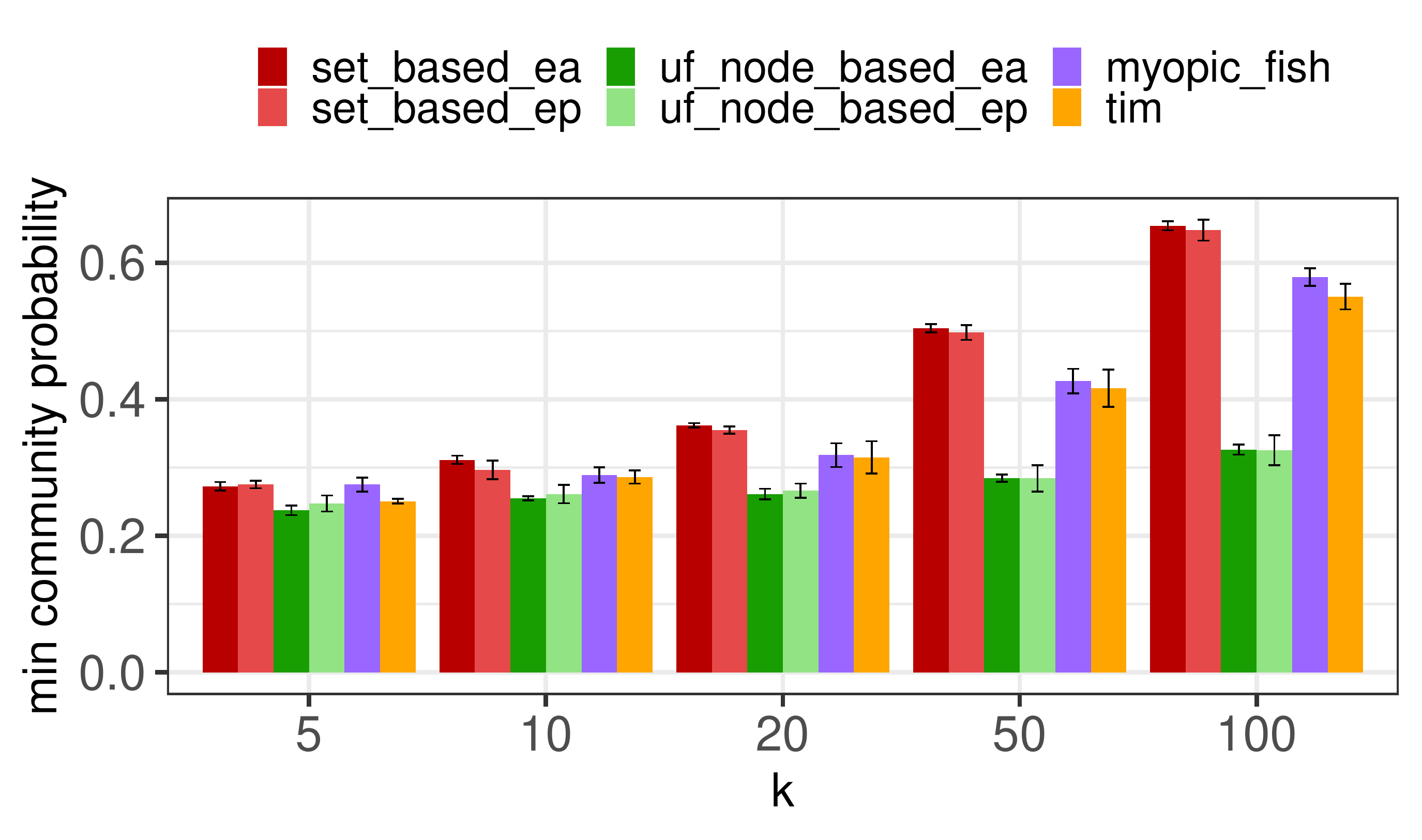}
    \includegraphics[trim={0 1.5cm 0 0}, clip, width=.49\linewidth]{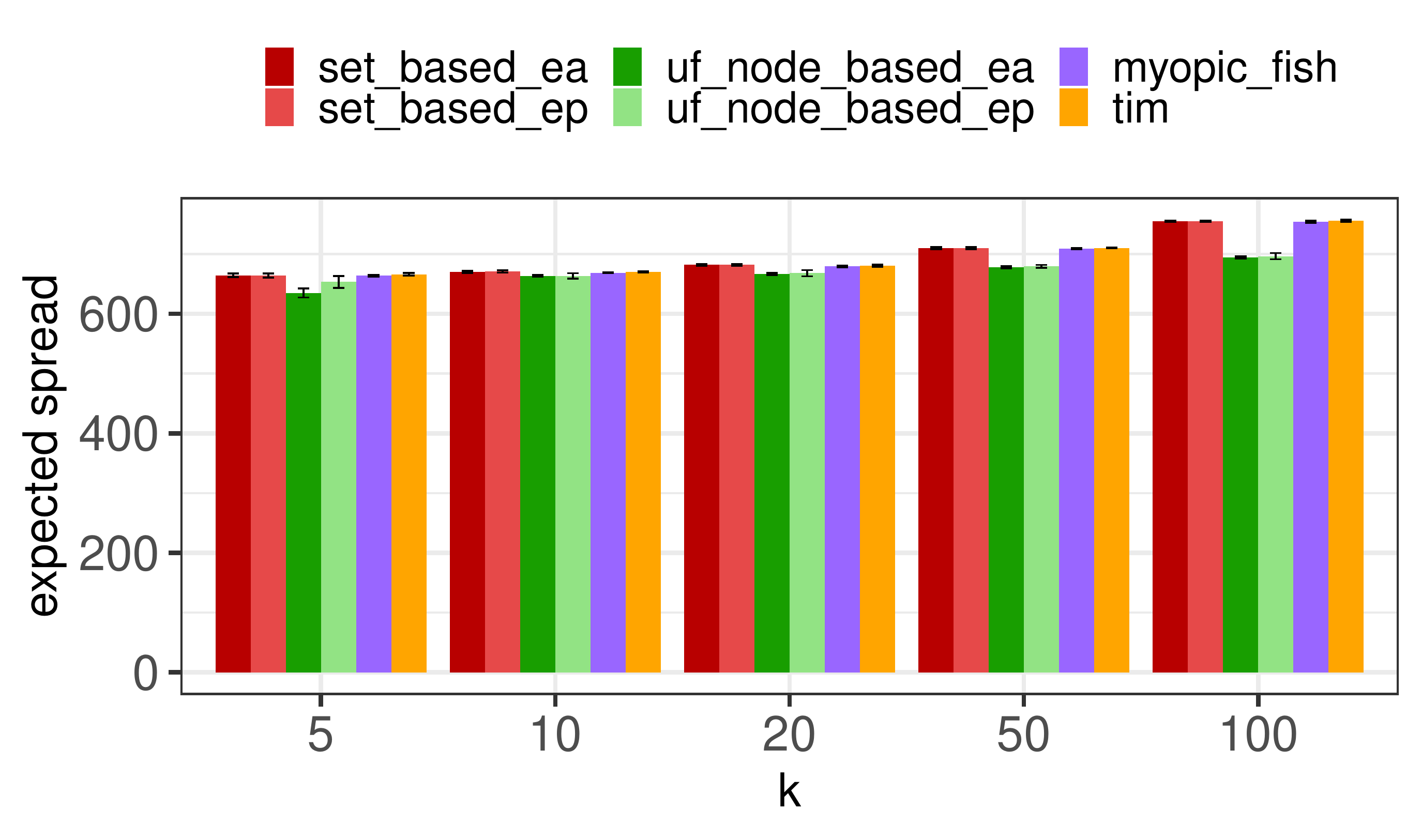} 
    
    \includegraphics[trim={0 1.5cm 0 3.5cm}, clip, width=.49\linewidth]{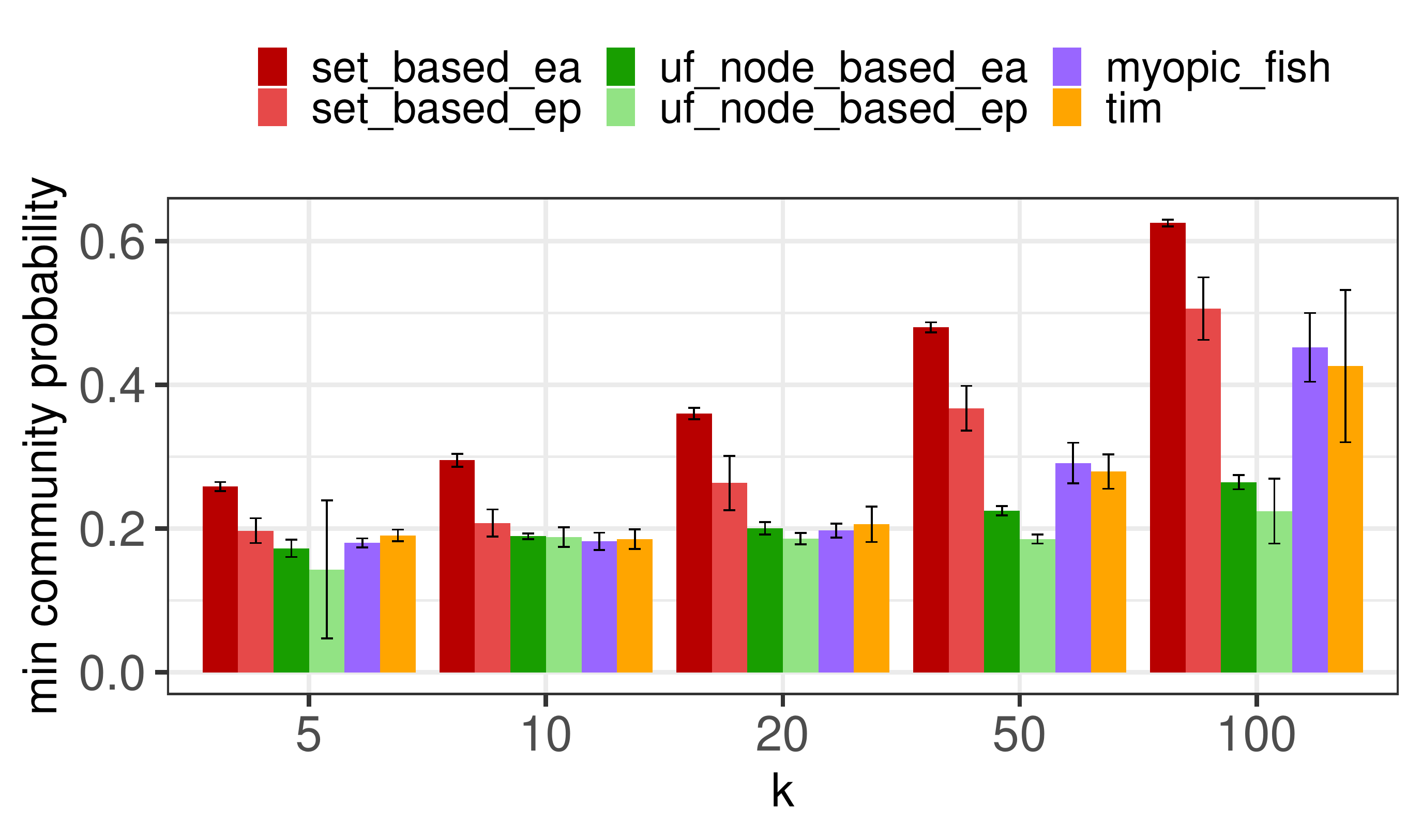}
    \includegraphics[trim={0 1.5cm 0 3.5cm}, clip, width=.49\linewidth]{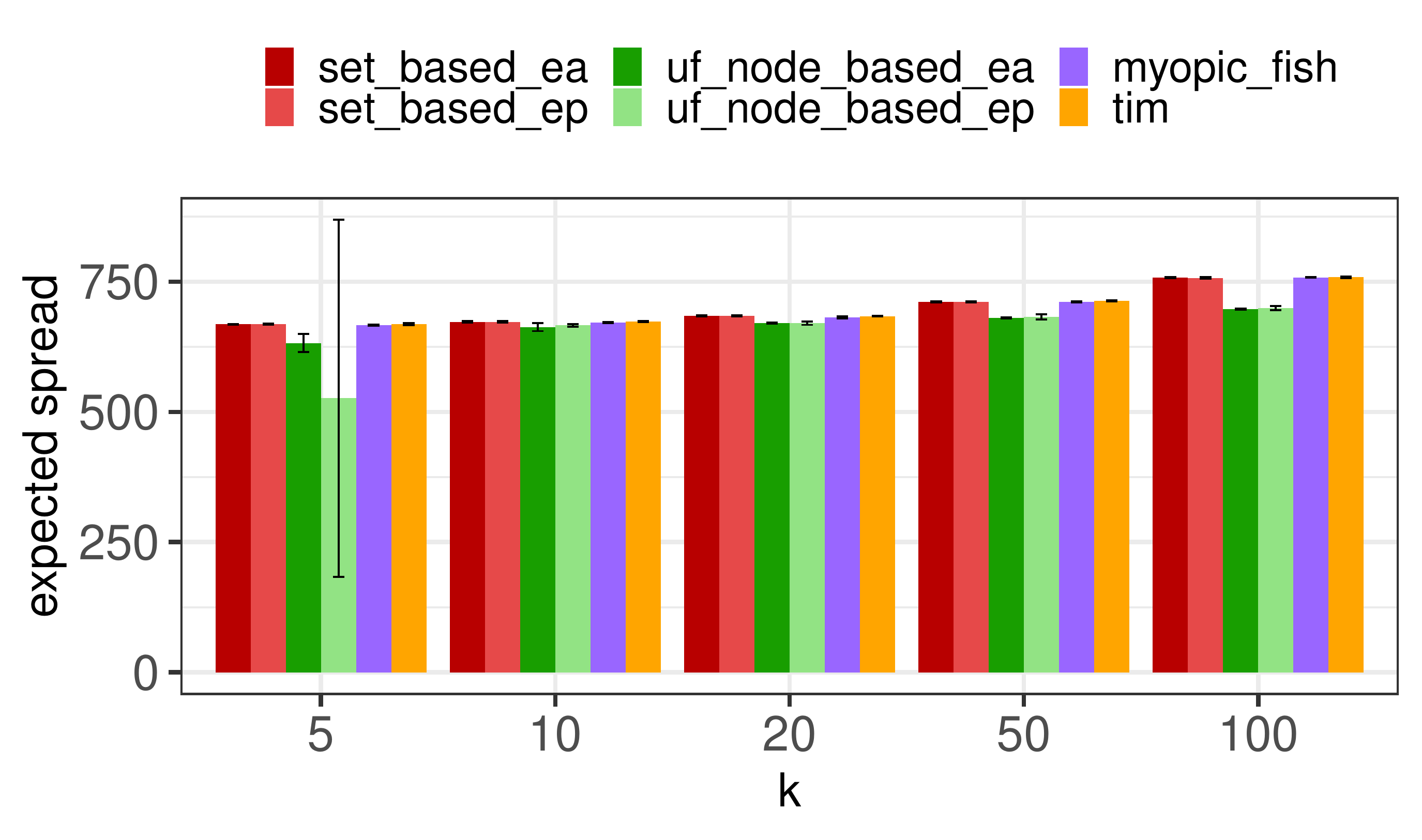}
    
    \includegraphics[trim={0 0 0 3.5cm}, clip, width=.49\linewidth]{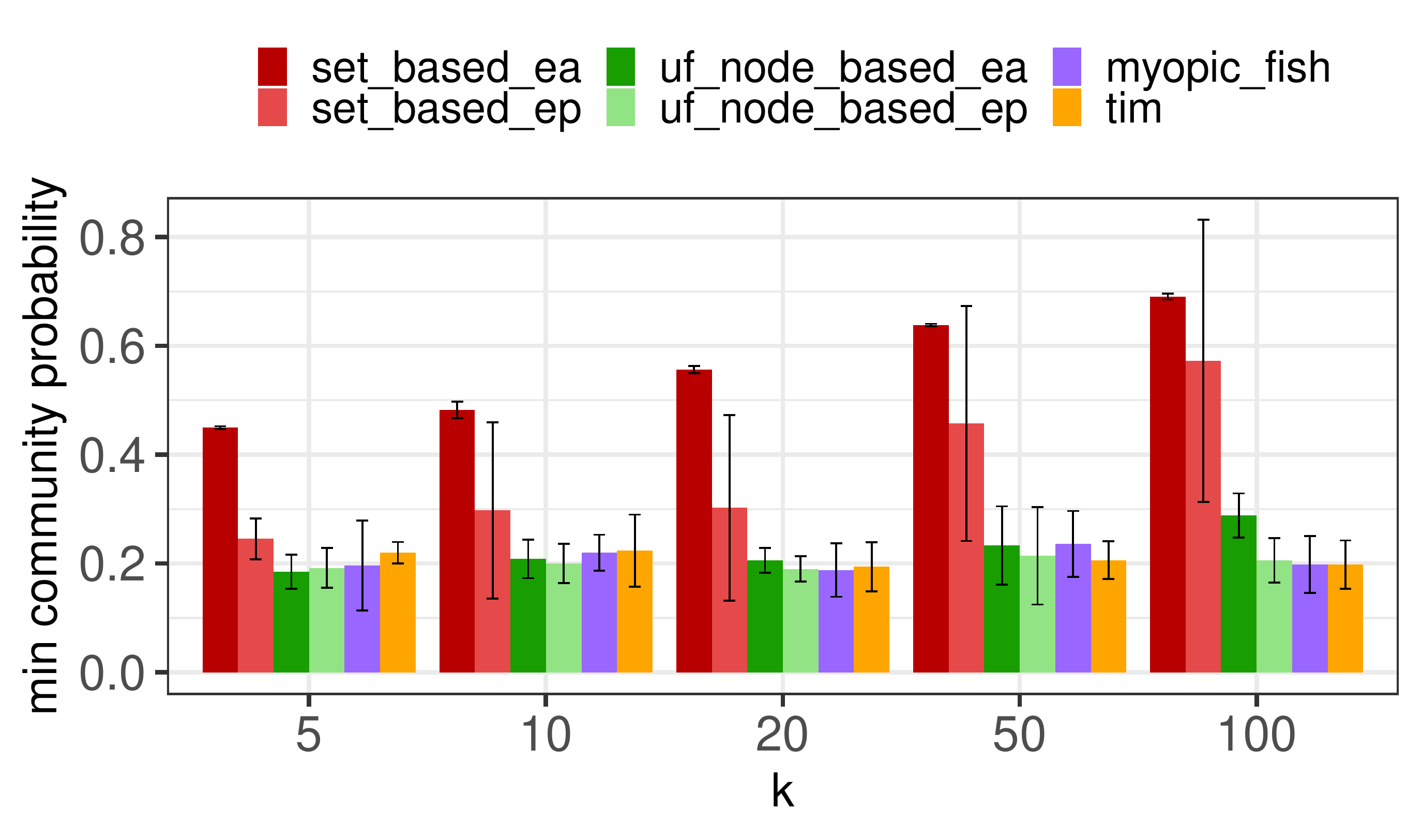}
    \includegraphics[trim={0 0 0 3.5cm}, clip, width=.49\linewidth]{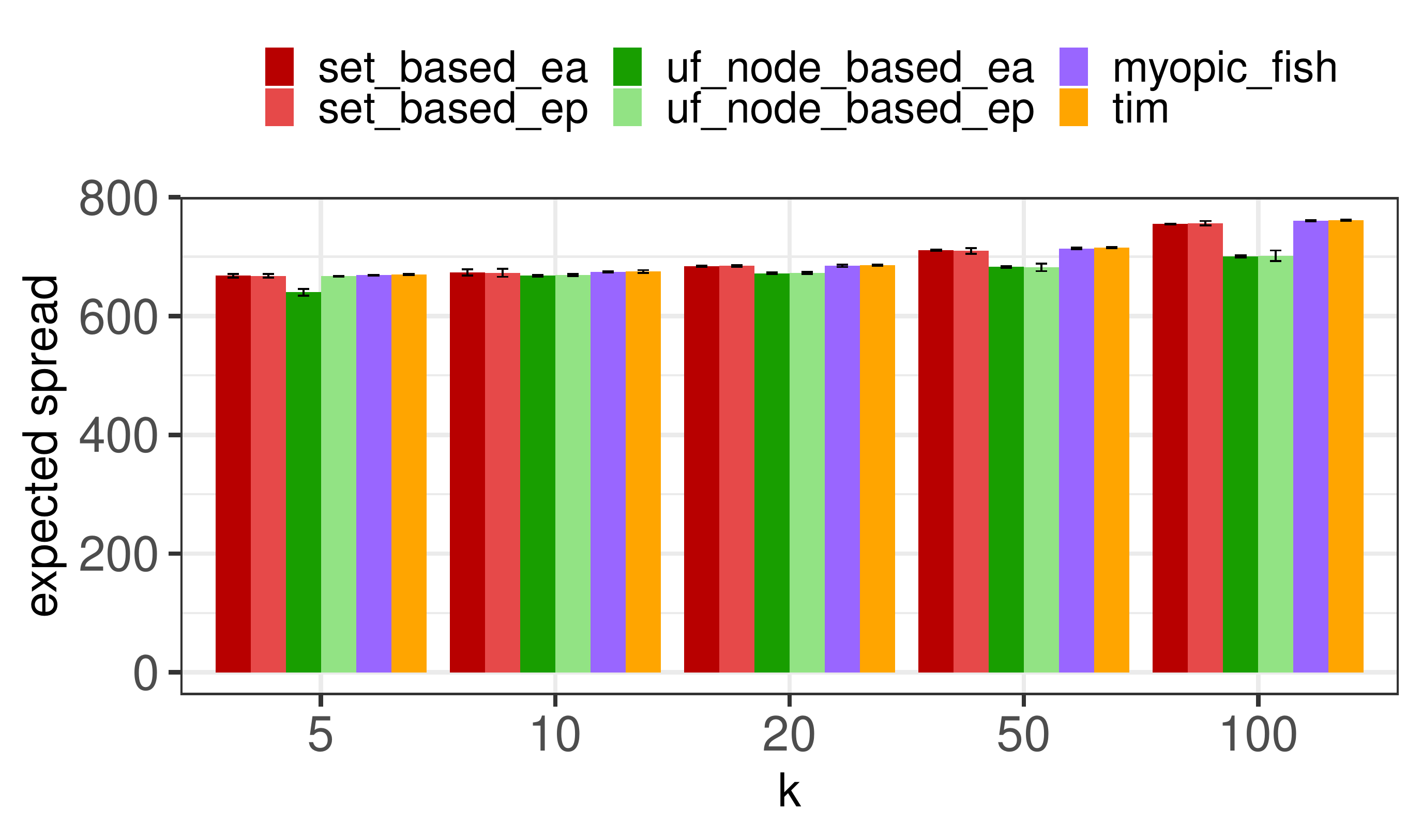}
    
    \caption{Results for \texttt{email-Eu-core} network for increasing $k = 5, 10 , 20, 50, 100$. The minimum community probability is shown on the left, while expected spread is shown on the right.
    (1) BFS community structure with $10$ communities, 
    (2) BFS community structure with $n/10$ communities, 
    (3) community structure induced by departments.}%
    \label{fig: eu-core}%
\end{figure}

\subparagraph{Evaluation on Networks used by Fish et al.}
We start with the networks that were also used in the study of Fish et al.~\cite{fish2019gaps}. We experiment with different community structures, both the BFS community structures with different community sizes and the random imbalanced community structure. The results can be found in Figure~\ref{fig: fish}. We omit the results for the BFS community structure with only 2 communities as they are very similar to the case of $10$ communities. We observe that in all cases the ex-ante value of our algorithm is dominating over all other values. In some cases, the values achieved by TIM are comparable but these are instances where all algorithms perform very close to each other and the minimum community probabilities are rather high anyways. Furthermore, on several instances, e.g., in the case of 10 BFS communities, the ex-post values of our algorithm are significantly better than the ex-post values of all other methods.

\subparagraph{Fairness vs.\ Efficiency on Email-Networks.}
We proceed by focusing on the email networks \texttt{Arenas} and \texttt{email-Eu-core} and comparing the fairness achieved by the different algorithms with the efficiency, i.e., expected spread. We again use the BFS community structures with different community sizes.
In the case of the \texttt{email-Eu-core} dataset, we evaluate the different algorithms on the community structure induced by the departments as well. We evaluate the algorithms for increasing values of $k = 5, 10 , 20, 50, 100$. We observe that \texttt{set\_based} performs best in terms of fairness among all ex-ante as well as ex-post values both on the \texttt{Arenas} dataset as well as on the \texttt{email-Eu-core} dataset. In terms of efficiency, we observe that on the \texttt{Arenas} dataset, \texttt{tim} and \texttt{set\_based} perform similarly good and much better than the other competitors. For the \texttt{email-Eu-core} dataset, we see that in terms of efficiency all algorithms (even \texttt{uf\_node\_based} for small values of $k$) perform almost identical.




\begin{figure}[htp]
    \centering
    \includegraphics[trim={0 1.5cm 0 0}, clip, width=.49\linewidth]{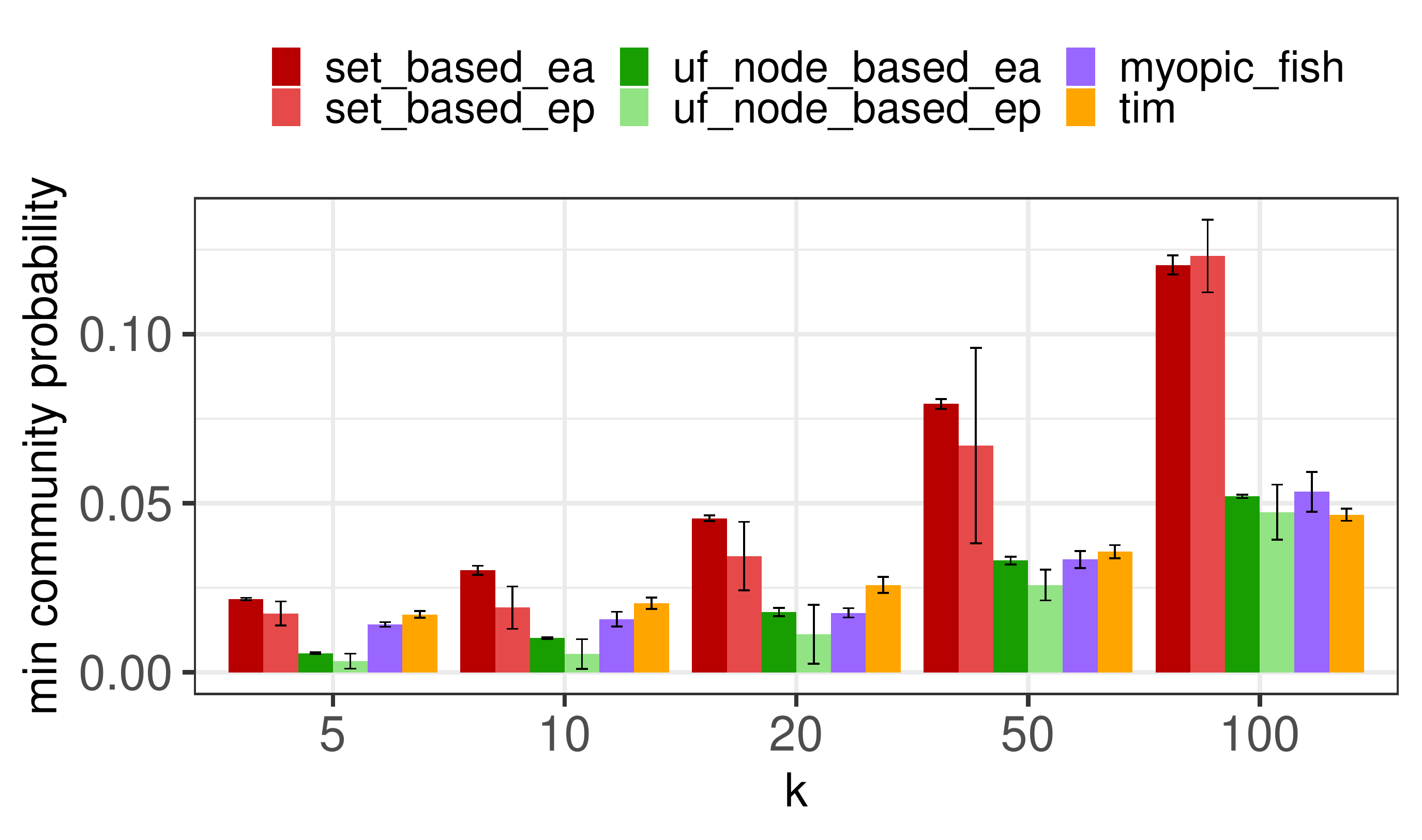}
    \includegraphics[trim={0 1.5cm 0 0}, clip, width=.49\linewidth]{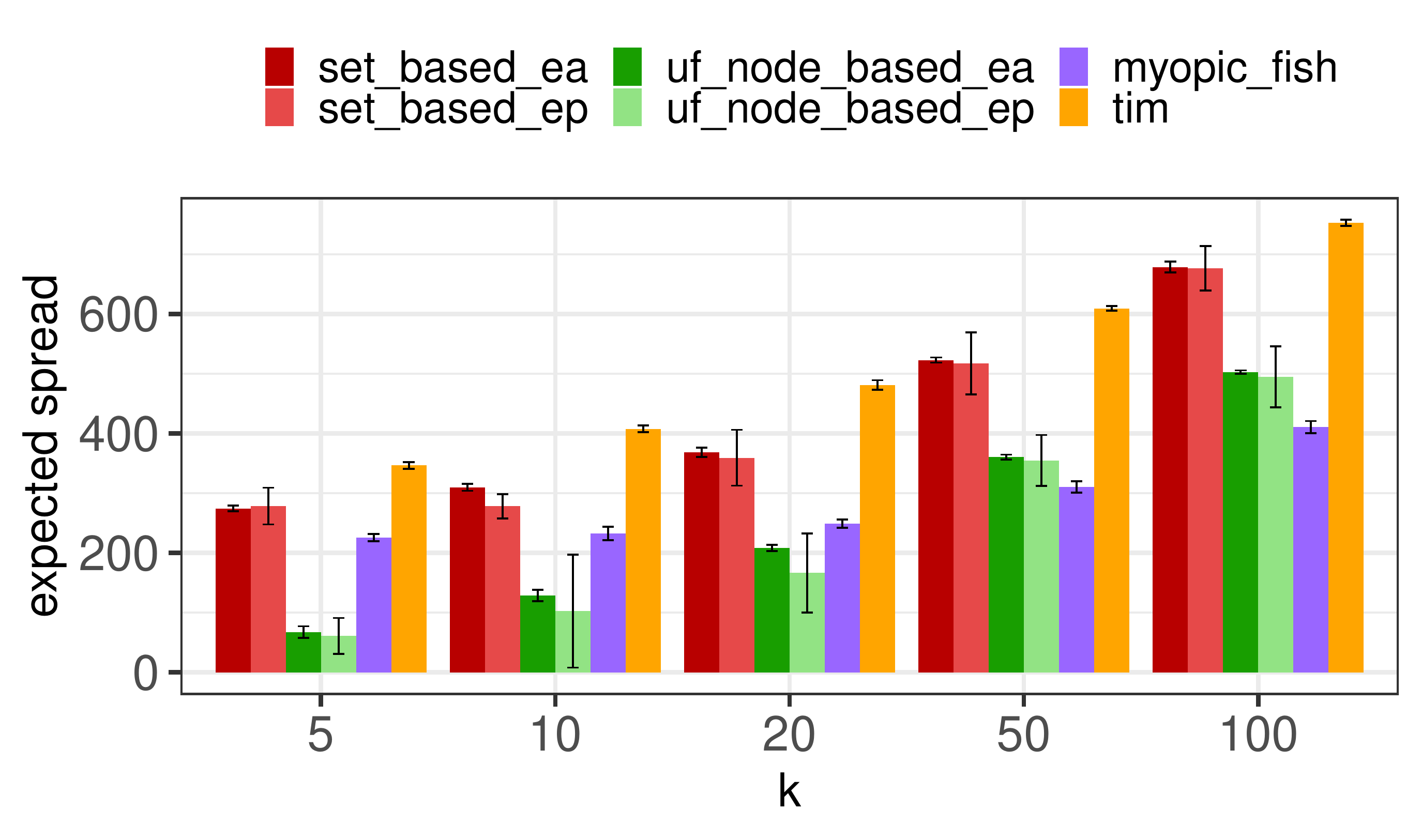}

    \includegraphics[trim={0 1.5cm 0 3.5cm}, clip, width=.49\linewidth]{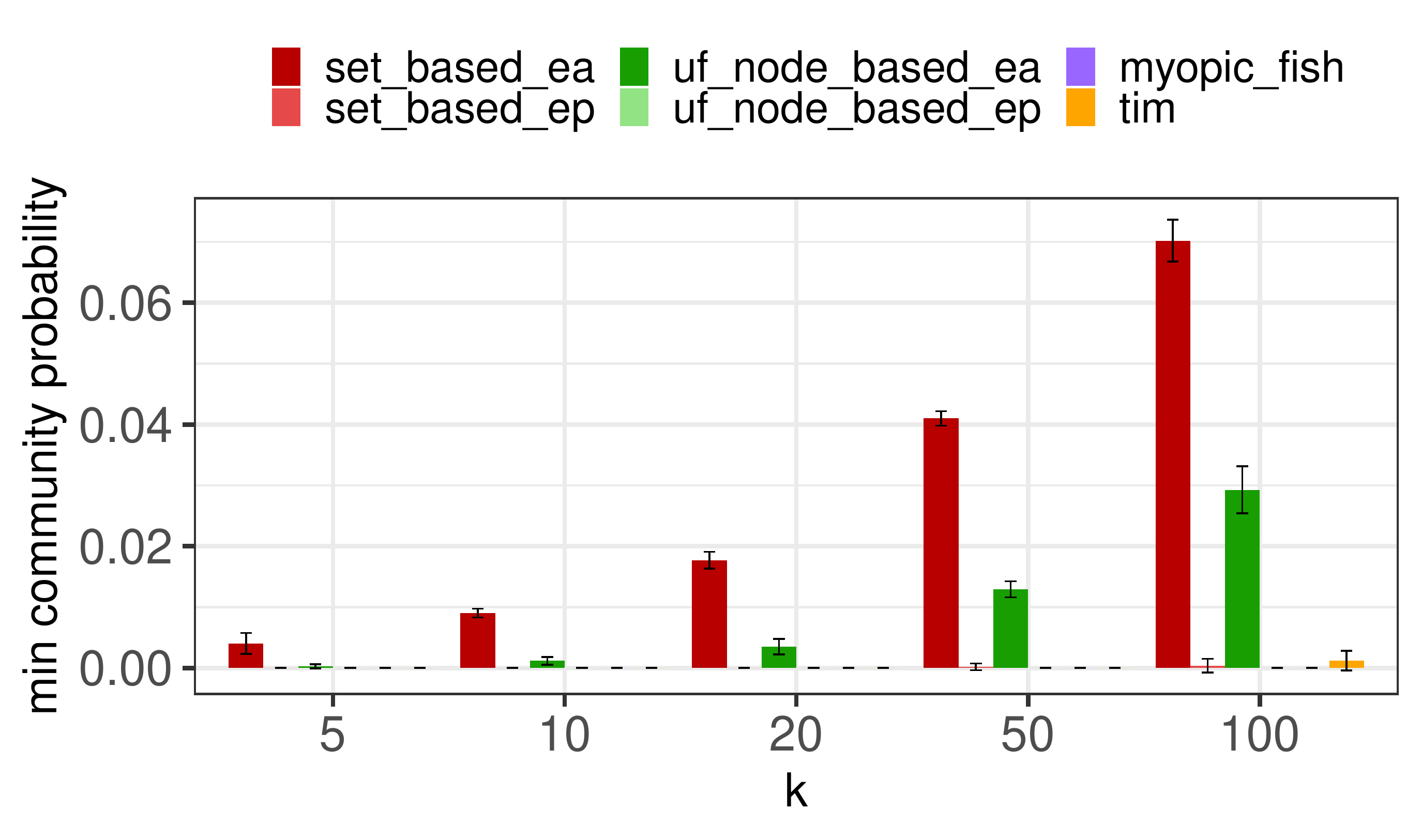}
    \includegraphics[trim={0 1.5cm 0 3.5cm}, clip, width=.49\linewidth]{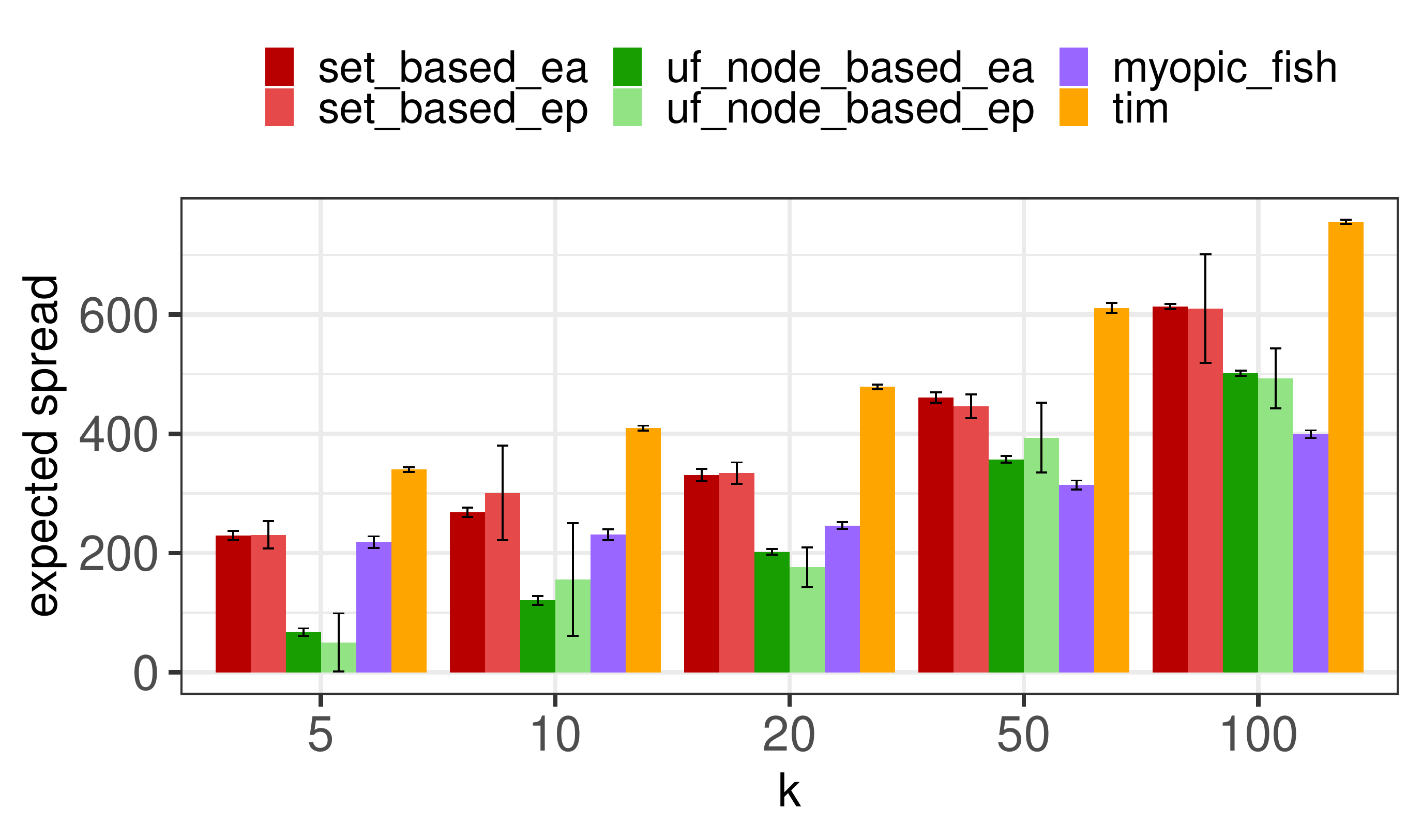}
    
    \includegraphics[trim={0 1.5cm 0 3.5cm}, clip,width=.49\linewidth]{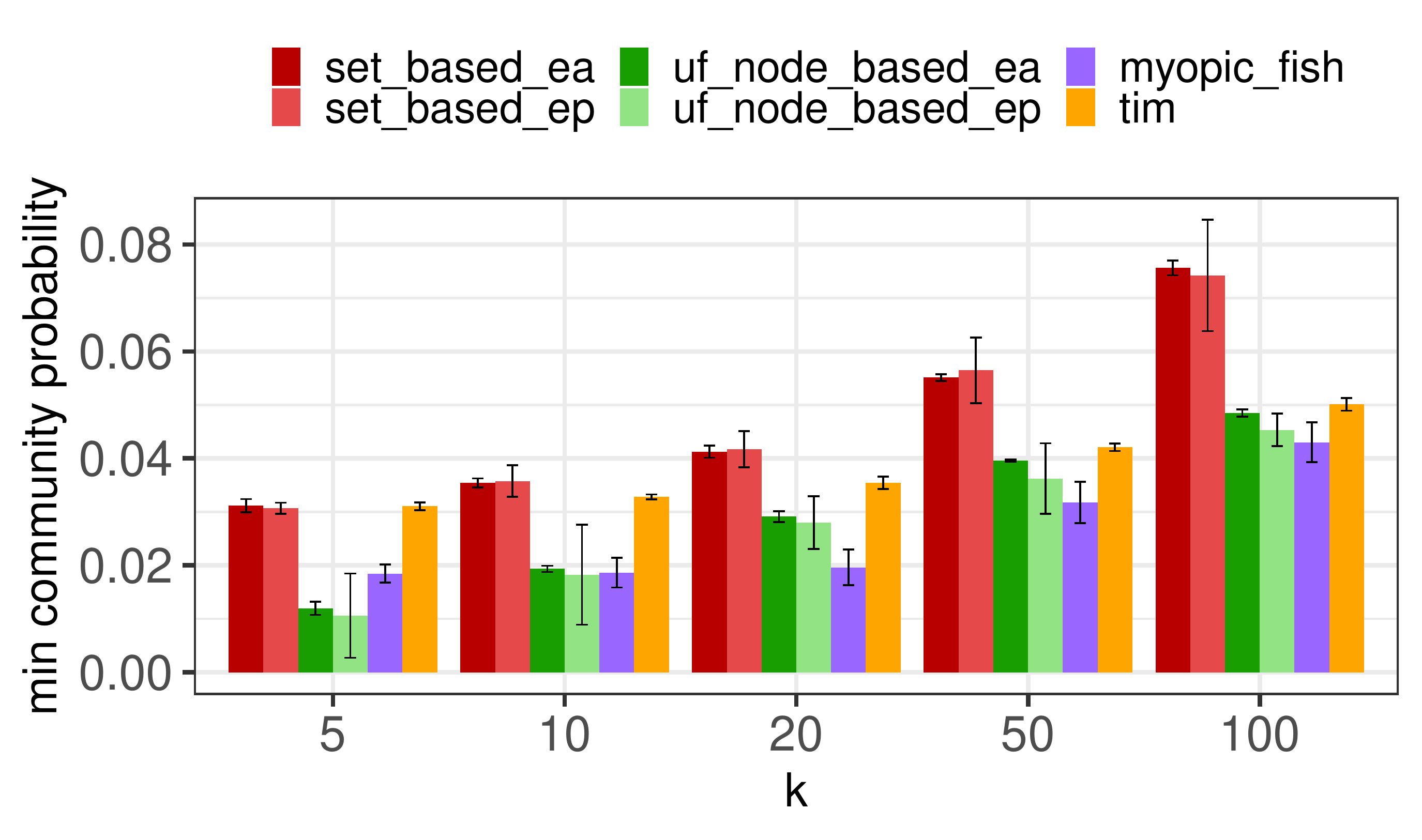}
    \includegraphics[trim={0 1.5cm 0 3.5cm}, clip,width=.49\linewidth]{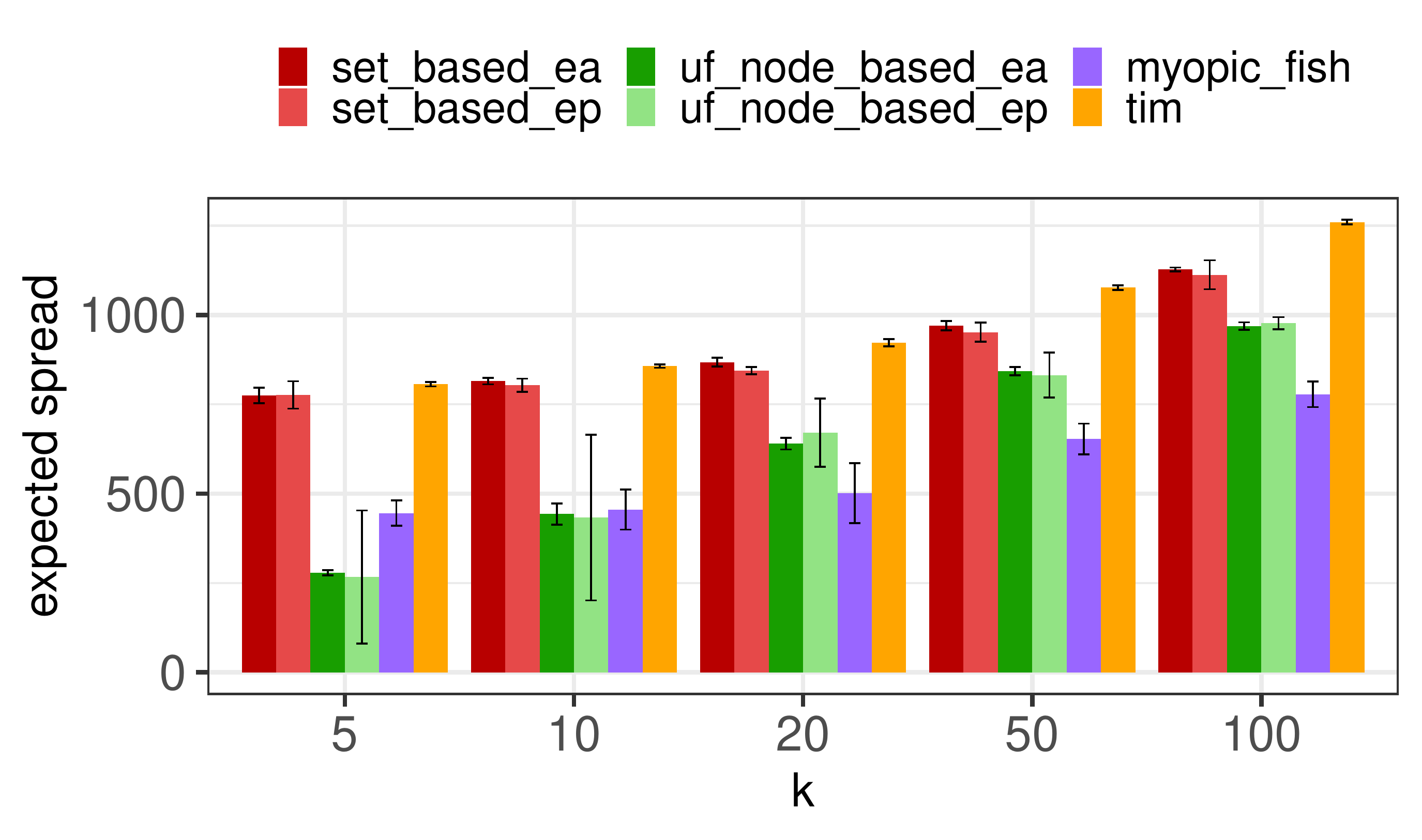}

    \includegraphics[trim={0 0 0 3.5cm}, clip,width=.49\linewidth]{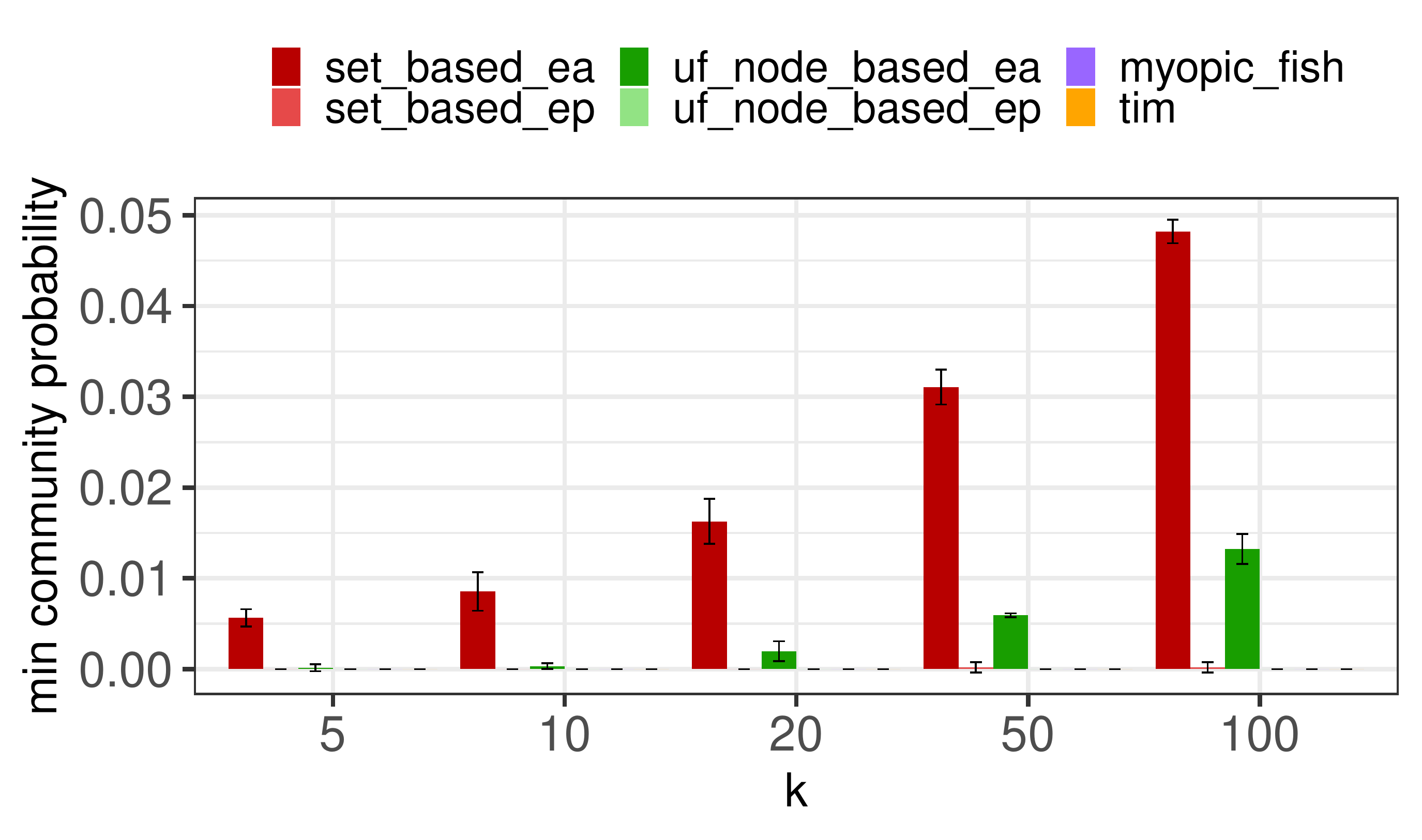}
    \includegraphics[trim={0 0 0 3.5cm}, clip,width=.49\linewidth]{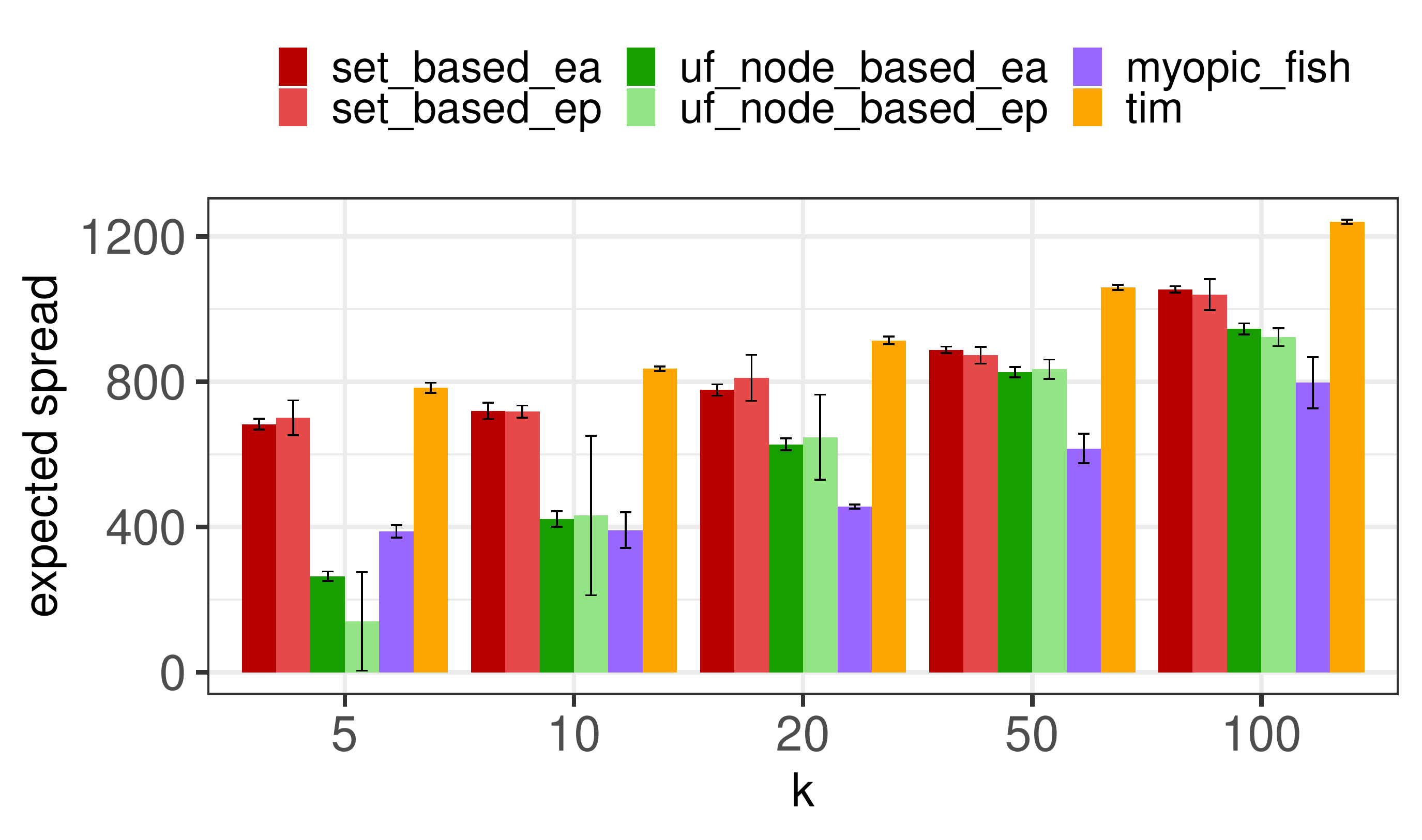}

    \caption{Results for co-authorship networks for increasing $k= 5, 10 , 20, 50, 100$. The minimum community probability is shown on the left, while expected spread is shown on the right. BFS community structure
    (1) \texttt{ca-GrQc} with $10$ communities,
    (2) \texttt{ca-GrQc} with $n/10$ communities,
    (3) \texttt{ca-HepTh} with $10$ communities
    (4) \texttt{ca-HepTh} with $n/10$ communities.}
    \label{fig: co-author}
\end{figure}

\subparagraph{Fairness vs.\ Efficiency on Co-Authorship Networks.}
We turn to the two co-authorship datasets \texttt{ca-GrQc} and \texttt{ca-HepTh} and evaluate the fairness and efficiency achieved by the different algorithms. Focusing on the fairness values first, we observe that in a setting with $n/10$ communities, no algorithm achieves a positive ex-post value. Instead the two randomized algorithms \texttt{set\_based} and \texttt{uf\_node\_based} do achieve a significantly non-zero ex-ante value, the results of \texttt{set\_based} being more than twice as high compared to the values of \texttt{uf\_node\_based}. For $10$ communities, we again end up in a setting where both the ex-ante and ex-post values of \texttt{set\_based} dominate over all other algorithms. The discrepancy between the ex-post value achieved by \texttt{set\_based} and the other algorithms appears to become more and more pronounced with increasing values of $k$. In terms of efficiency, we observe that again \texttt{tim} and \texttt{set\_based} perform the best, while there is a bigger advantage for \texttt{tim} in this case than with most other instances tested. Note however that \texttt{set\_based} does achieve significantly better fairness values as compared to \texttt{tim} in all settings where it falls behind \texttt{tim} in terms of efficiency.

\subparagraph{Fairness vs.\ Efficiency on \texttt{com-Youtube} Network.}
We conclude with the \texttt{com-Youtube} network and evaluate the different algorithms in terms of fairness and efficiency. In this network we choose edge weights uniformly at random in the interval $[0, 0.1]$.
In the left plot, we observe that the ex-ante fairness values achieved by \texttt{set\_based} are significantly better than the values of all other algorithms, especially with increasing values of $k$. The ex-post values of all algorithms are much smaller and close to each other. In terms of efficiency, we observe that the results of all algorithm are very similar. Note that \texttt{set\_based} performs almost the same as \texttt{tim} in the expected spread while being significantly better than \texttt{tim} in terms of fairness.

\begin{figure}[H]
    \centering
    \includegraphics[width=.49\linewidth]{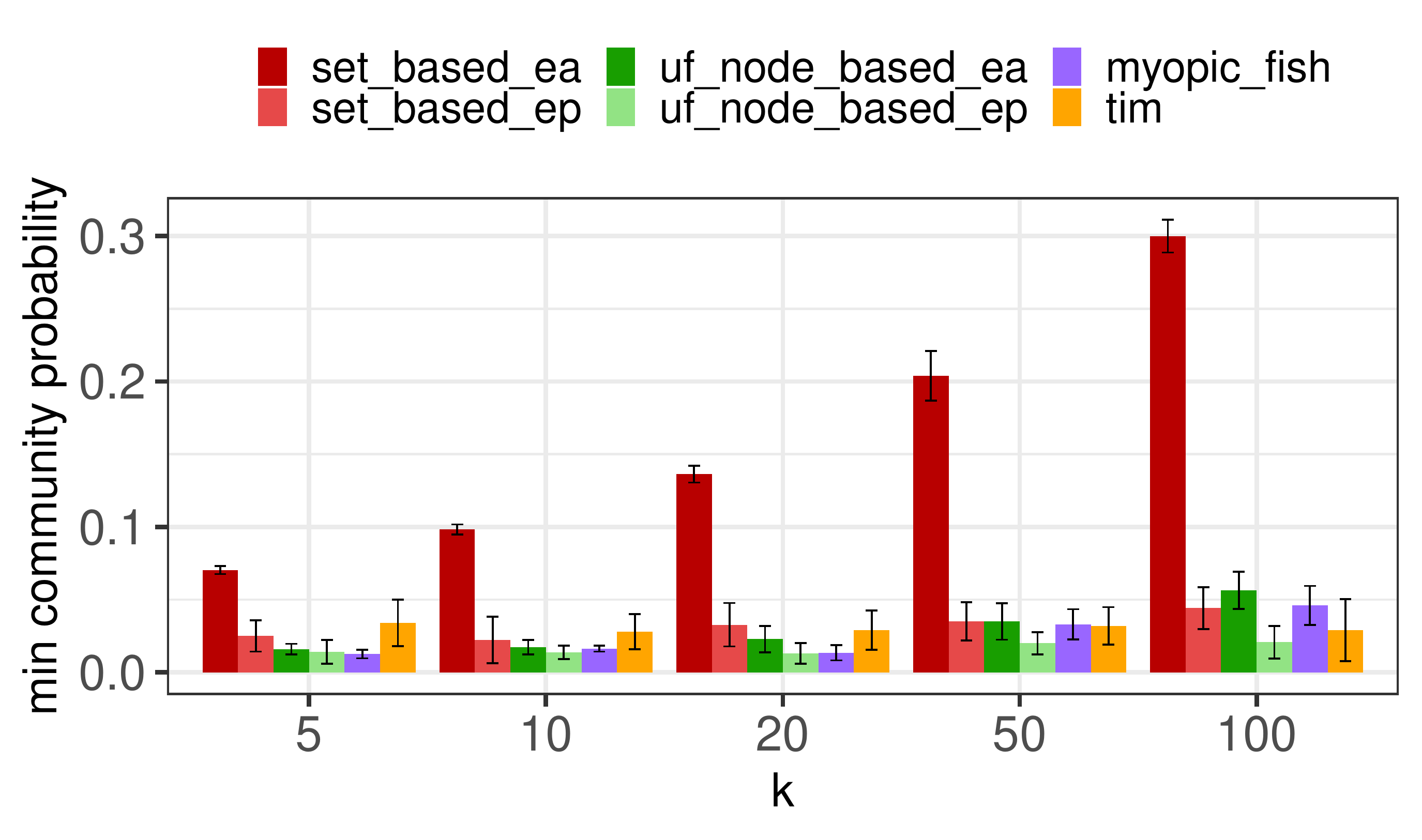}
    \includegraphics[width=.49\linewidth]{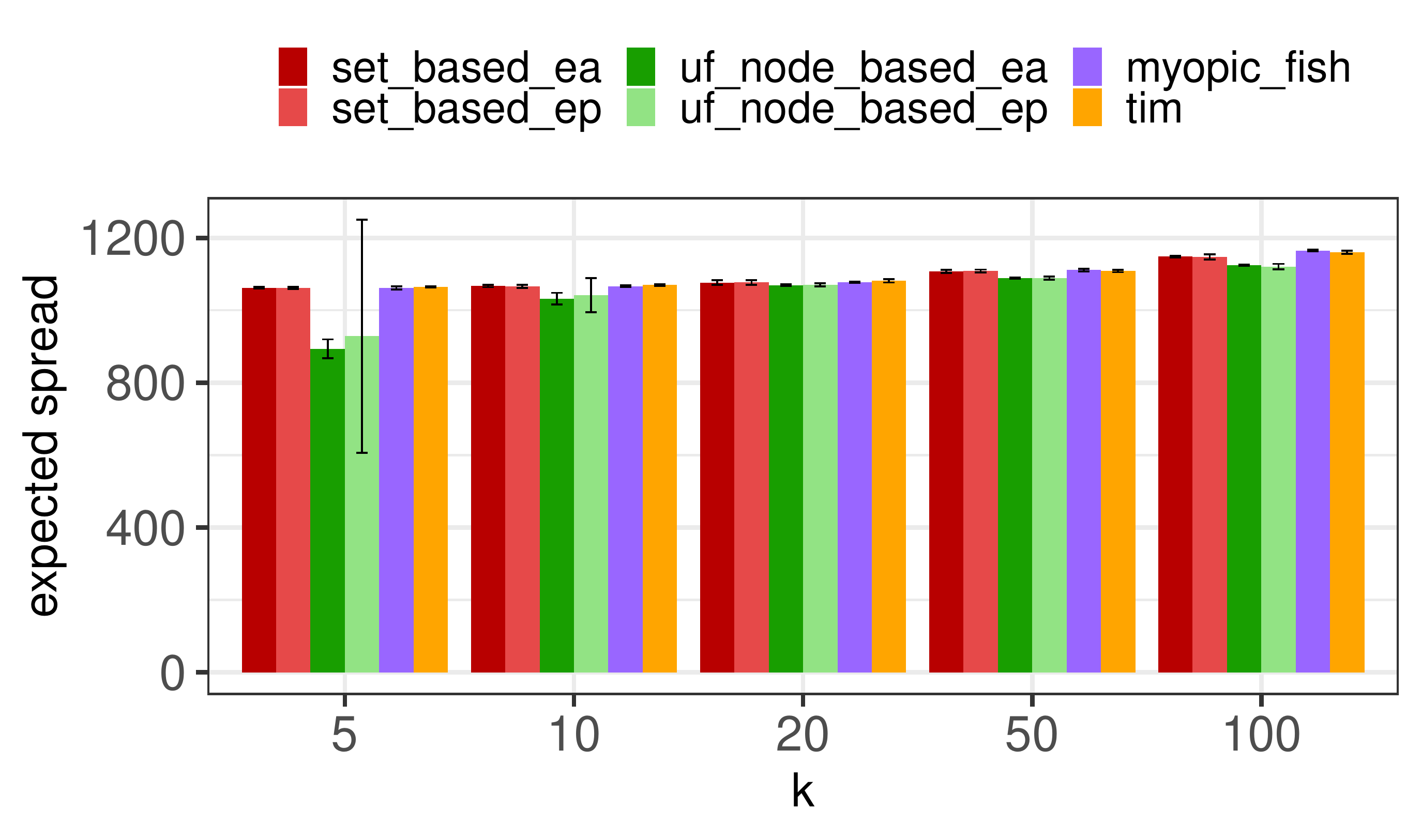}

    \caption{Results for the \texttt{com-Youtube} network for increasing $k = 5, 10 , 20, 50, 100$. The minimum community probability is shown on the left, while expected spread is shown on the right.}%
    \label{fig: com-youtube}%
\end{figure}

%% file: conclusion.tex
\section{Conclusion}
In this paper, we studied the problem of determining key seed nodes for influence maximization in social networks in an efficient \emph{and} fair manner. Notably, we have designed approximation algorithms achieving a constant multiplicative factor of $1-1/e$ (plus an additive arbitrarily small error) for the objective of maximizing the maximin influence received by a community. For one of the two variants of our problem, we have shown that this approximation factor is best possible (up to the arbitrarily small error term).
We achieved our algorithmic result by using randomized strategies, thus enlarging the solution set and enabling us to find fairer solutions ex-ante. Our detailed experimental study confirms the increase in ex-ante fairness achieved over previous methods~\cite{fish2019gaps,tsang2019group}, indicating that randomness as source of fairness in influence maximization is very promissing to be further explored. Interestingly, our study shows that even the ex-post fairness achieved by our methods frequently outperforms the fairness achieved by other tested methods. While our theoretical results predict that the price of fairness in the studied setting can be very bad, our experimental evaluation instead indicates that the loss in efficiency when using our methods is very limited.

Several directions are conceivable as future work. Improving our approximation guarantees for the set-based problem or providing a matching approximation hardness result seems a challenging direction of exploration. 
Tightening the result on the gap between node-based and set-based problem is certainly an interesting question. 
Further engineering the implementations of our methods 
could yield further gains in efficiency, particularly in terms of runtime. 
Lastly, we believe that the idea of using randomization to increase the fairness of solutions for influence maximization may be used for other fairness criteria as, e.g., the group rational criterion of Tsang et al.~\cite{tsang2019group}.